\documentclass[twoside,11pt]{article}
\pdfoutput=1
\usepackage{algorithm}
\usepackage{algorithmic}
\usepackage{microtype}
\usepackage{graphicx}
\usepackage{booktabs} 
\usepackage{amsmath}
\usepackage{mathtools}
\usepackage{hyperref}
\usepackage{amsfonts}
\usepackage{authblk}
\usepackage{color}
\usepackage{caption}
\usepackage{subcaption}
\usepackage{bm}
\usepackage{bbm}
\usepackage[margin=1in]{geometry}

\usepackage[style=authoryear,natbib, maxcitenames=2, uniquename=false]{biblatex}
\newtheorem{theorem}{Theorem} 
\newtheorem{lemma}{Lemma}
\newtheorem{proof}{Proof}
\newtheorem{corollary}{Corollary}

\usepackage[title]{appendix}

\begin{document}

\title{Online Estimation and Community Detection
of Network Point Processes for Event Streams}
\date{\today}


\author[1]{Guanhua Fang\thanks{These authors contributed equally.}}
\author[2]{Owen G. Ward*}
\author[3]{Tian Zheng}
\affil[1]{Department of Statistics and Data Science, School of Management, Fudan University}
\affil[2]{Department of Statistics and Actuarial Science, Simon Fraser University}
\affil[3]{Department of Statistics, Columbia University}
\date{October 2023}

\maketitle

\begin{abstract}
A common goal in network modeling is to uncover the latent community structure
present among nodes. For many real-world networks, 
the true
connections consist of events arriving as streams, which are then aggregated 
to form edges, ignoring the dynamic temporal component. 
A natural way to take account of these temporal dynamics of  
interactions is to use point processes as the foundation of network models
for community detection. 
Computational complexity hampers the 
scalability of such approaches
to large sparse networks. 
To circumvent this challenge, we propose a fast online variational
inference algorithm for
estimating the latent structure
underlying dynamic event arrivals on a network,
using continuous-time point process latent network models.
We describe this procedure for network models capturing community structure.
This structure can be learned as new events are observed on the network,
updating the inferred community assignments.
We investigate the theoretical properties of such an inference scheme,
and provide regret bounds on the loss function of this procedure.
The proposed inference procedure
is then thoroughly compared, using both
simulation studies and real data, to non-online variants. 
We demonstrate that online inference can obtain comparable 
performance, in terms of community recovery, to non-online
variants, while realising computational gains.
Our proposed inference
framework can also be readily modified to incorporate other popular 
network structures.

\end{abstract}


\section{Introduction}

\label{Introduction}
Network models are widely used to capture the structure in large complex data.
One common goal of many statistical network models is
\textit{community detection} \citep{zhao2012consistency, amini2013pseudo}, 
which aims to
uncover latent clusters of nodes in a network based on observed 
relationships between these 
nodes \citep{fortunato2016community}.
However, many of these models assume that the edges,
describing the relationship between these nodes,
are simple, i.e., with 
interactions between nodes described by binary edges
or weighted edges of counts. In 
reality, for many real networks,
activities between nodes occur as streams of interaction 
events
which may evolve over time and exhibit non-stationary patterns.
For example, social network data is commonly aggregated into 
binary edges describing 
whether there is a connection between two actors, when in reality the true 
underlying 
data could have consisted of multiple messages or other interactions
over a period of time. The binary 
edge might be constructed by considering if the number of such 
interactions is
above an 
arbitrary cut-off.
Aggregating these event streams and ignoring the time component 
to these interactions leads to an obvious loss of information.
Models which take advantage 
of the temporal dynamics of 
event streams
therefore
hold the potential to reveal richer latent structures behind these dynamic 
interactions
\citep{matias_semiparametric_2018}.



{\color{black}
To illustrate the role the event times can play in community detection, 
we simulate 50 replications of a small dense 
network of $n=100$ nodes with $K=2$ communities from the block inhomogeneous
Poisson process model described in Section~\ref{Background},
with the underlying intensity being
a simple step function which we include in Appendix~\ref{appendix:add_sims}.
Here the correct 
community structure 
is clear if each individual point process between a node pair $(i,j)$, $\lambda_{ij}(t)$, is known.
If we instead treated this as a traditional network community detection problem, ignoring
the presence of the event times, the following two 
approaches could be considered:
\begin{itemize}
    \item Aggregate the event data to form a single adjacency matrix $A$. This could
    be a weighted adjacency matrix, with the weights corresponding to the number of events
    observed between each node pair, resulting in a count matrix. Spectral clustering 
    could then be applied to this count matrix to infer the community structure.
    We show the performance of this clustering scheme across repeated simulations,
    in terms of adjusted rand index 
    (ARI) in the first boxplot in Figure~\ref{fig:agg}.
    \item An alternative approach would instead construct a series of adjacency matrices
    $A_1,\ldots,A_{M}$, where each such matrix is formed by aggregating data in consecutive
    time windows of some fixed length. For example, here we construct an edge between two nodes if there is at least one event between them in that time window.
    \citet{pensky2019spectral} provide a procedure
    with strong theoretical guarantees for such a series of adjacency matrices.
    In Figure~\ref{fig:agg} we consider
    window lengths $(T/100,2T/100,\ldots, T/10)$. In each case we use the method of
    \citet{pensky2019spectral} to then estimate the community structure, which is shown
    in the second boxplot of Figure~\ref{fig:agg}.
\end{itemize}
We see that both aggregation methods are unable to estimate the community
structure. In particular,
the method of \citet{pensky2019spectral} fails regardless of how we aggregate the
data to form a sequence of adjacency matrices. Similarly, performing spectral clustering
on the count matrix of events does not identify the communities.
In contrast, if we instead apply the 
network point process model of this work, which incorporates the event times when
performing community detection, we can correctly recover the true communities, as seen in the third boxplot of
Figure~\ref{fig:agg}.
}

\begin{figure}[ht]
    \centering
    \includegraphics[width=0.65\textwidth]{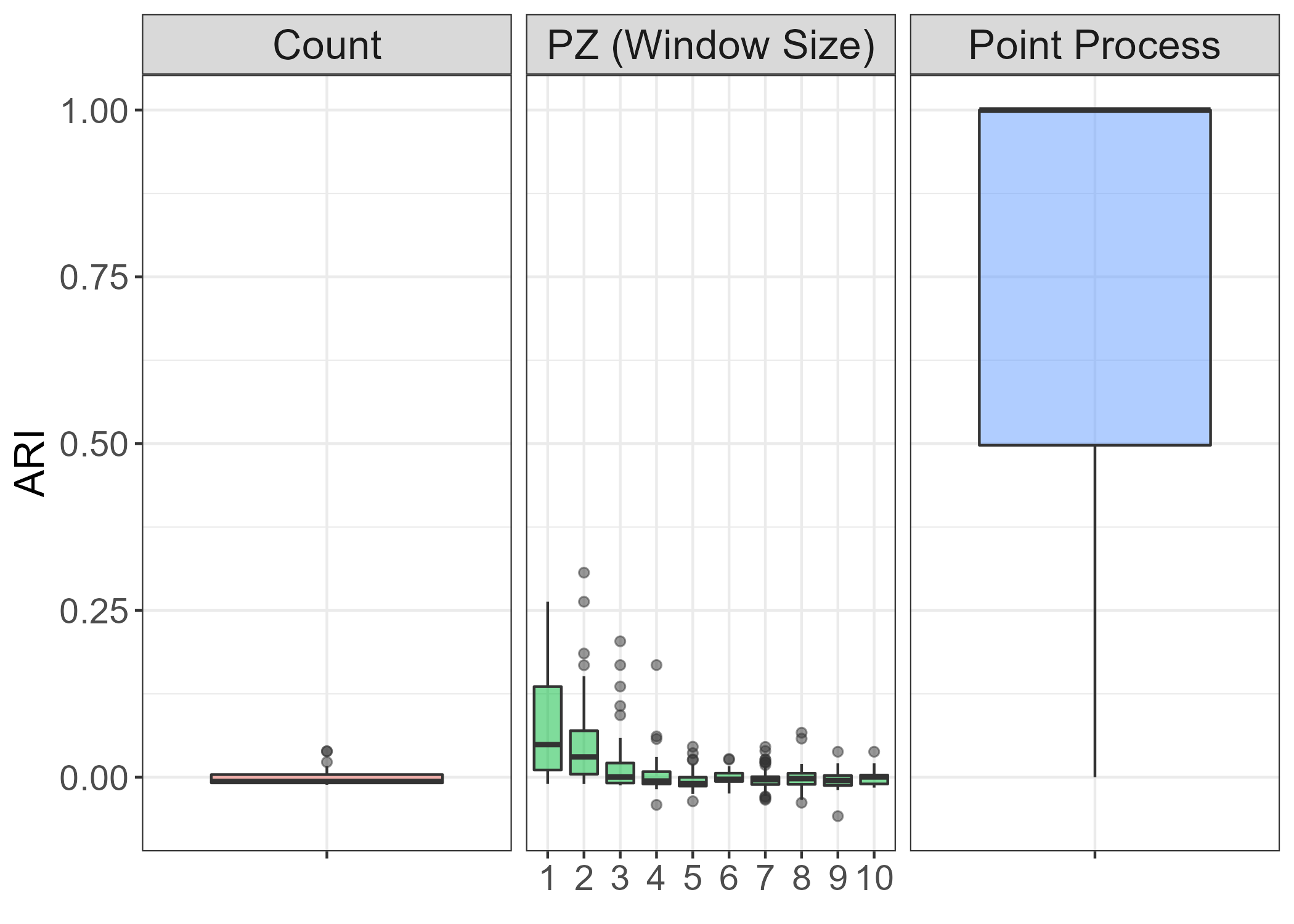}
    \caption{Community recovery in terms of
     Adjusted Rand Index (ARI) for events simulated from a point process
     block model. 
     Aggregate methods which look at the overall count between nodes,
     or bin the data (PZ)\citep{pensky2019spectral} cannot recover
     the community structure. This is the case regardless of the choice
     of window used to 
     aggregate the data to form the PZ estimator (as given by the multiple 
     box-plots).
     Modelling the exact event times 
     through a point process can recover the true community structure
     well.}
    \label{fig:agg}
  \end{figure}


Point processes are commonly used to model event streams, which can then 
be incorporated 
into network models to provide a community detection method
which accounts for the dynamics of these event streams on the network. 
Notably, these 
models are able to characterize 
sporadic and bursty dynamics, which are 
ubiquitous in event streams on networks. Network
models of this form have recently been
developed, uncovering more expressive community structure.
However, these methods suffer from the computational challenges 
associated with both 
network data and point process methods, 
and it is computationally difficult to scale them to large networks. 
Further, to truly account for the streaming nature of edges, 
we would like to be able to perform community detection as 
events are observed on the network, updating our 
model with the arrival of new data.
To do this, we propose an online variational inference framework and 
corresponding algorithms 
to learn the structure of these networks as interactions between
nodes arrive as event streams.

We derive theoretical
results for 
the proposed online variational algorithm.
These include a regret bound for the online
estimator,
along with convergence rates for parameter recovery 
and recovery of latent community assignments.
These results
demonstrate that our procedure is 
comparable to more expensive
non-online methods. 
We are not aware of comparable 
existing 
theoretical results in the 
context of online variational inference.
We then analyze the empirical performance of 
this algorithm and find that the proposed method 
performs well under various simulation settings,
in comparison to more computationally intensive methods 
which 
process the entire data set.
Finally, we compare our online estimation procedure with 
corresponding methods which process all data at once, and 
show that we obtain comparable results without 
repeated expensive computations over all events. We also 
discuss the potential to use online variational 
procedures of this form in other contexts
and for different types of network structure.

To the best of our knowledge, this is the first work on 
online estimation and 
online community detection for event streams on networks.
Existing counting process
models can be
readily
incorporated into our proposed framework. 
The computational issues
present in models of this form on networks
are resolved 
by introducing a new
online variational inference-based algorithm, which 
recursively updates the model parameters and nodes' 
latent memberships and has low memory cost.
Compared 
with the classical batch methods,
our algorithm is 
scalable with data size and can achieve similar 
prediction performance.
We also develop the first
corresponding theoretical results
in the context of online latent network 
models.
The performance of the proposed online method is 
guaranteed when the network structure is
sufficiently dense over time.

This paper is organized as follows. In 
Section~\ref{Background} we first
formally define the required notation for modeling event 
streams using 
point processes and consider existing work which posits block
type
models of point processes to model event streams on networks. 
We also review existing results for online variational inference.
In Section~\ref{Computation} we propose an online learning 
framework for
models of this form. We outline the main theoretical 
results for
this procedure in Section~\ref{Theory}. 
Section~\ref{Simulation} outlines
simulation studies comparing the performance of our procedure
to more expensive batch methods. In Section~\ref{Realdata} 
we implement our algorithm on multiple data sets of streaming
events 
on networks. Finally, in Section~\ref{Discussion}, we briefly
describe how this procedure could be modified and applied in 
other contexts,
demonstrating the usefulness of our developments more 
generally.

\section{Background}
\label{Background}
We first review the required framework of modeling event streaming data
using network point processes and describe previous work which
has been done to incorporate such structure into existing network models.
We then review existing work and results in online variational inference.

\subsection{Network Point Processes}
We wish to model pairwise directed interactions 
between $n$ nodes in a network over some time interval.
We observe events $\{(e_1,t_1), \ldots, (e_l, t_l)$, $\ldots, (e_L, t_L)\}$,
where $e_l$ is the $l$-th event 
and $t_l$ is its corresponding time stamp. We have 
$e_l \in \mathcal E$ for $l = 1, \ldots, 
L$ and $0 < t_1 < \ldots < t_L$,
where $\mathcal E$ is the set of all possible different event types. 
Specifically, for event data on a network, we have
$\mathcal{E} = \{(i,j) \in A ~ \vert ~ i,j \in [n]\} $ 
where $(i,j)$ represents a directed event from node $i$ to node $j$; $n$ is the 
size of the population, $[n] = \{1, \ldots, n\}$ and $A$ is the edge list, which encodes 
the network structure present. We use $\vert A \vert$ to denote the total number of 
node pairs which interact at least once in the network.
We can therefore equivalently represent these interactions as 
$$
\mathcal{D} = \{ (i_l, j_l, t_l): l=1,\ldots, L \},
$$
consisting of triples where $(i_l, j_l, t_l)$ denotes a directed interaction from 
the $i_l$-th node to the $j_l$-th node at time $t_l$.
Had only the event times been observed, without knowing 
the nodes involved in the interactions,
a natural way for modeling this type of streaming 
data is to use the machinery of counting processes. Under this framework,
$N(t)$ is used to denote the counting process, the number of events observed up to time 
$t$.
Along with this, the conditional intensity function is defined as 
\begin{eqnarray}
\lambda(t) = \lim_{dt \rightarrow 0} 
\frac{\mathbb E \left(N[t, t + dt) \vert \mathcal 
H(t)\right )}{dt},
\end{eqnarray}
where $N[t, t+ dt)$ represents the number of events between time $t$ and $t+dt$ and 
$\mathcal H(t)$ is the history filtration which is mathematically defined as 
$\sigma(\{N(s), s < t\})$ \citep{daley2003introduction}. 
The simplest counting process is the homogeneous 
Poisson process, where the intensity 
function does not depend on time $t$, i.e., $\lambda(t) \equiv \lambda$. 
Another common type of counting process is the
class of self-exiting processes, where the intensity 
function is positively influenced by historical events. Among self-exciting processes, the 
Hawkes process has been widely used, including 
for modeling earthquake occurrences
and financial data
\citep{ogata1988statistical,hawkes_hawkes_2018}.

Similarly, if we only observed the network structure $\mathcal{E}$ and
not the event times, traditional statistical network models could be 
applied to this data. 
Network models consisting of binary or discrete 
edges between nodes are extensively 
studied in the
statistical and machine learning literature. 
Perhaps the most widely used network model for binary edge networks 
is the stochastic block 
model (SBM). Stochastic block models assume that each node belongs to some latent 
cluster, with 
edges between nodes depending only on their latent cluster assignment 
\citep{nowicki_estimation_2001}.

When describing interactions between nodes in a network, it is often true
that the underlying interactions are in fact observed in continuous time
before then being aggregated into some discrete representation. For example,
repeated interactions between nodes in a social network could
be simply counted,
with a binary link formed if the number of (directed) interactions is 
above some threshold. One extension of these models for
static networks that has been considered
is to split the observations into multiple time windows
with a static network constructed for each of these windows.
In the context of messages on a social network, this would consist
of constructing a static network based on the interactions between nodes
in some time period (say, every week).
Community
detection methods have been developed for block models in this
context also \citep{pensky2019spectral}. However, these methods still require
compression of continuous time interactions into a static representation,
which can fail to capture the true expressive dynamics between nodes. 
Similarly, the length of window used is subjective, and it is not clear how
to choose the level of aggregation required.
The direct modeling of repeated event streams on a network has not been as
widely studied \citep{rossetti2018community}.

Recent extensions of stochastic block models have been used
to model events on networks
using point processes, the setting we consider here. 
This allows for community detection of nodes in a network
which captures the temporal dynamics which describe
events between nodes.
Suppose that $z = (z_1, \ldots, z_n)$ is a vector 
representing the latent class memberships of $n$ nodes in a network, where each node 
belongs to one of $K$ possible classes.
The latent classes are drawn 
from some vector $\pi$ 
which gives the latent probability of each of the $K$ classes. 
We assume that (directed) 
interactions between any two nodes in the network form a 
point process, which has 
intensity $\lambda_{ij}(t)$. We impose a block model 
structure on these intensities, in 
that the intensity between two nodes will depend on 
the latent class of both nodes. 
For a given node pair 
$(i,j)\in \mathcal{E}$ we have a counting process $N_{ij}(t)$.
Given node $i$ in latent class $z_i$ and node $j$ in latent class $z_j$ then we have
$$
\lambda_{ij}(t) = \lambda_{z_i z_j}(t).
$$
This model was first considered by \cite{matias_semiparametric_2018}.
In that setting, a block model was proposed where, conditional on the latent 
groups, interactions from any one node in the network to another follow an 
inhomogeneous Poisson process. The usual variational EM estimation procedure for 
binary networks was then extended to this setting, resulting in a variational 
semi-parametric EM type algorithm. Given the current estimate of the cluster 
assignments, the conditional intensities are then estimated 
using a non-parametric 
M-step, consisting of either a histogram or kernel based estimate.
A similar
model has been proposed elsewhere \citep{miscouridou2018modelling}, where edge 
exchangeable models for binary graphs are extended to this setting. 
Here, the baseline 
of a Hawkes process encodes the affiliation of each node to
the $K$ latent communities, with a common exponential kernel for all interactions. 
Inference for this model is carried out using Markov chain Monte Carlo (MCMC) 
\citep{gilks1995markov}.

While both these models are flexible and have been demonstrated 
to work well on real networks, they are both computationally intensive to fit.
Each method requires multiple iterations over all events in the network
to learn the 
community structure. 
Similarly, given the estimation procedures for these models, there 
is no immediate way to update these parameters in
the context of streaming events, 
to readily incorporate the observation of new events.
Given the continuous time nature of event streams we would like
to be able to update our estimated community structure either in real time
or, at least, without repeatedly using the entire event history. 
Online variational inference is one approach for this task, which we will first review.
We then provide an online learning procedure for models of this form which
avoids much of this computational burden and can more readily update the
community structure given new observations.

We will consider point process block models of this form
in this paper. In particular, we will consider several possible 
formulations of the conditional intensity, listed here.
Throughout this paper, we will use $\lambda$ to denote 
the generic parameters of a particular point process.

\begin{itemize}
    \item \textbf{Block Homogeneous Poisson Process Model}
    The intensity function of block homogeneous Poisson process model postulates the 
    following form
    \begin{eqnarray}\label{poisson:homo}
    \lambda_{ij}(t) = B_{z_i z_j}
    \end{eqnarray}
    The intensity function only depends on individuals' latent profile and does not depend 
    on time.
    \item \textbf{Block Inhomogeneous Poisson Process Model}
    The intensity function of block inhomogeneous Poisson process model postulates the 
    following form
    \begin{eqnarray}\label{poisson:nonhomo}
    \lambda_{ij}(t) = \sum_h a_{z_i z_j}(h) f_h(t)
    \end{eqnarray}
    where $f_h(t) \in \mathcal H$ with $\mathcal H$ being some functional space,
    {\color{black} where, throughout this paper, we use $H$ to denote the number of basis functions 
    in $\mathcal H$}. The 
    intensity function has the additive form, characterized by the linear combination 
    of basis functions. Under this case, the intensity function depends not only on an 
    individuals' latent profile but also on time.
    \item \textbf{Block Homogeneous Hawkes Process Model}
    The block homogeneous Hawkes is the exention of the original Hawkes model 
    \citep{hawkes1974cluster}.
    The intensity function postulates the following form
    \begin{eqnarray}
    \lambda_{ij}(t) = \mu_{z_i z_j} + b_{z_iz_j}\int_0^t f(s) d N_{ij}(s),
    \end{eqnarray}
    where $\mu$ represents the baseline intensity, $b$ represents the magnitude of impact 
    function and $f$ is the impact function, which indicates the influence of previous 
    events on the current intensity. A classical choice of $f$ is $f(s) = \lambda 
    \exp\{-\lambda s\}$ \citep{rizoiu2017tutorial}. {\color{black} This leads to a common
    $f(s)$ across all nodes (and communities) in the network. Note that other work 
    in the literature has considered other restrictions on $f$, with \citet{huang2022mutually}
    using a sum of known exponential kernels, to aid computational efficiency.}
    \item \textbf{Block Inhomogeneous Hawkes Process Model}
    The intensity function of the block inhomogeneous Hawkes process model postulates the 
    following form
    \begin{eqnarray}
    \lambda_{ij}(t) = \mu_{z_i z_j}(t) + b_{z_iz_j}\int_0^t f(s) d N_{ij}(s),
    \end{eqnarray}
    where $\mu$ is no longer constant over time. Instead, 
    $\mu_{kl}(t) = \sum_{h} a_{kl}(h) f_h(t)$ with $f_h(t)\in \mathcal H$ 
    with $\mathcal H$ being some functional space. 
    That is, we assume the baseline function can be characterized by the linear 
    combination of certain basis functions to capture different time patterns.
    
\end{itemize}

{\color{black} 
We here would like to highlight that this paper only focuses on
a network 
structure where a 
node behaves similarly as both source or as a 
destination and its latent membership does 
not change as the origin of an event changes.
In future work,
we could consider a ``double" group membership model as one possible 
extension. In such a formulation,
$z_i^{(s)}, z_i^{(d)}$ would be used to denote class labels of 
source and destination effects for node $i$, respectively. 
Another possible extension is to consider data with co-clustering structures (bipartite 
graph), which is a very common data format in gene expression data 
\citep{cheng2000biclustering, pontes2015biclustering} and
user-item recommendation systems
\citep{george2005scalable, wang2019solving}. In these settings,
the class labels of two 
groups of nodes are no longer symmetric. 
}

\subsection{Online Variational Inference}
In this paper we consider online variational inference 
for estimating community structure in event data on networks. 
In the general formulation of online learning,
data is observed sequentially in time with 
$\mathcal{D}_m$ being the $m$-th such observation 
and there is a loss function $\ell(\mathcal{D}_m,\hat{\theta}_m)$
for a given parameter estimate $\hat{\theta}_m$. This estimate 
$\hat{\theta}_m$
will be based on the past data 
$\mathcal{D}_{1:(m-1)}\vcentcolon= \{\mathcal{D}_1,\ldots,\mathcal{D}_{m-1}\}$
This loss function may vary depending on the inference procedure 
considered, but a natural choice is often the negative log 
likelihood, corresponding to online maximum likelihood estimation.

The aim of online inference is to find an estimate 
of the parameter $\theta$
which
is close to the best overall estimate, had all the data been observed.
We will denote such an estimate as $\theta^*$ which would 
minimize the generalization error 
$\mathcal{E}_{*}(\theta) = 
\mathbb{E}_{\mathcal{D}\sim P_{*}}\ell(\mathcal{D},\theta)$,
where $P_{*}$ is the true data distribution.
This quantity is unknown in practice and so interest instead 
lies in minimizing the cumulative error over time, 
$\sum_{i=1}^{M}\ell(\mathcal{D}_m,\hat{\theta}_m)$, 
where we observe data in $M$ sequential observations.
Given this estimate, a commonly studied quantity is the 
\emph{regret}, which is the difference between this 
cumulative error and the minimum cumulative error for a fixed 
estimate of the parameter, $\theta$,

$$
\sum_{m=1}^{M}\ell(\mathcal{D}_m,\hat{\theta}_m) 
- \inf_{\theta\in \Theta}\sum_{m=1}^{M}\ell(\mathcal{D}_m,\theta).
$$

Regret bounds can quantify the values of this quantity 
and have been obtained in certain settings \citep{shalev2012online}.
Such bounds can also
then be used to compute corresponding bounds on the 
generalization gap.
Online learning has been considered for Bayesian inference, 
which is in some sense a natural setting for such a scheme.
\citet{cherief2019generalization} describe such a setting,
where the ``online'' posterior can be written as 
$$
p_{m}^{\eta}(\theta) \vcentcolon= \frac{1}{Z_{m}^{\eta}}\pi(\theta)
e^{-\eta\sum_{i=1}^{m-1}\ell(\mathcal{D}_m, \theta)},
$$
for some learning rate $\eta$, prior $\pi$ and normalizing 
constant $Z_{m}^{\eta}$. In particular, if the loss function 
chosen is the 
log likelihood and $\eta=1$ then this is exactly standard Bayesian
inference, observing the data sequentially and updating the 
posterior. For $\eta<1$ then this is tempered Bayesian inference 
\citep{alquier2020concentration}.
\citet{cherief2019generalization} extend this 
idea to variational inference, utilising gradient 
updates based on the loss for the parameter estimates.
This is similar
to the setting we propose here. 

In particular, they 
consider three such formulations
for gradient based updates to the variational approximation
as data is observed in an online fashion. 
These differ in the objective function, which is composed 
of the gradient of the loss function and a KL term. 
In the first two cases, a regret bound can be derived for  
the corresponding updates, with one requiring the restriction 
that the variational family is mean field Gaussian.
For the gradient update corresponding to natural gradient
variational inference, a corresponding theoretical
regret bound cannot be obtained.
These results are investigated empirically on classical regression 
and classification problems, where there is \textbf{no latent 
structure} in the variational approximation.

\section{An Online Learning Framework for Event Streams}
\label{Computation}
We now outline the online learning framework we will 
utilise to perform online estimation of network community structure.
Many methods in statistics and machine learning process large data in batches. This often 
involves processing large volumes of data at the same time and repeatedly, with long 
periods of latency. 
More recently, data streaming is widely used for real-time aggregation, 
filtering, and testing. 
This allows for real time analysis of data as it is collected and can be used to gain 
insights in a wide
range of applications, such as 
social network data \citep{bifet2010sentiment} and transit data 
\citep{moreira2013predicting}.
Motivated by the aim of computational efficiency, 
in this work 
we propose a scalable online learning method for network point processes
with latent block structure to describe interactions between nodes.

\subsection{Online Learning Algorithms for Network Point Processes}
We first denote by $\theta$ the model parameters we 
wish to learn and by $l(\theta)$ 
the objective function, i.e. the log-likelihood function in our setting.   
Let $dT$ be a time window
such that $T$, the total time for which the event stream is observed,
can be subdivided into $M = T/dT$ time windows 
(we suppose $T/dT$ is an integer without loss of generality). 
For exposition we assume $T$ is known but it is not required, as we discuss later.
Following this subdivision into $M$ time intervals,
$l(\theta)$ can be rewritten as $l(\theta) = \sum_{m=1}^M l_m(\theta)$,
where $l_m(\theta)$ is the objective corresponding to log-likelihood of observed data in
$m$-th time window 
(in what follows, we use subscript $m$ to denote the quantity computed
 in $m$-th time window).

In a batch algorithm, the estimator $\hat \theta^{b}$ is defined as 
$\arg\max_{\theta} l(\theta)$, i.e. the best parameter estimate to achieve the 
maximum value of the objective function.
When $l(\theta)$ is taken as the log-likelihood 
function, $\hat \theta^{b}$ is also known as the maximum likelihood estimator 
(MLE).
Unfortunately, such optimization can become intolerably slow when
the data size becomes large and $l(\theta)$ contains latent discrete variables,
as in a SBM. 
Hence, we aim to construct an estimator $\hat \theta^{o}$ to 
approximate $\hat \theta^b$ with less computational burden,
while also
hopefully possessing the 
same properties as $\hat \theta^b$.
To this end,
we consider an online method for this optimization problem. The general scheme is described 
{\color{black} in Algorithm~\ref{gen_online}}.

\begin{algorithm}
    \caption{General Online Optimization}
    \label{gen_online}
    \begin{algorithmic}[] 
        \STATE Initialize parameters $\theta^{(0)} = \theta_0$.
        \FOR{$m = 1$ to $M$}
        \STATE \textbf{Update} $\theta$ by 
            $\theta^{(m)} = \theta^{(m-1)} + \eta_m \frac{\partial l_m(\theta)}{\partial \theta}$.
        \ENDFOR
        \STATE \textbf{Output} Set $\hat \theta^{o} = \theta^{(M)}$
    \end{algorithmic}
  \end{algorithm}

However, under our setting, the general online scheme does not apply by 
noticing that the true latent class label assignment is unknown to us. 
In other words, we need to integrate over all possible latent class 
configurations for computing the log-likelihood function, 
which is often intractable. 
In particular, for the class of models we consider here, 
$l(\theta) = \log\{\sum_z \pi_z \exp(l(\theta \vert z)) \} = 
\log \{ \sum_z \pi_z \exp (\sum_{m=1}^M l_m(\theta \vert z)) \}$, 
indicating that $l(\theta)$ can not be simply rewritten as 
$l(\theta) = \sum_{m=1}^M l_m(\theta)$.

We therefore use a variational approximation for the 
latent community assignments, which allows us to 
derive temporally evolving estimates of the 
community structure, and the corresponding 
point process intensity functions.
We take 
$$
q(z) := \prod_i q_i(z_i)
$$ 
and 
$q_i(z) = \mathcal M(\tau_i)$ where
$\tau_i = (\tau_{i1}, \ldots, \tau_{iK})$ and 
$\mathcal{M}(\tau)$ represents a multinomial distribution with 
parameter $\tau$. This is the standard mean field variational
approximation used for network models with latent 
community assignment
\citep{celisse2012consistency, matias_semiparametric_2018}
Given this, the remaining global 
parameters of our model are $\theta = (\pi, \lambda)$,
where $\lambda$ captures the parameters of the group level point processes
and $\pi$ the overall 
group proportions.
Our proposed online method for network point 
processes with latent community structure 
{\color{black} is described in Algorithm~\ref{online_point_process}}.

\begin{algorithm}
    \caption{Online Inference for Point Processes on Networks}
    \label{online_point_process}
    \begin{algorithmic}[] 
        \STATE Set initialization of $\theta^{(0)} = \theta_0$
        {\color{black} and $S^{(0)}(z_i) = 1/K$ for $z_i = 1, \ldots, K$, for each node $i$.}
        \FOR{$m = 1$ to $M$}
        \STATE \textbf{Update} the latent distribution 
            $q^{(m)}(z) = \prod_{i=1}^n q_i^{(m)} (z_i)$ by 
            \begin{equation}\label{key:step}
            q_i^{(m)} (z_i) \propto \pi^{(m-1)} \times 
            \exp\left\{\mathbb E_{q^{(m-1)}(z_{-i})} l_m(\theta^{(m-1)}\vert z)\right\} 
            \cdot S^{(m-1)}(z_i), 
            \end{equation}
            {\color{black} where 
            \begin{equation*}
            S^{(m)}(z_i) = S^{(m-1)}(z_i) \times  
            \exp\{\mathbb E_{q^{(m-1)}(z_{-i})} 
            l_m(\theta^{(m-1)}\vert z)\}
            \end{equation*}
            }
            giving estimates $\hat{\tau}^{m}$.
        \STATE \textbf{Update} the point process parameters by 
            \begin{align}\label{eq:alg:gd}
            \lambda^{(m)} = \lambda^{(m-1)} + \eta_m \frac{1}{ \vert A \vert }\frac{\partial \mathbb E_{q^{(m)}(z)} 
            l_m(\theta \vert z)}{\partial \lambda}.
            \end{align}
        \STATE \textbf{Update} the community proportions using 
        $$
        \pi_k =\frac{1}{n}\sum_{i=1}^{n}\tau_{ik}, \mbox{ for } k=1,\ldots,K.
        $$
        \ENDFOR
        \STATE \textbf{Output} Set $\hat \lambda^{o} = \lambda^{(M)}$,
        $\hat{\tau}^{o}=\hat{\tau}^{(M)}$ and $\hat{\pi}^{o}=\hat{\pi}^{(M)}$.
    \end{algorithmic}
  \end{algorithm}

{\color{black} To further expand on Algorithm~\ref{online_point_process},}
the 
quantity $S^{(m)}$ can be viewed as an $n$ by $K$ matrix 
which stores personal cumulative group evidence up to the
current time window for each individual $i$ and latent 
class $k$.\footnote{Here $z_{-i}$ is a sub-vector of $z$ with the
$i$th entry removed.}
The step size $\eta_m$ is the adaptive learning speed, 
which may depend on $m$.

One of the main contributions of our algorithm is that we update the 
distribution of latent profiles adaptively by using cumulative historical information. 
An individual's latent profile is approximated by a sequence of probability distributions,
$q^{(m)}(z) = \prod_{i=1}^n q_i^{(m)}(z_i)$, by 
assuming there is no dependence structure between the latent assignment of nodes.
In the update of $q^{(m)}$ we do not need to go through past events, as all group 
information has been compressed into the cumulative matrix $S^{(m)}$. 
Under mild assumptions and in suitable settings, this approximation works well and 
leads to consistent parameter estimation. We discuss 
the details of
this approximation 
further in Section~\ref{vi_approx}  

This model is of a similar form to that proposed for online estimation of LDA, 
where documents arrive as streams \citep{hoffman2010online}. In that setting, 
each document of $D$ known documents in the corpus is observed sequentially. 
After word counts of an
individual document are observed, an E-step 
is performed to determine the optimal local parameters for the per-document 
topic weights and per word topic assignments. Then an estimate of the optimal 
global of the topic weights is computed, $\tilde{\lambda}$, as if the total
corpus consisted of the current document observed $D$ times. 
The actual estimate of $\lambda$,
which parameterizes the posterior distribution over the topics,
is estimated using a weighted average of
the previous estimate and $\tilde{\lambda}$. This is similar 
in spirit to our proposed method, where we compute optimal values 
given the current observation data and update our overall estimates using 
these estimates from our current window.
Another related procedure was proposed by \citet{broderick2013streaming}, 
who highlights that the posterior targeted by SVI is generally based on 
the existence of a full dataset involving $D$ data points (e.g, documents).
A similar scenario occurs here, where standard SVI could be applied if we knew 
$T$, the total observation period, in advance. However, like 
\citet{broderick2013streaming}, we wish to consider a model 
which is flexible to the total observation period and so 
can be applied in a true streaming setting, as new data is being 
observed and consequently, the total observation period is 
changing. Our inference procedure is defined in terms of a fixed $T$,
however, our updates are only dependent on the data that has been observed 
up to the current time window. We can continue to 
apply our method as new data are observed, resulting in 
an increased $M$.

In the framework described by \cite{cherief2019generalization}, 
our proposed variational procedure is of the form denoted ``streaming variational
Bayes", where we seek to optimize a loss function in terms of
the current parameter estimates and the current variational
family, given previously observed data.

We provide detailed algorithms for learning Poisson processes and Hawkes processes
on networks of event streams. Specifically, Algorithm~\ref{algo:hom_pois}
describes the detailed online estimation procedure for the homogeneous Poisson process.
This is the simplest case but illustrates the main components of our inference scheme.
\footnote{The algorithm for the non-homogeneous Poisson process is similarly constructed.} 
It only requires storing the
cumulative number of events without storing any event history. 
This largely reduces memory cost.
Similarly, the appendix
describes the detailed online estimation procedure for the homogeneous Hawkes process 
with exponential-type impact function.\footnote{The corresponding algorithm for the 
non-homogeneous Hawkes process is similarly constructed.}
Also included is a support algorithm which describes the detailed procedure for keeping 
historical data by creating a hash map with the key being the pair of nodes and their 
history information.
From the view of statistical discipline, we only need to store the \textit{sufficient 
statistics} \citep{lehmann2006theory} which already contains all information about model 
parameters.
Specifically, we create a hashmap $\mathcal H$, whose key is `$(i,j)$" ($i,j \in [n]$) 
and corresponding value is the sufficient statistic of the specific model. These
values will be updated by incorporating new information, as new data in the current 
time window is processed.
Hence, the proposed algorithm effectively optimizes computational memory costs.

\subsection{Approximation via Variational Inference}
\label{vi_approx}
Before continuing, we first wish to expand the discussion of the 
streaming variational approximation being considered here.
When the labels of individuals are known, the conditional log likelihood can be written 
explicitly as
\begin{equation*}
l(\theta \vert z)   = \sum_{(i,j) \in A} \bigg\{\int_0^T \log \lambda_{ij}(t\vert z) dN_{ij}(t)
 - \int_0^T \lambda_{ij} (t\vert z) dt \bigg \}.
\end{equation*}
Then the complete log likelihood is
\begin{eqnarray}\label{likelihood_complete}
l(\theta, z) = \sum_{i=1}^n \log \pi_{z_i} + l(\theta\vert z).
\end{eqnarray}
Furthermore, the marginal log likelihood can be written as 
\begin{eqnarray}\label{likelihood}
l(\theta) = \log \left \{ \sum_z \left [ \prod_{i=1}^n \pi_{z_i} L(\theta\vert z)
\right ] \right \},
\end{eqnarray}
where $L(\theta \vert z) = \exp\{l(\theta \vert z)\}$ is the conditional likelihood.

As seen in \eqref{likelihood}, it is difficult to compute this likelihood directly,
which requires summation over exponentially many terms.
An alternative approach is by using variational inference \citep{hoffman2013stochastic} 
methods to optimize the evidence lower bound (ELBO) instead of the log likelihood.
The ELBO is defined as 
\begin{eqnarray}
\textrm{ELBO}(\theta) = \mathbb E_{q(z)} l(\theta, z) - \mathbb E_{q(z)} \log q(z),
\end{eqnarray}
where this expectation is taken with respect to $z$ and $q(z)$ is some approximate 
distribution for $z$.  
For computational feasibility, we take $q(z) := \prod_i q_i(z_i)$ and 
$q_i(z) = \mathcal{M}(\tau_i)$ with 
$\tau_i = (\tau_{i1}, \ldots, \tau_{iK})$.

By calculation, the ELBO can be obtained,
\begin{equation}
\textrm{ELBO}  = \sum_{(i,j) \in A} \sum_{k,l} \tau_{ik} \tau_{jl} \bigg \{\int_0^T \log \lambda_{kl}(t) dN_{ij}(t) 
  - \int_0^T \lambda_{kl} (t) dt \bigg \} 
  + \sum_i \sum_k \tau_{ik} \log{\pi_k / \tau_{ik}}.
\label{full_elbo}
\end{equation}
We can then define
\begin{equation*}
\mathbb E_{q(z)} l_m(\theta \vert z)  = 
\sum_{(i,j) \in A} \sum_{k,l} \tau_{ik} \tau_{jl} 
\bigg \{\int_{(m-1)\cdot dT}^{m \cdot dT} \log \lambda_{kl}(t) dN_{ij}(t) 
 - \int_{(m-1) \cdot dT}^{m \cdot dT} \lambda_{kl} (t) dt \bigg \},
\end{equation*}
and therefore, the ELBO can be rewritten as 
\begin{eqnarray*}
\textrm{ELBO} &=& \sum_{m=1}^M \mathbb E_{q(z)} l_m(\theta \vert z) +
 \sum_i \sum_k \tau_{ik} \log \pi_k/\tau_{ik}.
\end{eqnarray*}
Hence, the new representation is in additive form, which is more amenable to 
online optimization.  

Define the estimator $\hat \tau_i^{(m)}$ to be the maximizer for $m$-th time 
window of individual $i$ as
\begin{equation}
\label{objective:tau}
\hat\tau_i^{(m)}  \equiv 
\textrm{argmax}_{\tau_i}
\bigg \{\sum_{w=1}^{m} \mathbb E_{q_i(z_i)} \mathbb E_{q^{(w-1)}(z_{-i})}
l_w(\theta^{(w-1)} \vert z) 
  + \sum_i \sum_k \tau_{ik} \log \pi_k^{(m-1)}/\tau_{ik} \bigg \} .
\end{equation}
We {\color{black} utilise} Theorem \ref{thm:approx} 
to explain that the approximation step in our proposed algorithm is aiming
to find the best approximate posterior distribution for each individual at each time window.
\begin{theorem}\label{thm:approx}
The optimizer of \eqref{objective:tau} is given by equation \eqref{key:step}.
\end{theorem}
\begin{proof}
By simplification, we have that 
\begin{multline*}
 \sum_{w=1}^{m} \mathbb E_{q_i(z_i)} \mathbb E_{q^{(w-1)}(z_{-i})} l_w(\theta^{(w-1)} \vert z) 
+ \sum_i \sum_k \tau_{ik} \log \pi_k^{(m-1)}/\tau_{ik} \\
= \sum_{k = 1}^K \tau_{ik} \sum_{w=1}^m 
\mathbb E_{q^{(w-1)}(z_{-i})} l_w(\theta^{(w-1)} \vert z_{-i}, z_i = k) 
+ \sum_k \tau_{ik} \log \pi_k^{(m-1)} - \sum_k \tau_{ik} \log \tau_{ik} + C_1 
\end{multline*}
\begin{multline*}
= \sum_{k=1}^K \tau_{ik} \log \bigg \{  \pi_k^{(m-1)} \exp \bigg [ 
\sum_{w=1}^m
 \mathbb E_{q^{(w-1)}(z_{-i})} l_w(\theta^{(w-1)} \vert z_{-i}, z_i = k) \bigg ] \bigg \} 
 -  \sum_{k=1}^K \tau_{ik} \log \tau_{ik} + C_1 
\end{multline*}
\begin{equation*}
 = - KL(q_i \ \vert  p_i) + C_2,
\end{equation*} 
where $C_1, C_2$ are some constants free of $\tau_i$ and $p_i$ is some 
multinomial distribution with 
\begin{equation*}
p_i(z = k) \propto \pi_k^{(m-1)} \times   
\exp\left\{\sum_{w=1}^m \mathbb E_{q^{(w-1)}(z_{-i})} 
l_w(\theta^{(w-1)} \vert z_{-i}, z_i = k)\right\}.
\end{equation*}
Hence, the maximizer is achieved when $q_i = p_i$, that is 
$$
\tau_{ik} \propto \pi_k \exp\{\sum_{w=1}^m \mathbb E_{q^{(w-1)}(z_{-i})} 
l_w(\theta^{(w-1)} \vert z_{-i}, z_i = k)\}.
$$
Lastly, we denote $\exp\{\sum_{w=1}^m 
\mathbb E_{q^{(w-1)}(z_{-i})} l_w(\theta^{(w-1)} \vert z_{-i}, z_i = k)\}$
as $S^{(m)}(k)$, which could be computed recursively by the 
formula
\begin{equation*}
S^{(m)}(k) = S^{(m-1)}(k) \times 
\exp\left\{\mathbb E_{q^{(n-1)}(z_{-i})} {\color{black} l_m}(\theta^{(m-1)} \vert z_{-i}, z_i = k)\right\}.    
\end{equation*}
This completes the proof.
\end{proof}

\begin{algorithm}[ht]
	\caption{Online-Poisson}
	\begin{algorithmic}[1]
	    \STATE Input: $data$, number of groups $K$, window size $dT$, edge list $A$.
	    \STATE Output: $\hat B$, $\hat \pi$.
	    \STATE Initialization: {\color{black}Set $S = 1/K$ for all entries, $\pi=1/K$, $B = [B_{kl}]_{k,l=1}^{K}$ from a uniform distribution. Set $\tau$ from a non-informative prior or using Algorithm~\ref{alg:ini}.}
	    \STATE Set $M = T/dT$ 
		\FOR{window $m = 1$ to $M$}
		\STATE Read new data between $[(m-1)\cdot dT, m \cdot dT]$
		\STATE Create temporary variables $S_p\in \mathbb{R}^{n\times K},B_{p1},B_{p2}\in \mathbb{R}^{K\times K}$.
		\STATE Set learning speed: $\eta = \frac{K^2}{\sqrt{m} n_t}$, where $n_t$ is the number of events between $[(m-1)\cdot dT, m  \cdot dT]$.
		\FOR{events in current window}
		\STATE Compute $B_{p1},B_{p2},S_p$:
		\STATE {\color{black}$S_p(i,k) = S_p(i,k) + \tau_{jl} $} for $i,j$ in events  
		\STATE {\color{black}$S_p(i,k) = S_p(i,k) -\tau_{jl}B_{kl}dT $} \\ \hspace{2.5cm} for $i,j$ in $A$
		\STATE {\color{black}$B_{p1}(k,l) = B_{p1}(k,l) + \tau_{ik}\tau_{jl}$} \\ \hspace{2.6cm}for $i,j$ in events
		\STATE $B_{p1} = B_{p1}/B$
		\STATE {\color{black}$B_{p2}(k,l) = B_{p2}(k,l) + \tau_{ik}\tau_{jl}$} for $i,j$ in $A$
	    \STATE {\color{black}$S\ = S + S_p$.}
	    \ENDFOR
	    \STATE Compute the negative gradient: $grad_B=B_{p1}-B_{p2}$. 
	    \STATE Update the parameters: 
	    \STATE Update $B$ by setting $B = B + \eta\cdot grad_{B}$
	    \STATE Update $\tau$ by setting $\tau_{ik} = \frac{\pi_k S_{ik}}{\sum_{k}\pi_k S_{ik} }$ for $i \in [n]$ and $k \in [K]$.
	    \STATE Update $\pi$ by setting $\pi_k = \frac{1}{n}\sum_{i}\tau_{ik}$ for $k=1,\ldots,K$.
		\ENDFOR
	\end{algorithmic} 
	\label{algo:hom_pois}
\end{algorithm}

\section{Local Convergence Analysis}
\label{Theory}
One natural question is how to better understand the theoretical properties of 
our proposed estimator.
Does the online algorithm provide a consistent estimator? How fast does the 
estimator converge to the true model parameters?
Different from regular online algorithm analysis, the key difficulties 
under the current setting are that the model we consider is a latent 
community
network model with complicated dynamics, and the proposed algorithm involves 
a variational approximation step. 

We present results which aim to address these questions. Specifically, 
Theorem \ref{regret:thm} provides a theoretical guarantee for the regret bound,
and 
Theorem \ref{rate:thm} characterizes the local convergence rate of the 
proposed online estimator.
Theorem \ref{thm:q:community} states that the variational distribution will 
place mass on the true node labels exponentially fast,
as the number of observation windows increases.
In these theoretical results we use $C, C_0$ and $c$ to denote constants which
arise in the statements of these theorems, as required.

\paragraph{Regret Analysis.} Before describing the main results, 
we first introduce some required notation and 
definitions. We define the loss function over the $m$-th time window as the 
negative normalized log-likelihood, i.e.
\begin{equation*}
\tilde l_m(\theta\vert z) = - \frac{1}{\vert A\vert }\sum_{(i,j)\in A} 
\bigg\{ \int_{(m-1)dT}^{m dT} \log \lambda_{ij}(t\vert z)dN_{ij}(t)  
 - \int_{(m-1)dT}^{m dT} \lambda_{ij}(t \vert z)dt \bigg\},
\end{equation*}
and define the regret as 
\begin{eqnarray*}
\textrm{Regret}(T) = 
\sum_{m=1}^M \tilde l_m(\theta^{(m)} \vert z^{\ast}) - 
\sum_{m=1}^M \tilde l_m(\theta^{\ast} \vert z^{\ast}),
\end{eqnarray*}
with $M = T/dT$. In these results we assume $T$ fixed, although as
stated this is not required for our practical implementation.
Regret($T$) quantifies the gap of the 
conditional likelihood, given the true latent membership $z^{\ast}$,
between the online estimator and the true optimal value,
after constructing our online estimates using observations over a fixed time period.

Notice that this problem is not convex, and 
we cannot guarantee the global convergence of the proposed method. 
However, when we take the initial value of $\theta$ sufficiently close 
to the true model parameters, 
we show that
the average regret vanishes with high probability.
The result is stated {\color{black}in Theorem~\ref{regret:thm}},
with the detailed conditions and technical proofs included in 
Appendix~\ref{app:thm_proof}.

\begin{theorem}\label{regret:thm}
Under regularity conditions C1 - C7 \footnote{See Appendix \ref{app:condition}},
for any $\theta^{(0)} \in B(\theta^{\ast}, \delta)$ and 
step size $\eta_m = \frac{c}{\sqrt{T}}$, we have that 
\begin{eqnarray}
\textrm{Regret}(T) \leq C_0 \sqrt{T} (\log (T \vert A\vert ))^2,
\end{eqnarray}
which holds with probability tending to 1 as $n\rightarrow\infty$.
\end{theorem}
Note that in Theorem \ref{regret:thm}, the  ``log" term comes from the fact 
that the number of events is not bounded in any fixed length time window, 
but can be bounded by some large number in log order with high probability.

\paragraph{Parameter Estimation.}
By considering the step size $\eta_m = m^{-\alpha}$, we 
further have 
{\color{black}Theorem~\ref{rate:thm}, describing the rate of local convergence for the model parameters.}
\begin{theorem}\label{rate:thm}
Under the same regularity conditions of 
Theorem~\ref{regret:thm} and $\theta^{(0)} \in B(\theta^{\ast}, \delta)$,
for $0 < \alpha < 1$, we have that 
$\ \vert \theta^{(m)} - \theta^{\ast}\ \vert _2^2 = O_p\left(m^{-\alpha} \log(T\vert A \vert\right)^2 +
 1/\sqrt{\vert A\vert })$ (for $m = 1, \ldots, M$) as $n \rightarrow \infty$.
\end{theorem}

In Theorem \ref{rate:thm}, we show that the proposed estimator convergences 
to the true value under a certain rate with high probability. 
The rate is affected by three factors, (1) the learning speed $\eta_m$, 
(2) the noise term $1/\sqrt{\vert A\vert }$ and 
(3) the additional ``log" term.
Noting that $\vert A\vert $ is the number of active pair of nodes,
we can view this as the level of network sparsity. 
Thus, as is expected, sparser networks may lead to larger sampling errors. 

\paragraph{Community Recovery.}
Moreover, with suitable initialization conditions on 
model parameters $\theta^{(0)}$ and the variational distributions $q^{(0)}$'s,
we show that $q_i$ will concentrate on the true label of the 
$i$-th node at an exponential rate.
We state this result below, deferring further discussion to the appendix.

\begin{theorem}\label{thm:q:community}
Under the regularity conditions of Theorem~\ref{regret:thm}, with 
probability $1 - M \exp\{-C d_n\}$, we have 
$q_i^{(m)}(\hat z_i = z_i^{\ast}) \geq 1 - C \exp\{-c m d_n\}$ for all $i \in \{1, \ldots, n\}$ and $m = 1,\ldots, M$. (Here degree number $d_n = \Theta(n^{r_d})$ with $0 < r_d < 1$.)
\end{theorem}

{\color{black} Before examining the empirical performance of our
proposed algorithm, we first describe some key steps in the 
practical implementation and highlight one important potential extension.}

\paragraph{Initialization.} 
In our theory, we require that the starting point $\theta^{(0)}$ is close to 
the true 
model parameter and the initial variational distribution $q^{(0)}$ satisfies Condition C7.
In practice, we can choose a random starting point 
for our initial values, such as by
sampling from 
uniform priors. 
Although this
initialization can work
well when we simply sample $q_i^{(0)}$ from the non-informative prior distribution
(e.g. the multinomial distribution $\mathcal{M}(1,(1/K, \ldots, 1/K))$), 
we acknowledge that there is a 
theoretical gap between the theory and the algorithm.
Recent work has characterized the 
landscape of the variational 
stochastic block model \citep{mukherjee2018mean},
claiming the futility of random initialization by 
showing that the parameter 
estimate falls in the neighborhood of a local stationary point with high probability.
However, such results are not enough to imply that the algorithm always fails to
find the global optimum. 
It is possible that the parameter estimate may leave the region of the local optimum after several iterations.

To further improve the performance of our online estimation method,
we propose a heuristic initialization procedure as given in Algorithm \ref{alg:ini}.
The high level idea is that we use the data in first $n_0 := C \log M$ windows
for estimating initial parameter and group labels.
We first check whether there are enough nodes with sufficiently large degrees.
If not, we determine that there is insufficient information in this initial 
data and use random initialization.
Otherwise, we get initial estimates of the 
community structure by performing three steps.
In the first step, we isolate nodes with large degrees.
For each such node,
we estimate the node-wise intensity parameters based on the 
events between that node and its neighbours.
We then do clustering on these estimates, obtaining node-wise center parameters.
We next do $K$-means clustering on these node-wise center parameters,
to get initial labels of the large-degree nodes. 
Given these nodes and their label estimate,
we fit the block point process model,
getting initial model parameter estimate $\tilde \lambda$.
In the second step, we estimate an initial group label for the remaining nodes.
To do this, for each remaining node,
we assign it to a community corresponding to the highest log-likelihood type score.
In the third step, given all initial node labels, 
we use all events within the initial time interval $[0, n_0 \cdot dT]$
to get final initial estimates of the network point processes.

Theoretically speaking, Algorithm \ref{alg:ini} works when there
are sufficiently many large-degree nodes,
leading to accurate initial estimates of the community structure and 
resulting estimates of the point process parameters also.
We note that we do not adopt a spectral clustering method
\citep{ng2001spectral, gao2017achieving} as an initializer in our setting for two reasons. One 
is to avoid large memory and computational cost. Secondly, 
there is no universal way to compute a good Laplacian matrix from event stream data
(and doing so may remove the dynamic structure which can be captured by our
point process models).

{\color{black}

\paragraph{Window Size}
One parameter in our model is $dT$, the chosen window size over which the
updates occur. We note that as this $dT$ does not lead to an aggregation
of the underlying data, the choice of $dT$ should not directly influence community 
detection performance.
We consider the choice of $dT$ in terms of community recovery
performance and computational time further in Appendix~\ref{appendix:add_sims}.
These results indicate that our model is reasonably robust to different choices of 
$dT$. However, in practice it makes sense to choose $dT$ sufficiently
large so that each window will contain at least one event, 
avoiding redundant updates.

\paragraph{Learning Rate.}

As suggested by Theorem \ref{rate:thm}, it is sufficient to ensure 
the step size $\eta_m$ has order
$m^{-\alpha}$ with $0 < \alpha < 1$.
When $\alpha > 1$, the learning rate 
is too small,
not allowing sufficient 
exploration over the parameter space. 
When $\alpha = 0$, the learning rate will be too large, 
such that the variance of the data observed in each window will 
dominate the estimation error and 
$\vert \theta^{(m)} - \theta^{\ast}\ \vert _2$  does not go to zero.  
When $\alpha$ approaches $0$, larger weight will be placed on the most recent 
observations in the process.
Furthermore, in the update of point process parameters, 
we could perform the gradient ascent step given in \eqref{eq:alg:gd} 
more 
than one time,
which would allow the data in the $m$-th time window to have more impact on $\theta^{(m)}$.  
}

{\color{black}
\paragraph{Community Switching.}
In the current work, we assume that the underlying latent class label $z_i^{\ast}$'s remains the same 
throughout 
the whole time span.
However, in many real-world situations, $z_i^{\ast}$ could change over time. 
Based on our current algorithm \ref{online_point_process}, 
we could make use of $S^{(m)}(z_i)$ to design an index 
to detect the node-level change, i.e., whether the 
underlying class label has changed or not for node $i$.
To be more specific, we could monitor the index 
$\mathcal D_i(m_1, m_2) := \tilde S^{(m_2)}(\hat z_i^{(m_1)}) - \tilde S^{(m_1)}(\hat z_i^{(m_1)})$,
where $\tilde S^{(m)}(k):= S^{(m)}(k) - \max_{k'} S^{(m)}(k')$ stands for the cumulative evidence gap (if the 
value is higher, the node is more likely to belong to class $k$) and $\hat z_i^{(m_1)} := \arg\max_k 
q_i^{(m_1)}(k)$ is the estimated label for node $i$ at the $m_1$-th window.
If the index $\mathcal D_i(m_1, m_2)$ is far below some threshold, then we could conclude that the 
underlying class label has changed between the $m_1$-th window and the $m_2$-th window. We hope to address
this issue further in future work.
}

\begin{algorithm*}[ht]
	\caption{Initialization Procedure for parametric network point process}\label{alg:ini}
	\begin{algorithmic}[1]
		\STATE Input: $data$, number of groups $K$, window size $dT$, edge list $A$.
		\STATE Output: $\hat z^{(0)}$, $\hat \theta^{(0)}$.
		\STATE Read data between $[0, n_0 \cdot dT]$, where $n_0 = C \log(M)$ with $C$ being a large constant. 
		\STATE We order nodes according to the decreasing order of degree numbers such that $d_{i_1} \geq d_{i_2} \geq \ldots \geq d_{i_m}$.
		\STATE We take out first $m_0$ nodes such that $d_{i_{m_0}} \geq C K$. 
		\STATE If $m_0 < K$, then we do random initialization of $\theta$ and $z$'s.
		\STATE If $m_0 \geq K$, then do the following.
		\STATE [~~~~~\texttt{==== Step 1 =====}]
		\STATE For $i \in \{i_{1}, \ldots, i_{m_0}\}$, we fit a uni-parametric model to pair $(i,i')$ with $i' \in \mathcal N_i$ to get estimated parameter $\theta_{ii'}$.
		We then perform $K$-means method to ($\theta_{ii'}$, $i' \in \mathcal N_{i}$) to get center vector $\overrightarrow{\theta_i}$ (entry of vector is sorted in ascending order).
		\STATE We perform $K$-means method again onto center vectors 
		$(\overrightarrow{\theta_i}, i \in \{i_{1}, \ldots, i_{m_0}\})$ to obtain the estimated memberships of those $i$'s, (i.e., $\hat z_{i_1}^{(0)}, \ldots, \hat z_{i_{m_0}}^{(0)}$).
		\FOR{$k_1$ from 1 to $K$}
		\FOR{$k_2$ from 1 to $K$}
		\STATE Find all pairs, $i'$ and $i''$ $\in \{i_1, \ldots, i_{m_0}\}$ such that
		 $i'$ belongs to group $k_1$ and $i''$ belonging to group $k_2$ and nodes $i'$, $i''$ connect with each other.
		\STATE Set $\tilde \theta_{k_1 k_2}$ be the estimated parameter by fitting the process from all these pairs of nodes.
		\ENDFOR
		\ENDFOR
		\STATE [~~~~~\texttt{==== Step 2 =====}]
		\FOR{$i$ from $i_{m_0 +1}$ to $i_m$}
		\FOR{$k$ from 1 to $K$}
		\STATE We plug in estimated parameter $\tilde \theta = (\tilde \theta_{k_1 k_2})$to get the score $s_{ik}$,
		\[s_{ik} := \sum_{i': i' \in \{i_1, \ldots, i_{m_0}\}, A_{ii'} = 1} l(\tilde \theta\vert z_i = k, \hat z_{i'}^{(0)}) + 
		\sum_{i': i' \in \{i_1, \ldots, i_{m_0}\}, A_{i'i} = 1} l(\tilde \theta\vert\hat z_{i'}^{(0)}, z_i = k). \]
		\ENDFOR
		\STATE We set the initial distribution of $z_i$ as $\tau_{ik} = q_i(k) \propto \exp\{s_{ik}\}$ and let $\hat z_i^{(0)} = \arg\max_k q_i(k)$ (If there is a tie, we break them uniformly randomly).
		\ENDFOR
		\STATE [~~~~~\texttt{==== Step 3 =====}]
		\STATE Given estimated label $\hat z^{(0)}$, we fit all data in [0, $n_0 \cdot dT$] to get initial parameter estimate $\hat \theta^{(0)}$.
		\STATE Return $\hat \theta^{(0)}$ and $\hat z^{(0)}$.
	\end{algorithmic} 
	\label{alg:initial:TPP:sparse}
\end{algorithm*}

\section{Experiments}
\subsection{Evaluation on Synthetic Data}
\label{Simulation}

Given our proposed inference scheme, we first wish to 
thoroughly validate its performance in simulation studies.
We shall evaluate our procedure in terms of both community 
and parameter recovery across a range of experimental settings,
while also investigating 
the empirical regret performance and monitoring the online
loss.
{\color{black}
For each of these experimental settings, unless otherwise stated,
we consider a fixed total
observation period of $T=200$ with $K=2$ equally sized communities
and a network of $n=200$ nodes.
The full details of this simulation setting are provided in 
Appendix~\ref{appendix:add_sims} and in the associated code repository.
We 
repeat each experiment 50 times, allowing us to 
examine the variability in these estimates. 
We demonstrate the performance of our algorithm 
using the block inhomogeneous Poisson process model
with $H=2$ here, (reference equation)
a flexible model for network point processes. We show several simulation settings of interest here, including further 
simulations for this model, and for the block homogeneous Hawkes model in 
Appendix~\ref{appendix:add_sims}.
All code used to create these simulations is available in 
an \href{link to code}{online repository}.
}

One important consideration here is the choice of $dT$, the window
size over which event data is processed. Here, throughout,
we use a fixed window size of $dT=1$. Additional simulations 
in Appendix~\ref{appendix:add_sims}
indicate that in practice, once the number of windows, $M$, 
is sufficiently large, this choice does not impact the overall
algorithm performance. We demonstrate the performance of our algorithm 
using the inhomogeneous Poisson process here, 
a flexible model for network point processes. 
We show several simulation results here, deferring remaining
experiments to Appendix~\ref{appendix:add_sims},
where we also demonstrate many of these results using
Hawkes process network point processes.

\paragraph{Impact of initialization procedure.}
As described in Section~\ref{Theory}, we include a practical initialization
scheme for online community detection, obtaining estimates 
using some of the initial events. To evaluate the effectiveness 
of this initialization procedure. We show the performance 
in terms of community recovery, as we vary the sparsity, 
{\color{black}$\rho$, of the network.
This corresponds to the probability of a node pair interacting, leading to the 
observation of temporal events between those nodes. We consider here small networks with
$n=200$ nodes. In small networks
the initialization can be particularly important, due
to the limited data available. We measure the performance here in terms of
the adjusted rand index (ARI) \citep{hubert1985comparing} between the true community assignments
and the estimates at the end of our online procedure, after random initialization and using 
our procedure.
Fig~\ref{fig_ocd:init} shows the corresponding ARI across 200 simulations for each choice 
of $\rho$ considered.
Except when the network is
very sparse, the initialization scheme leads to much better community recovery, 
after using these initial estimates to fit 
the model to the remaining events. We use $\rho=0.15$ in the following 
experiments unless otherwise stated.
}

\begin{figure}[ht]
  \centering
  \includegraphics[width=0.45\textwidth]{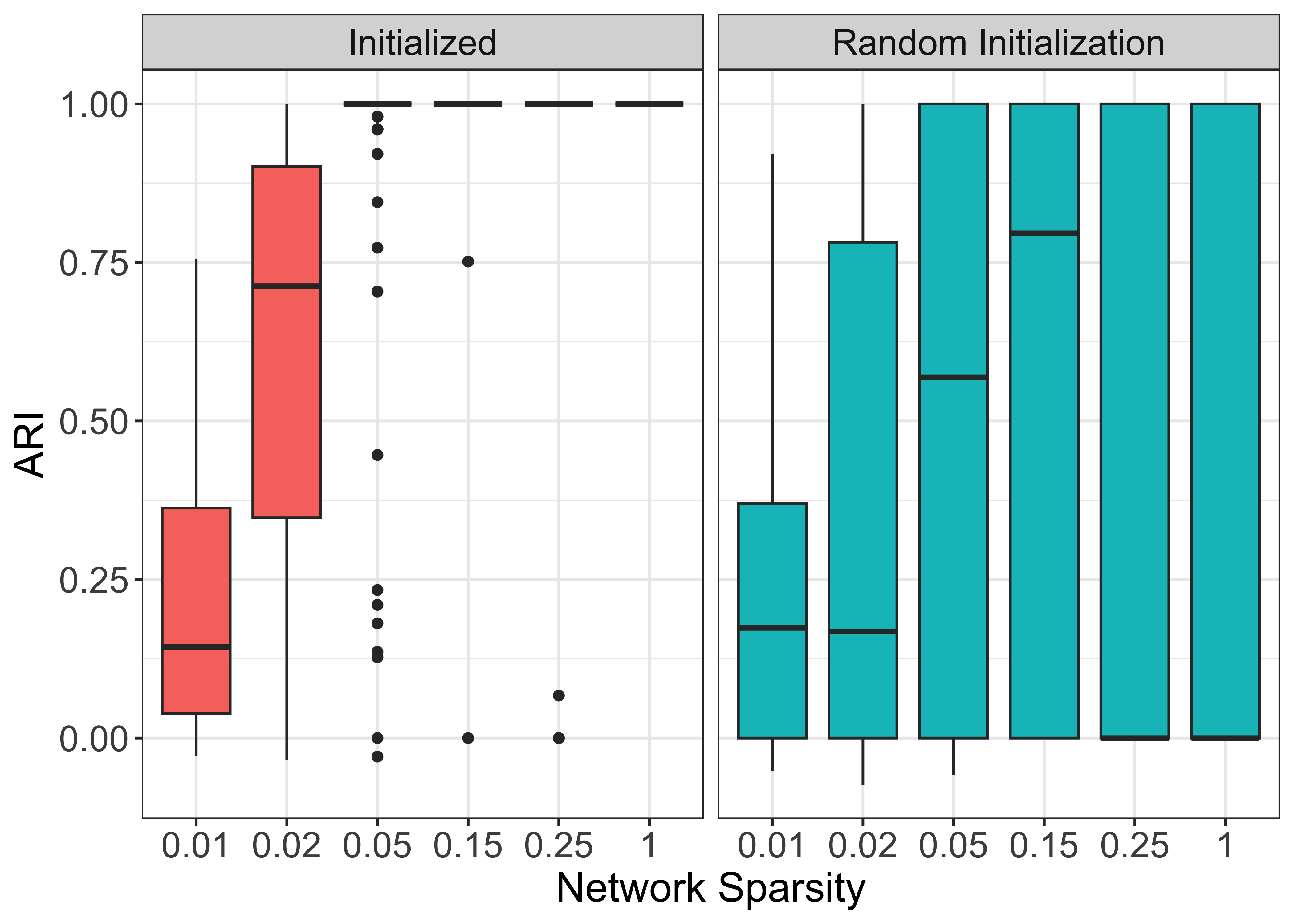}
  \caption{The role of our proposed initialization scheme on final community recovery. 
  While community recovery is challenging in very sparse networks, the initialization scheme leads to
  an improvement over random initialization.}
  \label{fig_ocd:init}
\end{figure}

\paragraph{Recover communities.}
Given this initialization scheme, we wish to demonstrate that
we can correctly recover the true communities 
for simulated data, and that this can be achieved
in an online fashion, as the data is observed. 
{\color{black} Here we consider 
unequal community sizes, with 40\% of nodes in one community and
60\% in the other.
We evaluate the performance of}
our proposed inference scheme on the final estimated community assignments,
having learnt this structure in an online fashion 
taking a single pass through the observed events. 
Figure~\ref{fig_ocd:recov_n} demonstrates the performance, in 
terms of ARI,
as we increase the size of the
network, considering $n=100,200,500,1000$ with
{\color{black} all other}
parameters fixed. 
We see that as the number of nodes in the network grows, with the total observation
period remaining fixed, we can better recover the true community 
structure of the nodes. In particular, for small networks it can be quite
challenging to recover the true structure, however once the number of 
nodes increases,
along with the corresponding number of total 
events which occur on the network,
we are better able to capture the true structure.
Similarly, while we can also investigate community recovery as we vary the
number of communities, for a fixed number of nodes
{\color{black} considering $K$ equally sized communities.}
Figure~\ref{fig_ocd:recov_k}
illustrates that as 
we increase the number of communities {\color{black}$K$ in 
our data, holding all other parameters fixed}, we remain 
able to recover the true community structure quite well, with expected increased 
uncertainty.

\begin{figure}[ht!]
    \centering
    \begin{subfigure}{0.45\textwidth}
     \centering
     \includegraphics[width=\textwidth]{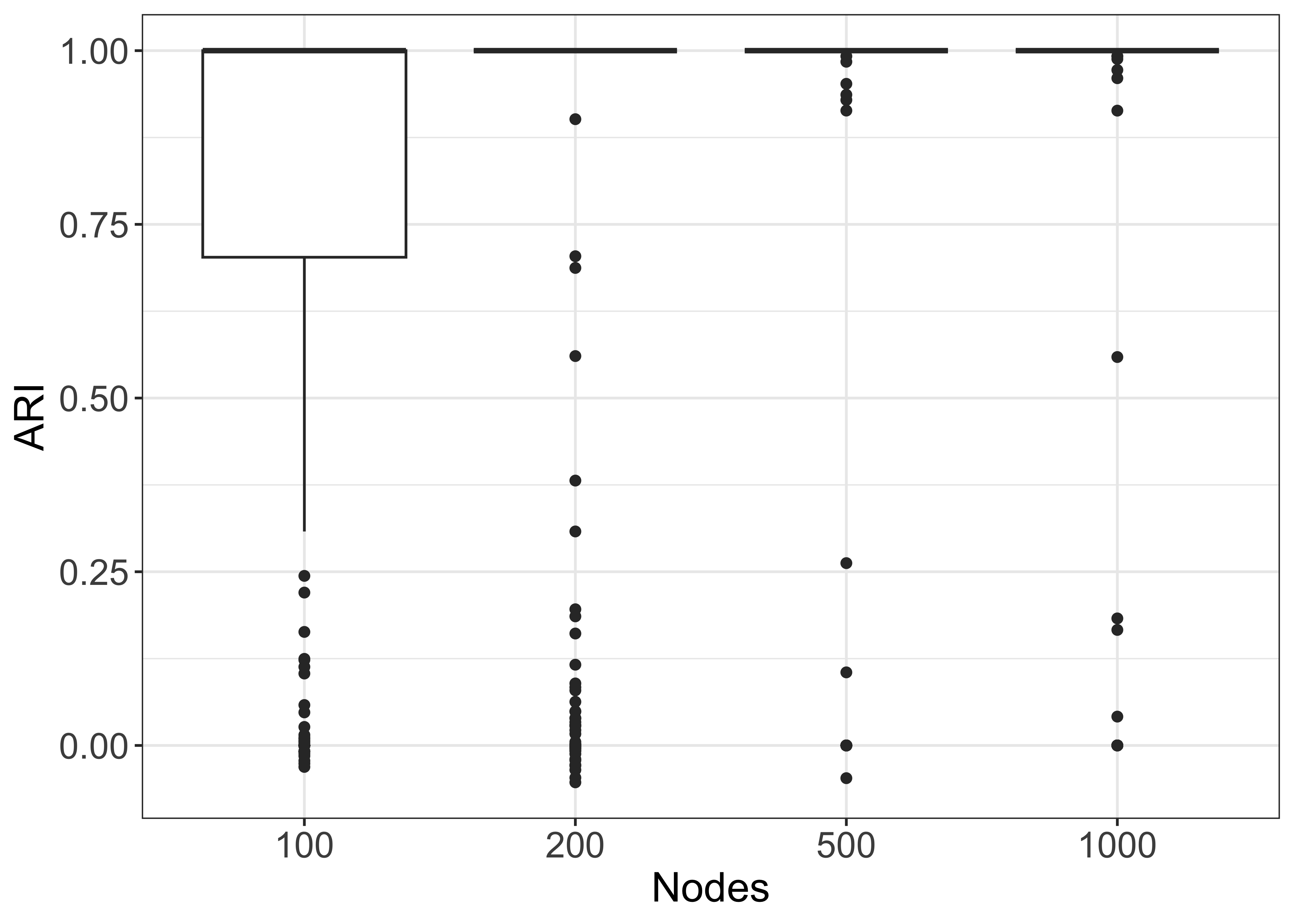}
     \caption{Community recovery, varying the number of nodes.}
     \label{fig_ocd:recov_n}
    \end{subfigure}
    \hfill
    \begin{subfigure}{0.45\textwidth}
     \centering
     \includegraphics[width=\textwidth]{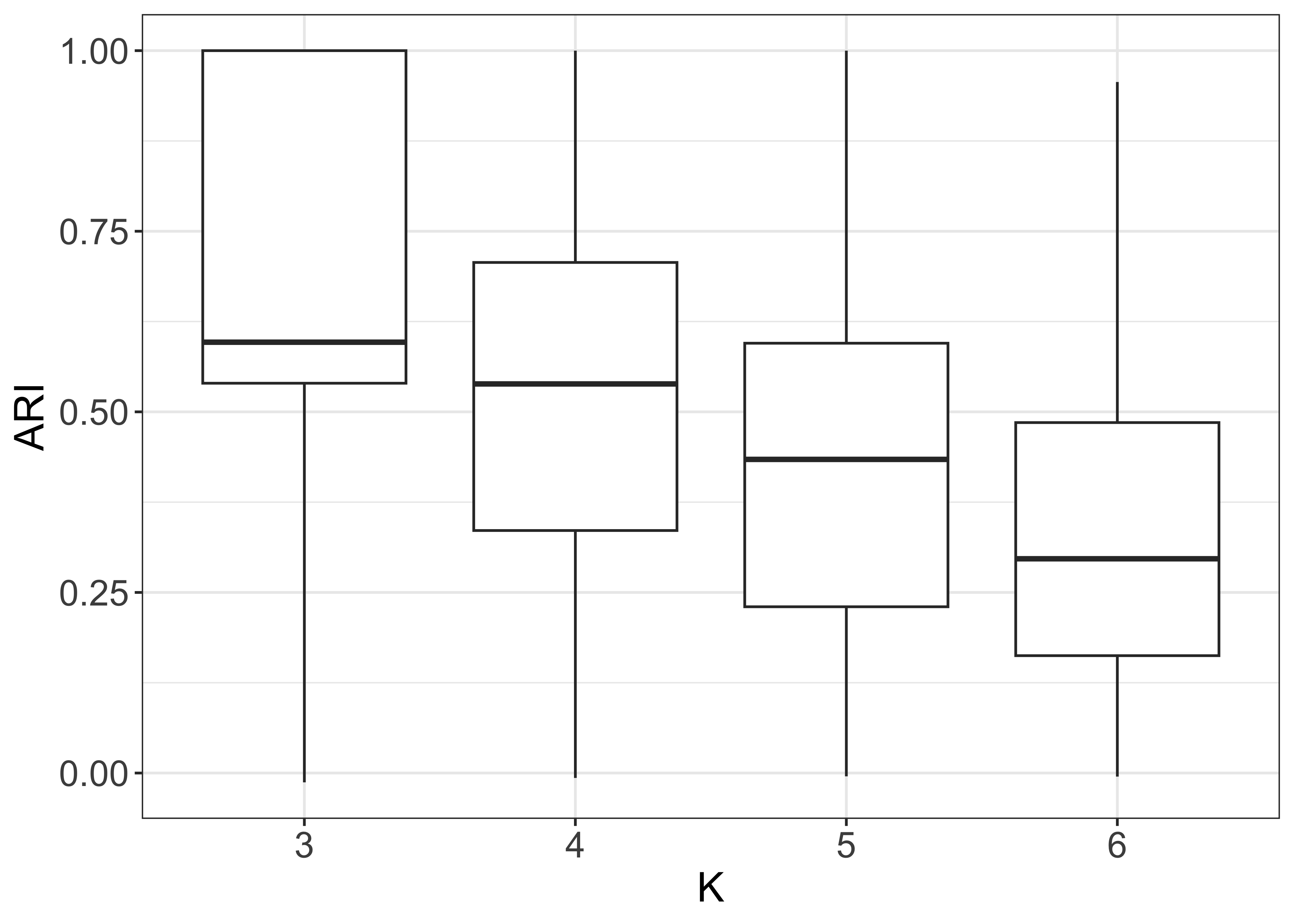}
     \caption{Community recovery while varying $K$, the number of communities.}
     \label{fig_ocd:recov_k}
    \end{subfigure}
    \caption{For constant network parameters, the performance of our procedure 
    for community recovery improves as the number of nodes in the network increases. Similarly, as
    we increase the number of communities, recovery becomes more challenging.}
    \label{overall-label}
\end{figure}

\paragraph{Online Community Recovery.}
Given that we iteratively update our estimates 
as we observe events, we can also examine how the performance
of our estimation scheme evolves over time. For example, we can investigate
how quickly the community structure, $\bm{z}$, is captured as events on the 
network are observed.
The online recovery of the other parameters, $\theta$, 
is discussed below.
Figure~\ref{fig_ocd:online_community}
illustrates this for the same simulation setting considered previously. 
As we observe and process the data 
in an online fashion we store 
$$
\bm{\hat{\tau}}^{(m)}= (\hat{\tau}_{1}^{m},\ldots, \hat{\tau}_{n}^{m}),
$$
the estimated node labels for the $m$-th window. We then compute
the ARI between the true communities and the estimates at these intermediate
time points. Figure~\ref{fig_ocd:online_community} shows this 
performance for a range of network sizes, 
{\color{black} 
where we have $K=2$ communities with 40\% and 60\% of the nodes in
each community respectively,
keeping all other model parameters as previously specified.}
Unsurprisingly, we see that as the number of nodes in the 
network increases, the estimated community structure quickly agrees 
with the true assignments. While all network sizes show a large degree of
variation in the ARI between the estimated and true labels initially, 
this decreases quickly as the number of nodes increases and
as we observe more events. 
In particular, for sufficiently large 
networks we are able to recover the community 
assignments well having observed only a quarter of all events.

\begin{figure}[ht]
  \centering
  \includegraphics[width=0.65\textwidth]{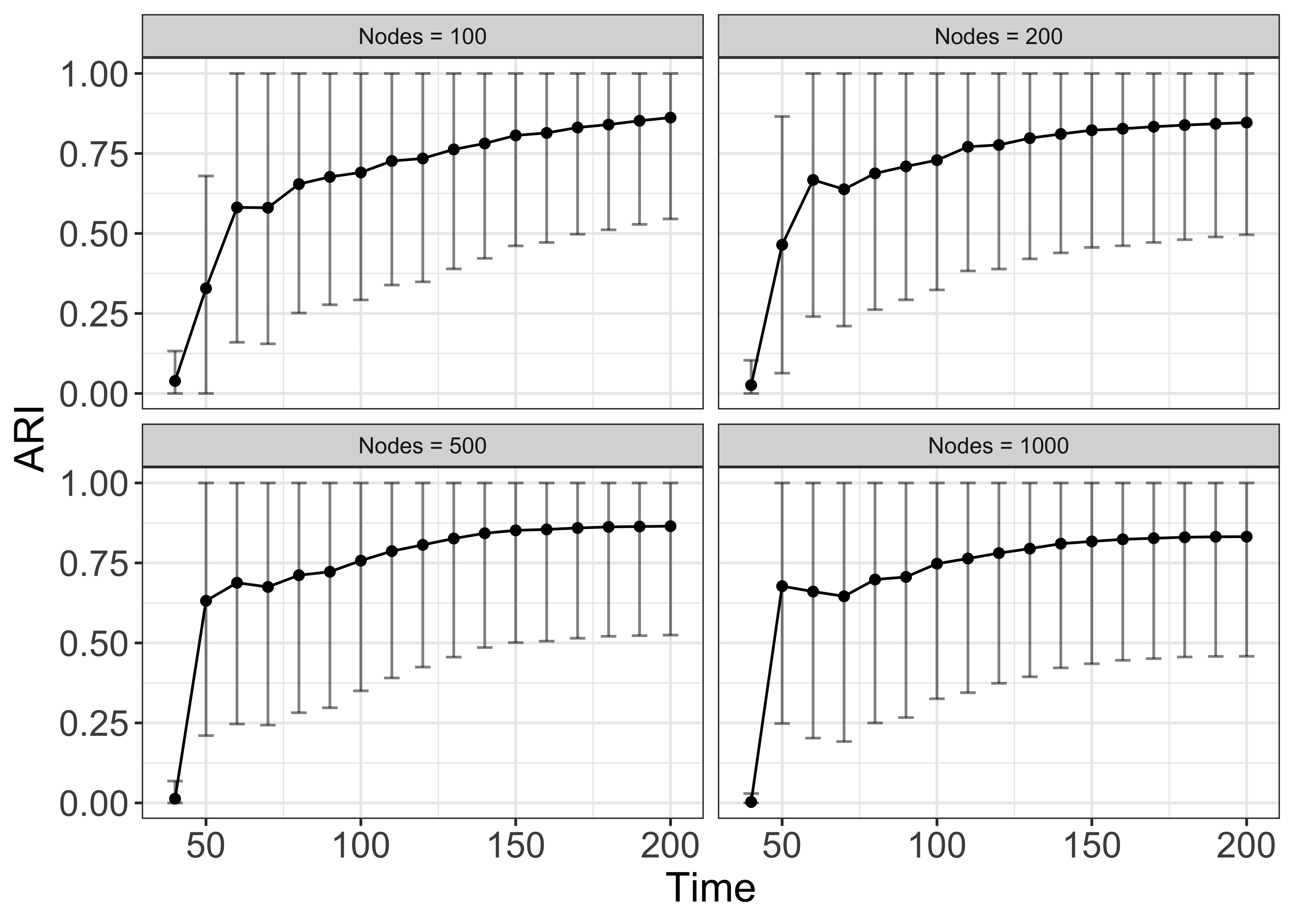}
  \caption{Demonstrating the ARI of the estimated clustering as we
  observe events in time, for varying network sizes. Here we show the mean ARI 
  with error bars corresponding to one standard deviation. We see that 
  for all network sizes, we can identify community structure,
  with the larger networks doing so quickly and with 
  less variability.}
  \label{fig_ocd:online_community}
\end{figure}

\paragraph{Monitoring Convergence.}
A natural question in variational inference is 
how to identify whether the model has converged and whether 
it has converged to a local optima. For coordinate ascent variational 
inference convergence can be 
assessed by monitoring the ELBO, identifying when the change
in this quantity from the previous
iteration of the coordinate ascent scheme falls below some threshold 
\citep{blei2017variational}.
As we are observing the data sequentially here we cannot use this metric 
to assess convergence in practice. The ELBO
will decrease as we observe new events.

Were the total observation period and all event times known 
in advance, we could 
compute the ELBO for the full data set, using the estimates 
we observe in an online manner (i.e, using only the data up to the current 
time point to form the parameter estimates, but then
computing the ELBO for all data). 
In what follows, this quantity will be denoted the \textit{full ELBO}.
An advantage of the online procedure is that 
we obtain good parameter estimates using only a small number 
of events, compared to batch estimates which must use all events
repeatedly. We illustrate the 
typical performance of this metric for several simulated datasets
in 
Figure~\ref{fig_ocd:elbo_events},
{\color{black} using the parameter settings initially specified for all
simulations}. 
We show the full 
ELBO against the percentage of events used to compute
the corresponding
parameter estimates. Similarly, we show the convergence 
of the batch procedure, which repeatedly uses all events to 
obtain a similar optimum. In each case, the online procedure quickly converges to
the same maximum as the batch method, only requiring one pass through all 
data. The batch method requires using all data multiple times.

\begin{figure}[ht]
  \centering
  \includegraphics[width=0.65\textwidth]{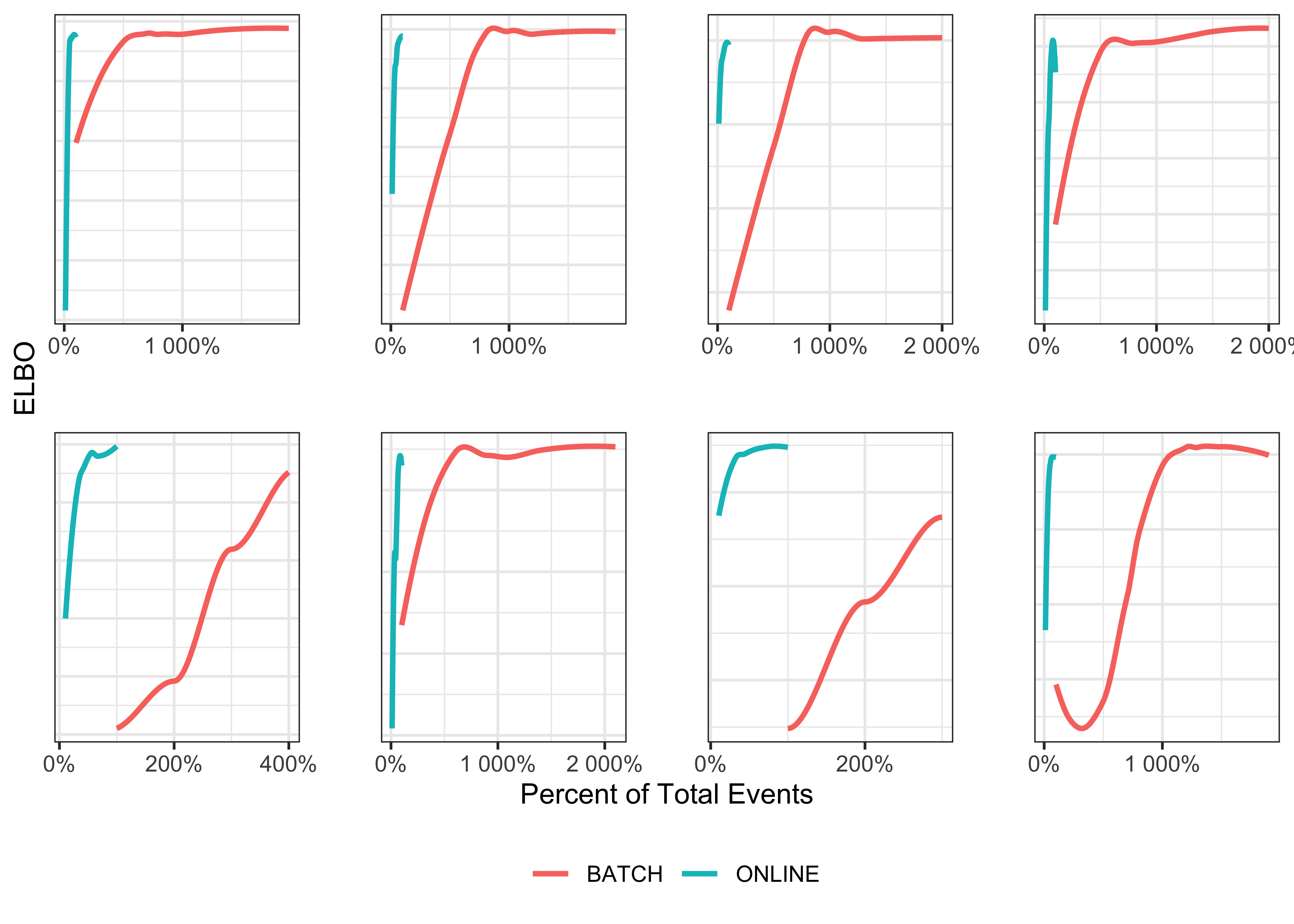}
  \caption{The full ELBO plotted against the percent of 
  total events (on the log scale)
  used to obtain the corresponding parameter estimates, for 8 simulated datasets.
  We see that the online procedure obtains good estimates of the
  full ELBO using all of the events at most once.}
  \label{fig_ocd:elbo_events}
\end{figure}

\paragraph{Parameter Recovery.}
We also assess our ability to recover the parameters of our model in 
an online fashion. {\color{black}For a block inhomogeneous} Poisson
model 
{\color{black} at time $t$ the rate function between communities $k_1$ and $k_2$
is given by $B_{k_1 k_2}(t) \coloneqq\sum_{h}a_{k_1 k_2}(h) f_h(t)$.
We monitor the relationship between $\hat{B}(t)$
and $B(t)$ across all community pairs $k_1,k_2$ by computing
$\frac{1}{H K^2}\vert\sum_{ijh}\hat{a}_{ij}^{(m)}(h)- \sum_{ijh}a_{ij}(h)\vert$,
to account for possible permutation of the node labels. 
}
{\color{black}Here we consider the default experiment setting considered
above, where we now vary the total observation time $T$, from
$T=50$ to $T=500$. We wish to examine how increased observation time
effects parameter recovery. We store the parameter $\hat{a}_{k_1 k_2}$
estimates as we process the data in an online fashion 
and use these to track the relationship between $\hat{B}$ and $B$
as we observe data for larger values of $T$.}
Figure~\ref{fig_ocd:param_recov} shows 
{\color{black}GAM} smoothed estimates and standard errors 
of this quantity across simulations
as we increase $T$.
This difference shrinks quickly across all values of $T$.

\begin{figure}[ht!]
    \centering
    \begin{subfigure}{0.45\textwidth}
     \centering
     \includegraphics[width=\textwidth]{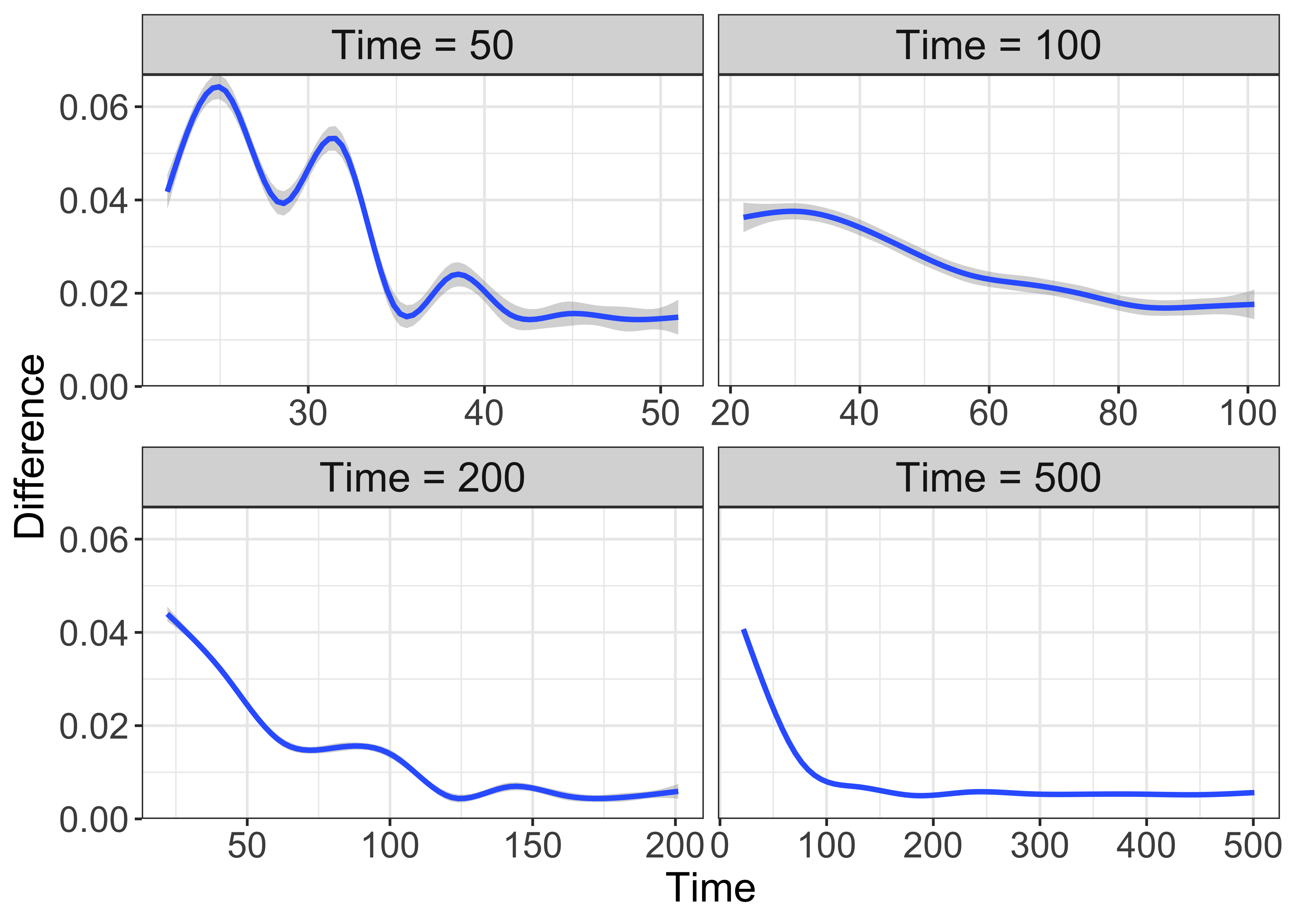}
     \caption{Recovery of rate matrix using online procedure.}
     \label{fig_ocd:param_recov}
    \end{subfigure}
    \hfill
    \begin{subfigure}{0.45\textwidth}
     \centering
     \includegraphics[width=\textwidth]{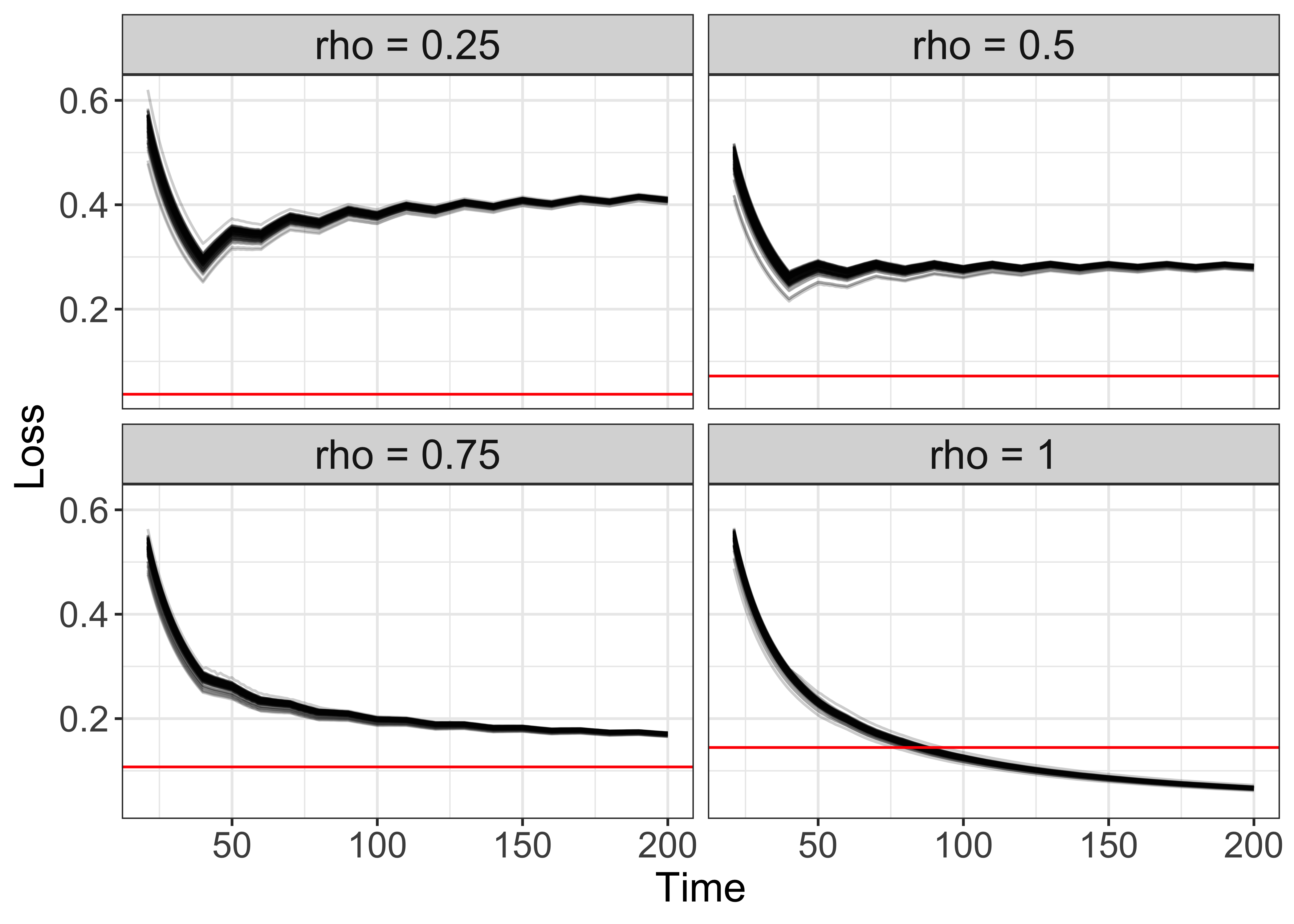}
     \caption{Average cumulative loss for our online inference procedure.}
     \label{fig_ocd:online_loss}
    \end{subfigure}
    \caption{Experiments demonstrating the online recovery of point process 
    parameters (a) for a fixed network setting as we increase the observation time,
    and (b), the corresponding average cumulative online loss, as we vary
    the sparsity of the underlying network. Both
    quantities decrease as we observe events and update our online estimates,
    with the online loss performing well as the density of the network increases.}
    \label{overall-label-2}
\end{figure}

\paragraph{Online Loss.}
A natural task when considering online learning is the ability to make
online predictions and look at the loss of such a procedure 
over time. We consider such a scenario here for event data 
from a block model. In particular, we will use the negative 
log likelihood as our loss function. Given parameter 
estimates from the $(m-1)$-th time window, the negative 
log likelihood for the next window of observed data is
given by
\begin{equation*}
  \ell_{m}(\theta) \vcentcolon= - l_{m}(\theta^{(m-1)}\vert z^{(m-1)}) = 
  - 
 \int_{(m-1)dT}^{m dT} \log \lambda^{(m-1)}_{ij}(t\vert z)dN_{ij}(t)  
 + \int_{(m-1)dT}^{m dT} \lambda^{(m-1)}_{ij}(t \vert z)dt.
\end{equation*}

We then define the average cumulative loss after $m$ observation batches 
as 
$\frac{1}{m}\sum_{i=1}^{m}\ell_{i}(\theta)$. We wish to compare this 
quantity to the best average cumulative loss, without learning the model 
in an online fashion. To do this, we determine the best overall 
batch estimate, 
using all events repeatedly, which we denote 
as $\hat{\theta}, \hat{z}$, and use these estimates to compute 
$\bar{l} \vcentcolon = \frac{1}{M}\sum_{m=1}^{M} - l_{m}(\hat{\theta}\vert\hat{z})$,
the best average cumulative loss in hindsight.
Note that here we are computing the average cumulative 
loss using both the current estimates of $\hat{\theta}$,
the parameters of the point process and also $\hat{z}$, 
the current estimates of the latent community assignments.

{\color{black}For the default simulation setting described previously,
we consider computing these quantities,}
varying $\rho$, the 
{\color{black} sparsity of the network point process}.
Each time we compute the online loss for the observations
in the subsequent window of length $dT = 1$. We show this in 
Figure~\ref{fig_ocd:online_loss}, where each black line denotes 
the average cumulative loss over time for one 
run of our online inference scheme. $\bar{l}$ is shown with the 
red horizontal line. 
We see that in sparse network settings, our method is 
unsurprisingly outperformed by the best average cumulative
loss in hindsight. However, as the density of the network
increases, we obtain similar loss in the online setting,
and our online procedure is actually better in a fully connected network.

\subsection{Experiments on Real Data}
\label{Realdata}

To evaluate our online algorithms on real data, 
we consider the problem of link prediction, using large temporal 
networks from the literature.
We consider three such networks, available from the Stanford Large Network Dataset 
collection \citep{leskovec2014snap}. They consist of the timestamps of:
\begin{itemize} 
    \item A collection of emails sent by users in a large university. 
    This consists of 300k emails between approximately 1000 users over 803 days.
    \item Messages sent between 2000 students on an online college social 
    network platform over 193 days, consisting of 60k messages. 
    \item Interactions from the Math Overflow website over 2350 days. 
    Here we have 25k users and 500k directed interactions, where an 
    interaction from user $i$ to user $j$ means that user $i$ responded 
    to a question posed by user $j$.
\end{itemize}

The temporal component in these networks changes over the observed time, 
with interactions much sparser towards the end of the observed time period. 
This makes link prediction a challenging problem in this setting.
For each of these networks, we fix $K$, the number of communities, 
based on knowledge of the network structure, as we aim to compare 
link prediction for a given $K$. We use $K$ as considered elsewhere 
for these examples \citep{miscouridou2018modelling}. 
We partition the events into training and test periods
which contain 85\% and 15\% of events respectively. 
Note that we consider the edge structure, $A$, known in advance, 
although we could also learn this from the training data 
and use that as our estimate of the overall edge list. 
Given the events observed initially, the goal is to 
predict the number of events that will occur between 
a directed pair over the test period.

To fit these models, we consider $dT$ such that
$M = \frac{T}{dT}\approx 100$ for the online estimators, 
with the same maximum number of iterations for our corresponding batch versions. 
For the inhomogeneous models, we consider $7$ step functions as our basis 
functions, aiming to capture day of the week effects present in our event 
streams. We take the average of these basis functions as an estimate our 
baseline rate. The results for this link prediction problem are shown in 
Table~\ref{table:real data}, with the corresponding 
computation times (in seconds) shown in Table~\ref{table:real_time}. 
Our online procedure obtains comparable estimates to more expensive batch 
estimates, and is better suited to estimation for the large networks 
considered here, obtaining comparable predictions generally quicker.

\begin{table}[ht]
  \centering
  \begin{tabular}{c |  c c c} 
   \hline
   \hline
   Method & Email & College & Math \\ [0.5ex] 
   \hline
   Poisson            & 11.73/12.9  & 5.16/13.96  & 2.13/1.99 \\
   Hawkes             & 19.42/12.74 & 5.32/5.09   & 2.06/2.14 \\
   In-Poisson ($H=7$) & 15.09/18.92 & 5.57/5.67   & 2.14/2.14 \\
   In-Hawkes ($H=7$)  & 14.84/12.9  & 5.58/5.44   & 2.14/2.14 \\
   \hline
  \end{tabular}
  \caption{Median RMSE of predicted event counts vs true event counts 
  in held out test set across 50 simulations. Online/Non-online estimates.}
  \label{table:real data}
  \end{table}

  \begin{table}[ht]
  \centering
  
  \begin{tabular}{c |  c c c} 
   \hline
   \hline
   Method & Email & College & Math \\ [0.5ex] 
   \hline
   Poisson            & 0.7/24.6  & 0.1/1.5  & 7.4/4.5     \\
   Hawkes             & 1.4/14.9  & 1.5/0.4  & 235.9/314.1 \\
   In-Poisson ($H=7$) & 2.0/20.4  & 1.9/5.1  & 257.7/28.1  \\
   In-Hawkes ($H=7$)  & 2.3/52.4  & 2.1/4.5  & 253.5/327.1 \\
   \hline
  \end{tabular}
  \caption{Median computation time for Online/Full
  Model fitting (seconds)
  across 50 simulations.}
  \label{table:real_time}
  \end{table}

\section{Discussion and Extensions}
\label{Discussion}
In this paper we propose a scalable online framework for 
learning the structure of large networks which are observed in the 
form of event streams between nodes in this network.
We develop a scalable online algorithm to 
estimate network models for this data, which
uncover community structure 
using point process models on the network, considering both computational 
speed and memory requirements. 
In both simulations and experiments utilising real data,
we observe that 
our method is scalable compared with batch methods,
especially for large networks 
where both $n$, the number of nodes and $T$, the total time for which 
events are observed, grow. 
We also provide theoretical results regarding the proposed 
online estimation procedure, in terms of convergence, regret and community
recovery, under mild conditions.

There are many ways this work could be extended. There are several aspects of 
community detection which we have not addressed. 
Further investigation could
indicate better methods for initializing algorithms
of this form in the online setting, or consider incorporating 
more heterogeneous network models \citep{zhao2012consistency, sengupta2018block}.
Similarly, selecting the number of communities is an important problem in these 
models and it is not immediate how to approach this with an online algorithm.
Our algorithm also assumes that the edge structure $A$ does not vary in time 
and it is of interest to consider a model where $A$ can also evolve over time. 
It is of interest to also estimate $A$ in an online 
fashion, along with deriving properties of estimators for this updated and
more challenging model.

One natural setting where online estimation procedures for networks
have meaningful applications concern identifying changes in the 
structure of networks. For example, social and computer networks 
are frequently the target malicious actors, aiming to disrupt the
nodes in a network. Previous work has considered this question in 
discretized time. \citet{heard2010bayesian} proposed a Bayesian model for anomaly detection,
to first identify a subset of nodes whose communication patterns changed
and then examine flagged nodes directly. This method can
identify anomalous nodes in discretized real time. A similar framework is proposed 
by \citet{lee2022anomaly}, which identifies departures from 
a fitted dynamic logistic model as anomalous.
Alternatively, changepoint models have been proposed for discretized 
dynamic networks. Under a Stochastic Block modeling framework, 
\citet{bhattacharjee2020change} describes identifying 
a changepoint in this model in the offline setting.
However, there is no existing work in the context of continuous 
time event data on networks, or to identify changepoints 
in an online fashion. The online procedure considered here
provides one potential avenue to consider these problems, 
with the aim of identifying anomalous nodes and changes 
in the network in an online fashion. 
{\color{black}For the class of models in this paper, 
changes could be in terms of the blockwise conditional 
intensity and or the community assignment of the nodes.
As highlighted by 
\citet{matias2017statistical}, it can be challenging to
identify changes in a model with community structure, with 
restrictions required on how the node community assignments can vary. We believe this
is a fruitful direction for future work.

An alternative challenge with data of this form would be to allow nodes to enter
or leave the network during the observation period. We do not consider such a 
possibility here, but this is an important next step for such models.
}

We also wish to point out the potential connections
between the framework we propose here  
and other popular longitudinal models for network
data (e.g. the dynamic latent space model 
\citep{sewell2015latent}, the temporal exponential random graph model 
\citep{leifeld2018temporal}, and the varying coefficient model for dynamic 
networks \citep{lee2017varying}) which can be viewed as the discrete 
time event processes. With suitable modifications, our results can be 
incorporated into
these related settings, and could be use to scale these methods 
to further large network data.

\printbibliography

\newpage



\begin{appendices}

\setcounter{figure}{0}
\setcounter{table}{0}
\setcounter{equation}{0}
\setcounter{theorem}{0}
\makeatletter
\renewcommand{\theequation}{S\arabic{equation}}
\renewcommand{\thefigure}{S\arabic{figure}}
\renewcommand{\thelemma}{S\arabic{lemma}}
\renewcommand{\thetheorem}{S\arabic{theorem}}

\section{Algorithm Details}

We include Algorithm~\ref{alg:hawkes} for the online Hawkes process as 
mentioned in the main text, along with Algorithm~\ref{alg:trim}, 
which is a key step for storing useful information in this procedure.
Some supporting functions used in Algorithm~\ref{alg:hawkes} are given below.
\begin{itemize}
    \item $a+=b$ represents $a = a+b$; $a-=b$ represents $a = a - b$.
    \item Formula for $impact(t)$ is 
    $\sum_{t_1 \in timevec} \lambda \exp\{- \lambda (t - t_1)\}$.
    \item Formula for $I_1$ is  
    $\sum_{t_1 \in timevec} \exp\{- \lambda (t - t_1)\}$.
    \item Formula for $I_2$ is  
    $\sum_{t_1 \in timevec} (t - t_1) \lambda \exp\{- \lambda (t - t_1)\}$.
    \item Formula for $integral(t,t_{end},\lambda)$ is 
    $ 1 - \exp\{-\lambda (t_{end} - t)\}$.
    \item Formula for $integral(t,t_{start},t_{end},\lambda)$ is 
    $ \exp\{-\lambda (t_{start} - t)\} - \exp\{-\lambda (t_{end} - t)\}$.
\end{itemize}

\begin{algorithm*}[t]
	\caption{Online-Hawkes}
	\begin{algorithmic}[1]
	    \STATE Input: $data$, number of groups $K$, window size $dT$, edge list $A$.
	    \STATE Output: $\hat \mu$, $\hat B$, $\hat \lambda$, $\hat \pi$.
	    \STATE Initialization: $S$, $\tau$, $\pi$, $B$, $\mu$, $\lambda$.
	    \STATE Set $M = T/dT$ and create an empty map $\mathcal D$.
		\FOR{window $m = 1$ to $M$}
		\STATE Read new data between $[(m-1)\cdot dT, m \cdot dT]$ and apply \textbf{Trim}.
		\STATE Create temporary variables: $\mu_{p1}$, $\mu_{p2}$, $B_{p1}$, $B_{p2}$, $S_p$.
		\STATE Set learning speed: $\eta = \frac{K^2}{\sqrt{m} m_t}$, 
    where $m_t$ is the number of events between $[(m-1)\cdot dT, m  \cdot dT]$.
		\FOR{key $(i,j)$ in $\mathcal D$}
		\STATE Create sub temporary $K$ by $K$ matrix variables: $\mu_{p1,tp}$, $B_{p1,tp}$, $B_{p2,tp}$, $S_{p,tp}$ and $\lambda_{st}$.
		\STATE Update $\mu_{p2}$ by setting $\mu_{p2}(k,l) \mathrel{{+}{=}} \tau_{ik} \tau_{jl} dT$ for $k,l \in [K]$.
		\STATE Update $S_p$ by setting $S_p(i,k) \mathrel{{-}{=}} \tau_{jl} \mu_{kl} dT$.
		\STATE Get time stamps, $timevec$, corresponding to $(i,j)$.
		\FOR {$t$ in $timevec$}
		\IF{$t > (m-1)dT$ }
		\STATE Compute the impact function value, $impact(t)$.
		\STATE Compute $I_1$ and $I_2$.
		\STATE Compute $\Lambda$, where $\Lambda(k,l) = \mu_{kl} + B_{kl}~impact(t)$.
		\STATE $\lambda_{st} \mathrel{{+}{=}} B \cdot (I_1 - I_2)/\Lambda - B\cdot (T_e - t) \exp\{-\lambda (T_e - t) \}$.
		\STATE $\mu_{p1,tp}(k,l) \mathrel{{+}{=}} 1/ \Lambda(k,l)$.
		\STATE $B_{p1,tp}(k,l) \mathrel{{+}{=}} impact(t)/ \Lambda(k,l)$.
		\STATE $S_{p,tp}(k,l) \mathrel{{+}{=}} \log(\Lambda(k,l))$.
		\STATE $B_{p2,tp}(k,l) \mathrel{{+}{=}} integral(t,t_{end},lam)$.
		\ENDIF
		\IF{$t \leq (m-1)dT$}
		\STATE $B_{p2,tp} \mathrel{{+}{=}} integral(t, t_{start}, t_{end},lam)$.
		\STATE $\lambda_{st} \mathrel{{+}{=}} B_{kl} (T_s - t) \exp\{- \lambda (T_s - t)\} - (T_e - t) \exp\{- \lambda (T_e - t)\}$.
		\ENDIF
		\ENDFOR
		\STATE $\mu_{p1}(k,l) \mathrel{{+}{=}} \tau_{ik} \tau_{jl} \mu_{p1,tp}(k,l)$.
		\STATE $B_{p1}(k,l) \mathrel{{+}{=}} \tau_{ik} \tau_{jl} B_{p1,tp}(k,l)$.
		\STATE $B_{p2}(k,l) \mathrel{{+}{=}} \tau_{ik} \tau_{jl} B_{p2,tp}(k,l)$.
		\STATE $S_{p}(i,k) \mathrel{{+}{=}} \sum_l \tau_{jl} (S_{p,tp}(k,l) - B_{kl} B_{p2,tp}(k,l))$.
		\ENDFOR
		\STATE $S \mathrel{{+}{=}} S_p$.
		\STATE Compute the negative gradients:
		$grad_B = B_{p1} - B_{p2}$, $grad_{\mu} = \mu_{p1} - \mu_{p2}$, $grad_{\lambda} = \sum_{kl} \tau_{ik}\tau_{jl}\lambda_{st}(k,l)$.
		\STATE Update parameters: $B = B + \eta \cdot grad_B$, 
		 $\mu = \mu + \eta \cdot grad_{\mu}$, $\lambda = \lambda + \eta \cdot grad_{\lambda}$.
		\STATE Update $\tau$ by setting $\tau_{ik} = \frac{\pi_k S_{ik}}{\sum_{k}\pi_k S_{ik} }$ for $i \in [n]$ and $k \in [K]$.
		\STATE Update $\pi$ by setting 
		$\pi_k = \frac{1}{n} \sum_i \tau_{ik}$
		for $k = 1,\ldots,K$.
		\ENDFOR
	\end{algorithmic} 
	\label{alg:hawkes}
\end{algorithm*} 

\begin{algorithm*}[t]
	\caption{Trim}
	\begin{algorithmic}[1]
	    \STATE Input: $\mathcal D$, truncated length $R$, current time $t_{current}$, $data_{new}$. 
	    \STATE Output: $\mathcal D$.
		\FOR{$event$ in $data_{new}$}
		\STATE Get node pair $(i,j)$ and time stamp $t$.
		\IF{key $(i,j)$ is already in $\mathcal D$}
		\STATE We get the corresponding queue. We then push $t$ at the back of this queue and update $\mathcal D$.
		\ENDIF
		\IF{key $(i,j)$ does not exist in $\mathcal D$}
		\STATE We create an empty queue, push $t$ to it and update $\mathcal D$.
		\ENDIF
		\ENDFOR
		\FOR{key $(i,j)$ in $\mathcal D$}
		\STATE Get the queue $timeque$ corresponding to key $(i,j)$ and let $t_{front}$ be the first element of $timeque$.
		\WHILE{$t_{current} - t_{front} > R$}
		\STATE Pop the first element of $timeque$.
		\STATE Set $t_{front}$ be the first element of current $timeque$.
		\ENDWHILE
		\ENDFOR
	\end{algorithmic} 
	\label{alg:trim}
\end{algorithm*} 

As discussed in main paper, we only need to store the 
sufficient statistics of the particular model in each setting. 
We show two examples in Table~\ref{table:storage}. 
In the homogeneous Poisson setting, we only need to store the cumulative 
counts for each pair of sender and receiver ($l_{user1, user2}$). 
In the Hawkes setting, we only need to store the recent historical events 
since the old information decays exponentially fast and thus has 
vanishing impact on the current event.   

\begin{table}[h]
\caption{The Data Structure for Storing History Events. The left diagram 
shows the structure under Poisson model, where the key is the pair of nodes 
and the value is its corresponding cumulative number of all past events. 
The right diagram shows the structure under Hawkes model, where the key is 
still the nodes and the value is its corresponding time sequence 
between $t_{current} - R$ and $t_{current}$ stored in \textbf{queue} structure.}
\label{datastrcutre}
\centering
\begin{tabular}{c | c}
    \multicolumn{2}{c}{Poisson} \\
    \hline
    Key & Value \\
    \hline
   (User1, User3) & $l_{user1, user3}$ \\
   (User3, User8) & $l_{user3, user8}$ \\
   (User3, User1) & $l_{user3, user1}$ \\
   (User2, User4) & $l_{user2, user4}$ \\
   (User3, User5) & $l_{user3, user5}$ \\
   $\vdots$ &  $\cdots$ \\
   (User5, User3) & $l_{user5, user3}$ \\
   (User8, User3) & $l_{user8, user3}$ \\
   (User9, User2) & $l_{user9, user2}$ \\
   (User7, User1) & $l_{user7, user1}$ \\
   \hline
\end{tabular}
\quad
\bigskip
\begin{tabular}{c | c}
\multicolumn{2}{c}{Hawkes} \\
\hline
Key & Value \\
\hline
(User1, User3) & $t_{user1,user3}^{(start)}$, \ldots , $t_{user1,user3}^{(end)}$ \\
   (User3, User8) & $t_{user3,user8}^{(start)}$, \ldots , $t_{user3,user8}^{(end)}$ \\
   (User3, User1) & $t_{user3,user1}^{(start)}$, \ldots , $t_{user3,user1}^{(end)}$ \\
   (User2, User4) & $t_{user2,user4}^{(start)}$, \ldots , $t_{user2,user4}^{(end)}$ \\
   (User3, User5) & $t_{user3,user5}^{(start)}$, \ldots , $t_{user3,user5}^{(end)}$ \\
   $\vdots$ &  $\cdots$ \\
   (User5, User3) & $t_{user5,user3}^{(start)}$, \ldots , $t_{user5,user3}^{(end)}$ \\
   (User8, User3) & $t_{user8,user3}^{(start)}$, \ldots , $t_{user8,user3}^{(end)}$ \\
   (User9, User2) & $t_{user9,user2}^{(start)}$, \ldots , $t_{user9,user2}^{(end)}$ \\
   (User7, User1) & $t_{user7,user1}^{(start)}$, \ldots , $t_{user7,user1}^{(end)}$ \\
   \hline
\end{tabular}
\label{table:storage}
\end{table}

\section{Additional Simulation Results and Details}
\label{appendix:add_sims}

Here we 
{\color{black} provide additional details regarding the simulations in the main
text, and we also
include additional simulations 
experiments which were omitted from the main manuscript.
Alongside this, we provide the results of similar experiments in the  
Hawkes process setting, similar to those considered in the 
main text for the inhomogeneous Poisson model. The code used to create all 
results in this work is available at \url{https://github.com/OwenWard/OCD_Events}}.
We discuss the regret properties of our online procedure.
We demonstrate community recovery and other properties of our
online inference procedure for Hawkes 
process models. We also 
expand on some components of the 
inference procedure discussed in Section~\ref{Simulation}.

{ \color{black}
\paragraph{Illustrative simulation in Section 1.}
We first provide exact simulation settings for the
illustrative example in the introduction.
We consider a network of $n=100$ nodes, with 2 communities, 
with $40\%$ of nodes in the first community.
In particular, we consider an intensity function from nodes in 
group $k_1$ to nodes in group $k_2$ of the form

\begin{equation*}
    a_{k_{1}k_{2}}(t) = a_{k_1 k_2}^{1} \mathbbm{1}\{0 \leq t < T/3\} \\
    + a_{k_1 k_2}^{2} \mathbbm{1}\{T/3 \leq t < 2T/3\} +
a_{k_1 k_2}^{3} \mathbbm{1}\{2/3T \leq t < T\}
\end{equation*}

where the coefficient vectors 
$\bm{a_{k_1 k_2}}=(a^1_{k_1 k_2},a^2_{k_1 k_2},
a^3_{k_1 k_2})$ 
for each
block pair 
are of the form

$$
\begin{pmatrix}
  \bm{a_{11}}\\
  \bm{a_{12}} \\
  \bm{a_{21}} \\
  \bm{a_{22}}
  \end{pmatrix}
= 
\begin{pmatrix}
  0.25 & 0.5 & 1\\
  0.75 & 1 & 1 \\
  0.5 & 0.25 & 1 \\
  1 & 0.75 & 1
  \end{pmatrix}.
$$

This leads to a dense network, where we observe events between each node pair for the specified
choice of $T=100$.
If we wish to identify the community structure, classical network models require a single
adjacency matrix, encoding the relationship between each node pair. The simplest way to 
form such a matrix is to consider the count matrix with $A_{ij}=N_{ij}(T)$, the number of
events observed between a given node pair. Spectral clustering of this adjacency matrix 
would then lead to estimated community memberships for each node. However, for this choice of conditional
intensity function the event counts do not preserve the underlying community structure.

Rather than aggregating this data to form a single adjacency matrix, an alternative approach to cluster
network incorporates dynamics through the observation of a sequence of adjacency matrices
at some (equally spaced) time intervals. 
If we wanted to apply such a method to this event data the challenge is how to form the
adjacency matrices. An observation window must be chosen, with an edge between two nodes present 
for the corresponding adjacency matrix for that window if there is one (or more) events between them.
}

{ \color{black}
\paragraph{Simulation settings}
We first expand on the simulation settings used in Section~\ref{Simulation}.
Unless otherwise specified we consider a network of $n=200$ nodes and
$K=2$ equally sized communities. The block inhomgoeneous Poisson process model 
is given by 

$$
\lambda_{ij}(t) = \sum_{h=1}^{H=2}a_{z_i z_j}(h)f_h(t).
$$
Here we consider $f_h(t)$ to be step functions of common fixed length.
We consider the following coefficients for these basis functions
$$
\begin{pmatrix}
  \bm{a_{11}}\\
  \bm{a_{12}} \\
  \bm{a_{21}} \\
  \bm{a_{22}}
  \end{pmatrix}
= 
\begin{pmatrix}
  0.8 & 0.4 \\
  0.6 & 0.2 \\
  0.2 & 0.7 \\
  0.4 & 0.7
  \end{pmatrix}.
$$

When we vary $K$ we consider different coefficients $K$ which are multiplied by a constant which
varies with the community. Full details of the choices of parameters are given in 
the associated code repository.

}

\paragraph{Window Size}
Throughout the simulation studies 
in Section~\ref{Simulation} we have used a fixed window size
such that $dT=1$. Here we wish to investigate the 
effect that varying this window size has on the performance of our 
algorithm. We compare community recovery for varying window
sizes from $dT=0.25$ to $dT=5$,
{\color{black} keeping all other parameters fixed as in the default simulation setting.
We show boxplots of the community recovery performance, in terms of ARI, in Figure~\ref{fig_ocd:window}.
We see that in this scenario community recovery is possible, for all window sizes considered.
For each choice of $dT$ almost all simulations correctly recover the communities.
We note that the choice of appropriate $dT$ will depend somewhat on the data considered.
For sufficiently small $dT$ we might not observe any events in a given window, which
would not provide any update of the model parameters. Similarly, a $dT$ leading to a single event
in each window would mirror standard Stochastic Variational inference 
\citep{hoffman2013stochastic}. It appears that $dT$ should be small enough to avoid getting
stuck in local optima of the current estimate of the ELBO, while being large enough to ensure events have been
observed.

We also wish to investigate the role of $dT$ in the computation time required for this
procedure. In Figure ~\ref{fig_ocd:dt} we show the computation time for our inference 
procedure (in seconds) as we vary $dT$, the window size. It appears that the 
window size is not clearly related to the computation time, however we would recommend against 
extremely small values of $dT$, which may lead to too many windows where no events occur.
} 


\begin{figure}[ht]
\centering
\includegraphics[width=0.75\textwidth]{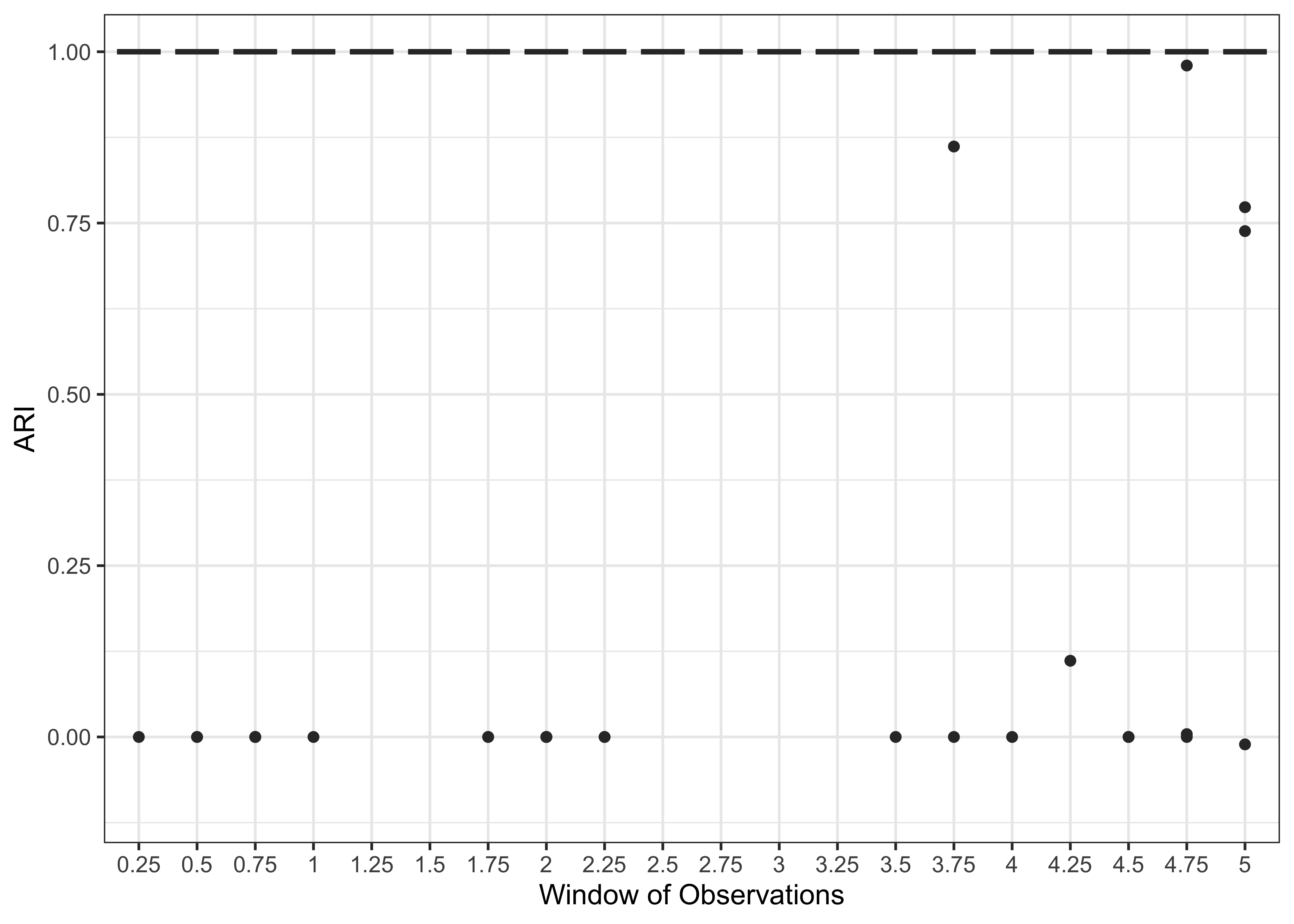}
\caption{Community recovery under block inhomogeneous Poisson simulated data for varying 
window size.}
\label{fig_ocd:window}
\end{figure}

\begin{figure}[ht]
\centering
\includegraphics[width=0.75\textwidth]{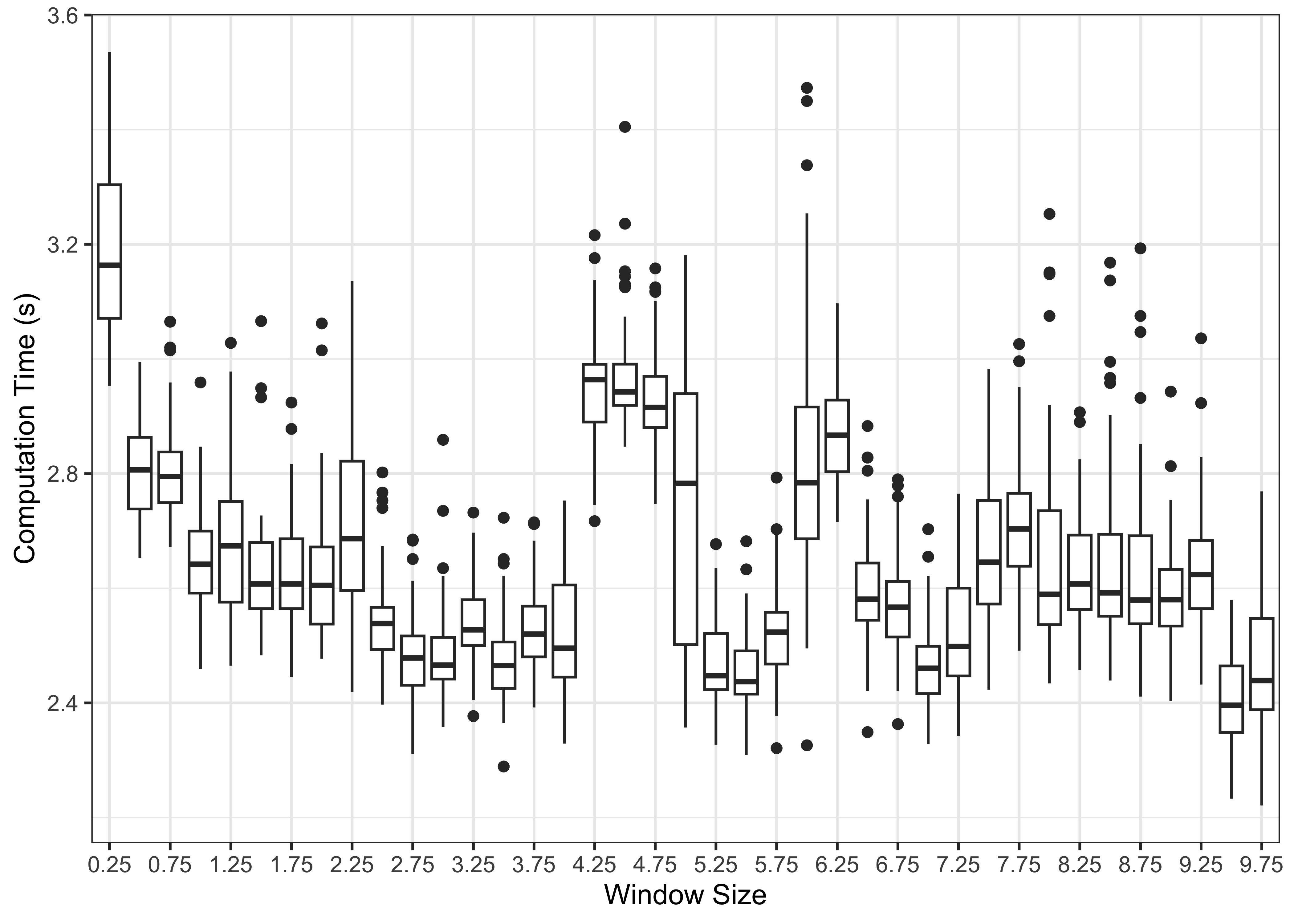}
\caption{Computation time for a fixed network setting as we vary $dT$, across 50 
replications.}
\label{fig_ocd:dt}
\end{figure}

\paragraph{Regret Rate.}

Quantifying the regret of an online estimation scheme is an important
tool in the analysis of such a procedure.
We can investigate the empirical performance of 
the regret for our method in simulation settings,
mirroring the theoretical results.
To do this, we need 
to consider a loss function. 
We define the loss function over the $m$-th time window as the negative
normalized log-likelihood, i.e.
\begin{multline*}
\tilde l_m(\theta \vert z) = - \frac{1}{ \vert A \vert }\sum_{(i,j)\in A} 
\bigg\{ \int_{(m-1)dT}^{m dT} \log \lambda_{ij}(t \vert z)dN_{ij}(t)  
 - \int_{(m-1)dT}^{m dT} \lambda_{ij}(t  \vert  z)dt \bigg\},
\end{multline*}
and define the regret as 
\begin{equation*}
\textrm{Regret}(T) = 
\inf_{\theta^{(m)}\in \Pi(\theta)}\left\{
\sum_{m=1}^M \tilde l_m(\theta^{(m)}  \vert  z^{\ast}) \right\} 
- \sum_{m=1}^M \tilde l_m(\theta^{\ast}  \vert  z^{\ast}),
\end{equation*}
where $M = T/dT$.
This regret function quantifies the gap of the 
conditional likelihood, given the true latent membership $z^{\ast}$,
between the online estimator and the true optimal value.
We note that this regret function is conditional on the true latent 
assignment being known, and as such, we need to account for 
possible permutations of the inferred parameters,
which is done using $\Pi(\theta)$.
While this regret quantity may be of theoretical interest,
in practice it may be more appropriate to 
instead look at the empirical regret, using the
estimated latent community memberships. We shall define this 
as 
\begin{equation*}
  \textrm{Regret}_{EMP}(T) = 
  \sum_{m=1}^M \tilde l_m(\theta^{(m)}  \vert  z^{(m)}) 
  - \sum_{m=1}^M \tilde l_m(\theta^{\ast}  \vert  z^{\ast}),
 \end{equation*}
measuring the cumulative difference between the estimated 
and true log likelihood, as we learn both $z$ and $\theta$ over 
time.

Given these two regret definitions,
we can simulate networks 
{\color{black}
as in the default setting} and compute the empirical 
regret for a fixed network, varying the 
range of time over which events are observed
{\color{black} from $T=50$ to $T=200$.}
This is
shown in Figure~\ref{fig_ocd:regret}.
Here we compute each quantity across 100 simulations, showing
loess smoothed estimates 
and their standard error
over time. We see that as we observe these
networks for a longer period this regret grows slowly.

\begin{figure}[ht]
  \centering
  \includegraphics[width=0.75\textwidth]{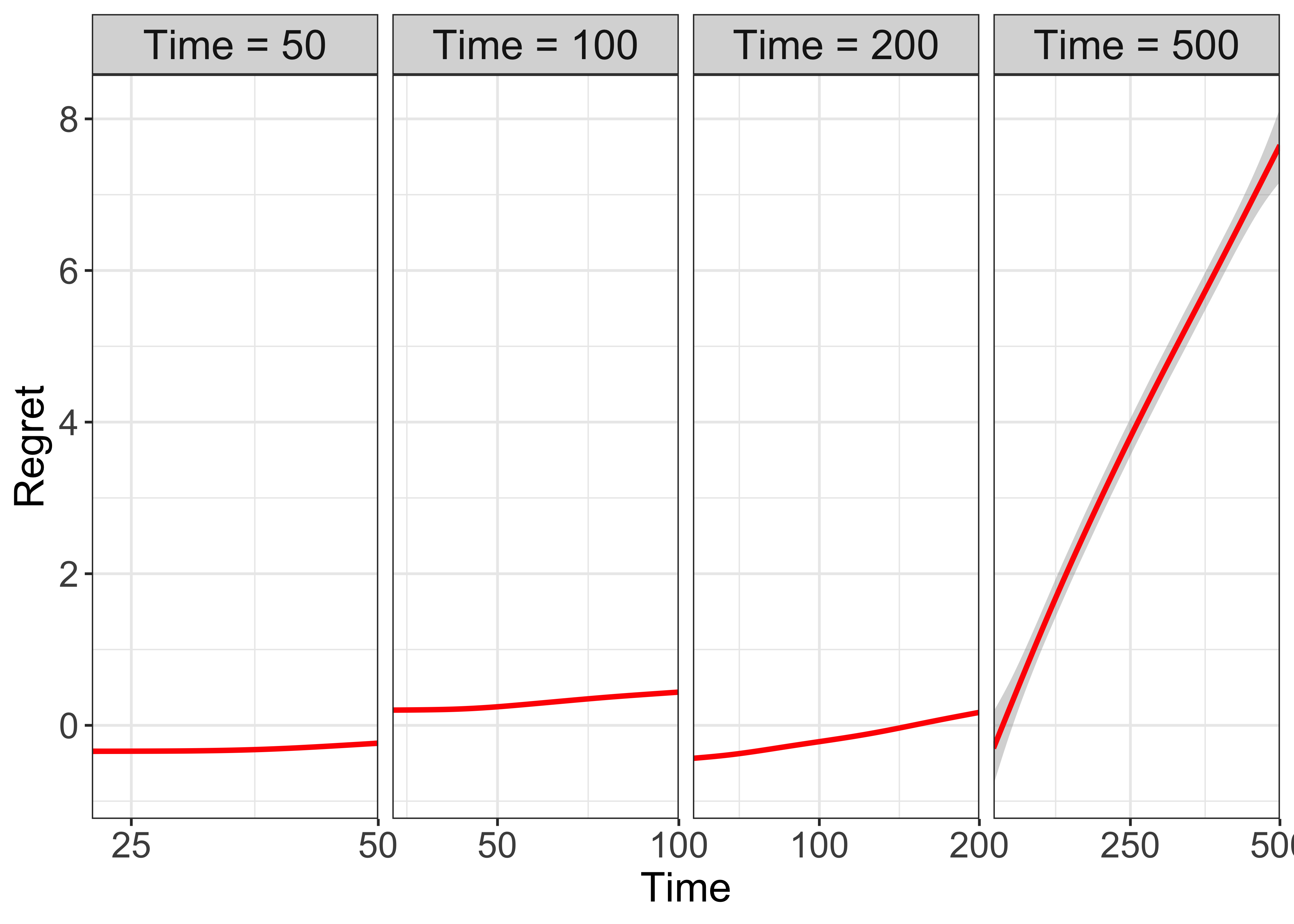}
  \caption{Smoothed regret estimates as a function of time for a 
  fixed network structure observed for varying lengths of time.}
  \label{fig_ocd:regret}
\end{figure}

\subsection*{Hawkes Models}

\paragraph{Hawkes Community Recovery}
The main experiments in Section~\ref{Simulation}
demonstrate the performance of our online learning
algorithm where the intensity function follows an 
inhomogeneous Poisson 
process. Here we demonstrate the performance
of this procedure for Hawkes block models also. 
{\color{black} We consider the same defaults in terms of network and community size and observation time.
The specific parameter settings are given in the associated code repository, 
in each case considering 50 simulations.}
In Figure~\ref{fig_ocd:hawkes_n} we 
first investigate the performance of our procedure
for community recovery, as we increase the number of nodes.
As the number of nodes increase, we can more consistently
recover the true community structure.

\begin{figure}[ht]
\centering
\includegraphics[width=0.75\textwidth]{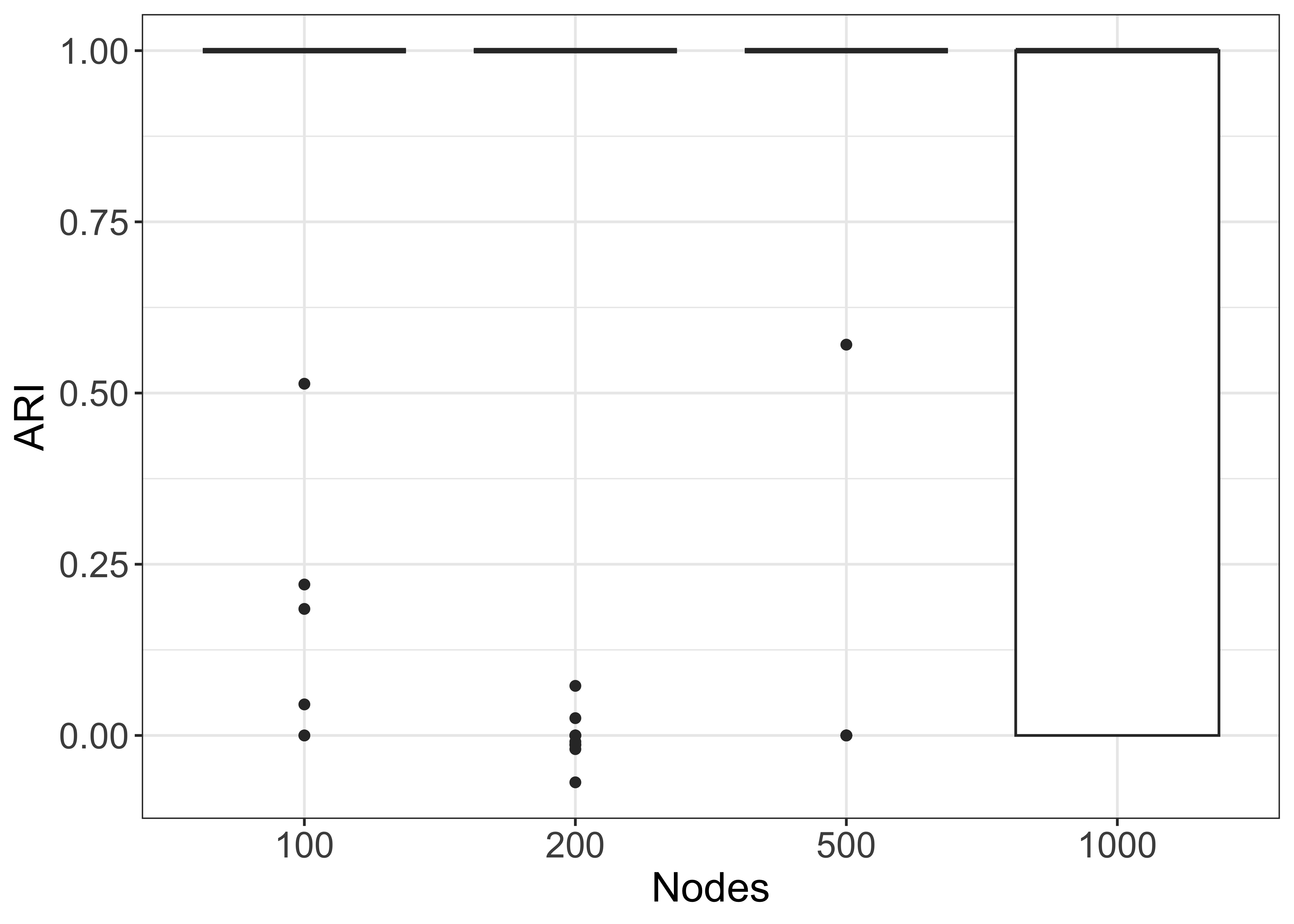}
\caption{Community recovery under the Hawkes model,
in terms of ARI, as we increase the number of nodes. }
\label{fig_ocd:hawkes_n}
\end{figure}

\paragraph{Hawkes, Varying Number of Communities}
We can also investigate the performance of our
procedure for community recovery under the Hawkes model as 
we increase $K$, the number of communities.
In Figure~\ref{fig_ocd:hawkes_k} we show the performance
as we consider more communities for a fixed number of
nodes. As $K$ increases, we are less able to recover the 
true community structure, which is seen by a decrease in the
ARI.

\begin{figure}[ht]
\centering
\includegraphics[width=0.75\textwidth]{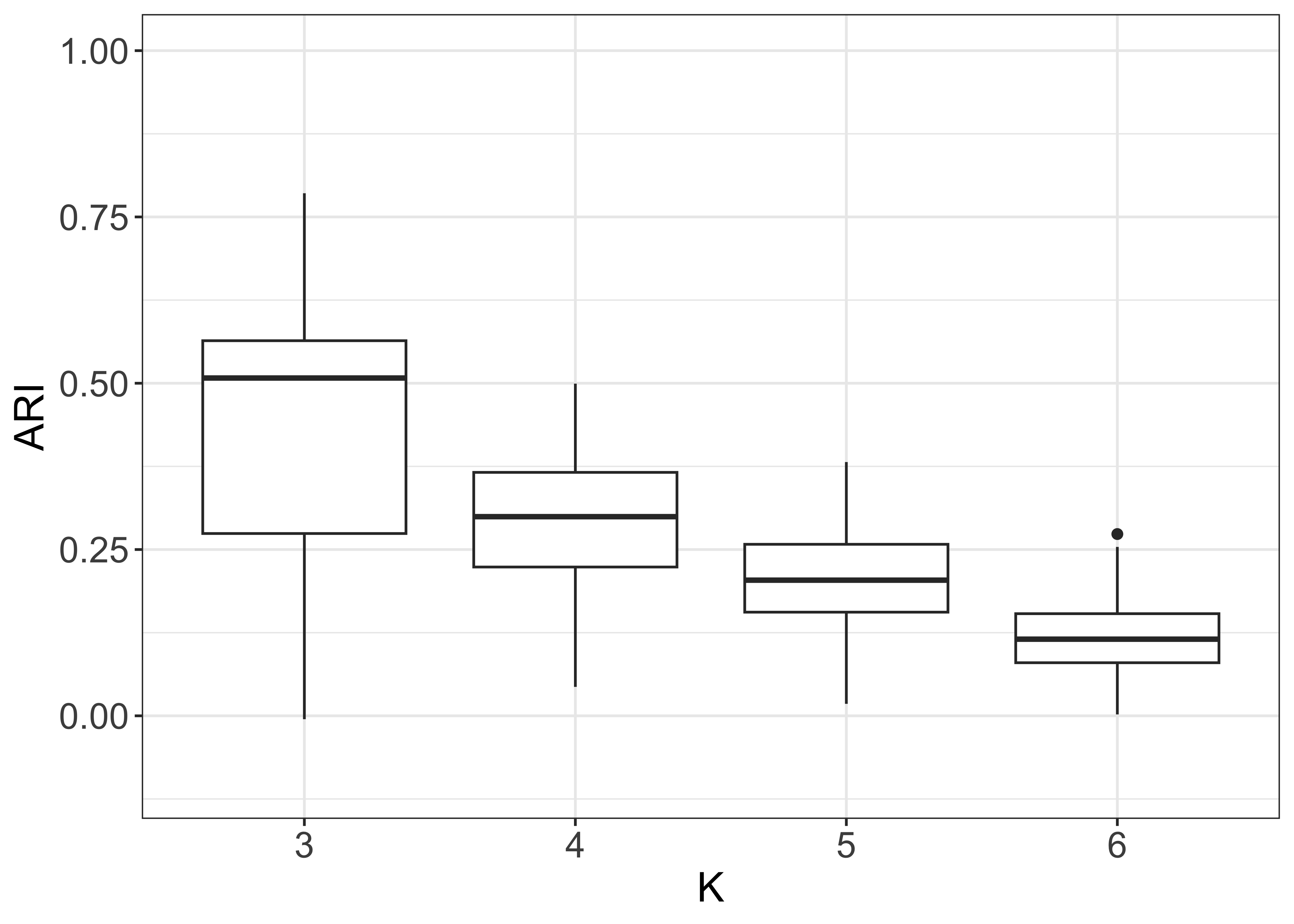}
\caption{Community recovery for the Hawkes model 
as we vary the number of communities.}
\label{fig_ocd:hawkes_k}
\end{figure}

\paragraph{Hawkes, Online Community Recovery}
In Figure~\ref{fig_ocd:hawkes_online} we illustrate 
how the community structure is learned as events are 
observed on the network, varying the number of nodes.
Community recovery is harder than in the Poisson setting,
but the average ARI increases quickly in time, before 
stabilising. There is considerably more uncertainty 
in this estimate than in the Poisson setting.

\begin{figure}[ht]
\centering
\includegraphics[width=0.75\textwidth]{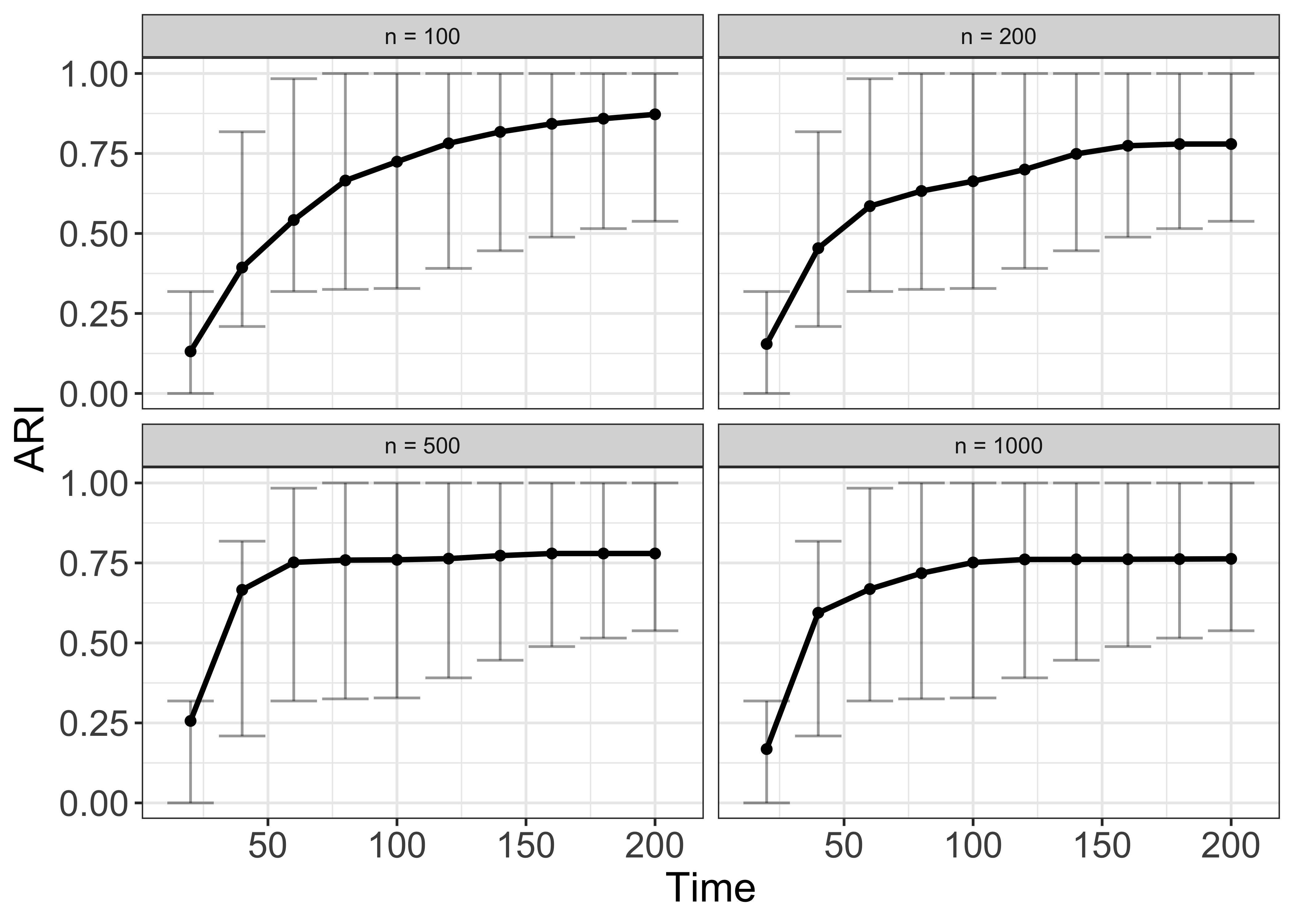}
\caption{Online community recovery under the Hawkes model
as the number of nodes increases, for a fixed observation period.}
\label{fig_ocd:hawkes_online}
\end{figure}

\paragraph{Hawkes, Parameter Recovery}
We can also look at how we recover the true parameters 
of our Hawkes process as events are observed in time 
over the network, as was considered for the Poisson process model 
previously.
Here we measure the recovery of both the baseline rate
matrix $M$ and the 
excitation parameter, $B$, along with the scalar decay parameter $\lambda$.
Figure~\ref{fig_ocd:hawkes_param} shows the {\color{black}GAM} smoothed difference 
between the true and estimated parameters across 50 simulations, {\color{black} with
included standard errors shown}.
For each of these three parameters, we see that the difference between the 
true and estimated parameter decreases, with the difference decreasing as 
we observe events for a longer time period.

\begin{figure}[ht]
\centering
\includegraphics[width=0.75\textwidth]{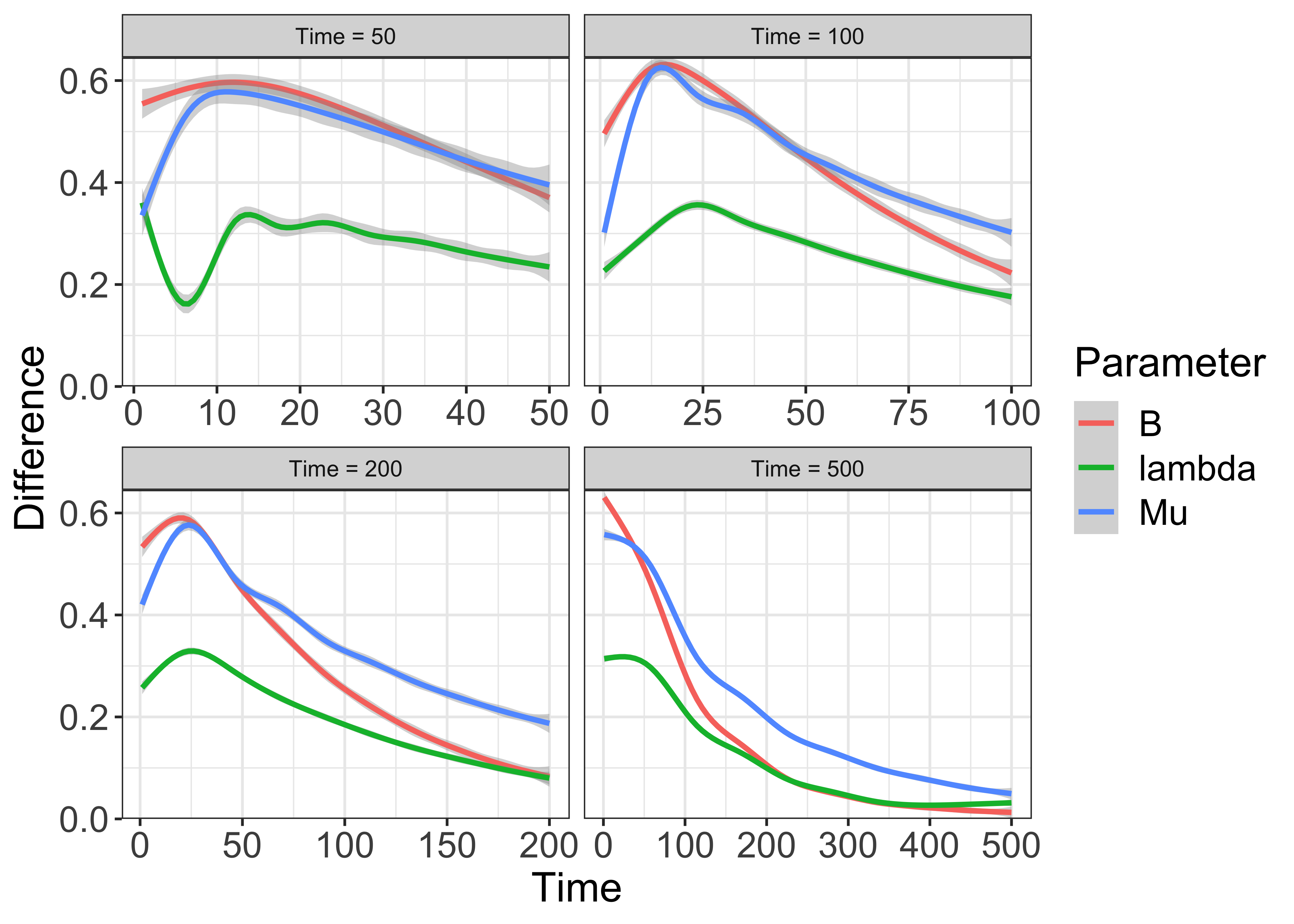}
\caption{Parameter Recovery for block Hawkes model as the 
observation time increases for a fixed network size.}
\label{fig_ocd:hawkes_param}
\end{figure}

\section{Technical Conditions}\label{app:condition}
\label{proof}
In this section, we provide the details of theoretical analyses of our
proposed algorithm under dense event setting, that is, 
integration of intensity function $\lambda_{ij}(t)$ over window length $dT$ is $\Theta(1)$.

Different from the analysis of regular online algorithms, the key difficulties 
in our setting are (1) the model we consider is a latent class network 
model with complicated dynamics, (2) the proposed algorithm 
involves approximation steps. 
Before the proof of the main results, we first introduce some 
notation and definitions. In the following, 
we use variables $c_0 - c_3$, $C$, and $\delta$ to denote some constants 
which may vary from the place to place. 
$\theta^{\ast}$, $z^{\ast}$ represents the true parameter 
and latent class membership, respectively.
\begin{itemize}
    \item[C0] [\textbf {Window Size}] Assume time window $dT$ is some 
    fixed constant which is determined a priori. 
    \item[C1] [\textbf {Expectation}]
    Define the normalized log likelihood over a single time window,
    \begin{equation*}
    l_{w}(\theta \vert z) = \frac{1}{ \vert A \vert } \sum_{(i,j) \in A} \big \{ \int_{0}^{dT} \log \lambda_{ij}(t  \vert  z) dN_{ij}(t)  
     - \int_{0}^{dT} \lambda_{ij}(t  \vert  z) dt \big \}. 
    \end{equation*}
    For simplicity, we assume the expectation of data process is stationary, 
    i.e, $\bar{l}_w (\theta \vert z) = \mathbb E^{\ast} l_{\omega}(\theta \vert z)$ 
    does not depend on the window number $w$.
    Here the expectation is taken with respect to all observed data 
    under the true process.
    \item[C2] [\textbf {Latent Membership Identification}]
    Assume 
    $$ \bar l_w(\theta  \vert  z) \leq \bar l_w(\theta  \vert  z^{\ast}) - 
    c \frac{d_m  \vert z - z^{\ast} \vert _0}{ \vert A \vert },$$
    for any $z \neq z^{\ast}$ and $\theta \in B(\theta^{\ast}, \delta)$.
    Here $B(\theta^{\ast}, \delta)$ is the $\delta$-ball around the 
    true parameter $\theta^{\ast}$;
    $d_m = m^{r_d} (r_d > 0)$ represents the graph connectivity and 
    $ \vert z - z^{\ast} \vert _0$ is the number of individuals such that $z_i \neq z^{\ast}_i$. 
    \item[C3] [\textbf{Continuity}]
    Define $Q$ function, $Q(\theta, q) = \mathbb E_{q(z)} l_w(\theta  \vert  z)$ 
    and $\bar Q(\theta, q) = \mathbb E^{\ast} Q(\theta, q)$. Suppose 
    \begin{eqnarray}
    \bar Q(\theta, q) - \bar l (\theta  \vert  z^{\ast}) \leq c d(q, \delta_{z^{\ast}})
    \end{eqnarray}
    holds, where $\delta_{z^{\ast}}$ is the probability function that 
    put all its mass on the true label vector $z^{\ast}$.
    The distance $d(q_1, q_2) \equiv TV(q_1, q_2)$, 
    where $TV(q_1, q_2)$ is the total variance between two distribution functions. 
    
    Let $\theta(q)$ be the maximizer of $\bar Q(\theta, q)$. Assume that 
    $ \vert \theta(q) - \theta^{\ast} \vert  \leq c d(q, \delta_{z^\ast})$ holds for 
    any $q$ and some constant $c$.
    \item[C4] [\textbf{Gradient Condition}]
     Assume that there exists a $\delta$ such that
    \begin{enumerate}
        \item    
    \begin{eqnarray*}
    \frac{\partial \bar Q(\theta, q)}{\partial \theta}^T (\theta - \theta(q)) < 
    - c \ \vert \theta - \theta(q)\ \vert ^2 <  0 \label{21}
    \end{eqnarray*}
    and
    \begin{eqnarray*}
    \ \vert \frac{\partial \bar Q(\theta, q)}{\partial \theta}^T (\theta - \theta(q))\ \vert  \geq
     c \ \vert \frac{\partial \bar Q(\theta, q)}{\partial \theta}\ \vert ^2  \label{22}
    \end{eqnarray*}
    hold for $\theta \in B(\theta(q), \delta)$ and any $q$ with $c$ being a universal constant.
        \item
        \begin{eqnarray}
        \mathbb E^{\ast} \frac{\partial Q(\theta, q)}{\partial \theta}^T 
        \frac{\partial Q(\theta, q)}{\partial \theta} \leq C
        \end{eqnarray}
        holds for any $\theta \in B(\theta(q), \delta)$ and any $q$.
    \end{enumerate}
    \item[C5] [\textbf{Boundedness}] For simplicity, we assume the functions 
    $\lambda_{ij}(t \vert z)$, $\log \lambda_{ij}(t \vert z)$ and their derivatives are 
    continuous bounded function of parameter $\theta$ for all $z$ and $t$.  
    \item[C6] [\textbf{Network Degree}] Let $d_{i}$ be the number nodes 
    that individual $i$ connects to. 
    We assume that $d_i \asymp d_n$ for all $i$, with $d_n = n^{r_d}$ ($0 <  r_d < 1$) (Here $a \asymp b$ means 
    $a$ and $b$ are in the same order.)
    \item[C7] [\textbf{Initial Condition}]
    Assume $\theta^{(0)} \in B(\theta^{\ast}, \delta)$ for a sufficiently small radius $\delta$ and $q^{(0)}$ satisfies 
    \begin{multline}\label{q:initial}
     \mathbb E_{q^{(0)}(z_{-i})} \bar l_w(\theta^{\ast}  \vert  z_i = z, z_{-i}) \\ 
     \leq  \mathbb E_{q^{(0)}(z_{-i})} \bar l_w(\theta^{\ast}  \vert  z_i = z_i^{\ast}, z_{-i}) \\ - c d_i /  \vert A \vert 
    \end{multline}
    for all $i$ and $z \neq z_i^{\ast}$.
\end{itemize}

These are the regularity conditions required for the proofs of 
Theorem 2 and 3 in the main text.
We first note some important comments on the above conditions. 
Here the window size $dT$ is assumed to be any fixed constant. 
It can also grow with the total number of windows (e.g. $\log T$), 
the result will still hold accordingly. 
Condition C1 assumes the stationarity of process for ease of the proof. 
This condition can also be further relaxed for non stationary processes 
as long as Condition C2 holds for any time window.
In Condition C2, we assume that there is a positive gap between 
log-likelihoods when the latent profile is different from the true one, 
which plays an important role in identification of latent profiles.
Condition C3 postulates the continuity of the $Q$ function. 
In other words, the difference between $Q$ and the true conditional 
likelihood is small, when the approximate posterior $q$ 
concentrates around the true latent labels.
Condition C4 characterizes the gradient of $Q$ function, along 
with the local quadratic property and boundedness.
Condition C5 requires the boundedness of the intensity function. 
It can be easily checked that it holds for Poisson process.
By using truncation techniques, the results can be naturally extended 
to the Hawkes process setting \citep{yang2017online, xu2020network}.
We also note that $d_n$ can be viewed as
the network connectivity,
the degree to which nodes in network connect with each other. 
Condition C6 puts the restriction on the network structure that the 
degrees should not diverge too much across different nodes.
Then $ \vert A \vert  \asymp m d_n$ controls the overall sparsity of the 
network. The network gets sparser when $r_d \rightarrow 0$. 
Here we do not consider the regime where $r_d = 0$ 
(in which case the network is super sparse, i.e. each individual 
has only a finite number of connections on average), 
which could be of interest in future work.
Condition C7 puts the requirement on the initialization of 
model parameters and approximate $q$ function.
Note that \eqref{q:initial} is satisfied when $q$ is close 
to the multinomial distribution which puts mass probability on the true label $z^{\ast}$.
Equation \eqref{q:initial} could also automatically hold when true model parameters of difference classes are well separated.
(That is, we take $q_i^{(0)}(z) = \text{multinom}(\frac{1}{K}, \ldots, \frac{1}{K})$ as non-informative prior so that \eqref{q:initial} holds.)

\section{Useful Lemmas}
In the main proof, we depend on the following Lemmas to ensure the 
uniform convergence of random quantities (e.g. likelihood, ELBO, etc.) 
to their population versions.
\begin{lemma}\label{concentration}
Under Conditions C0, C1 and C5, it holds that 
    \begin{eqnarray}
    & &P( \sup_z  \vert g(\theta  \vert  z) - \mathbb E g(\theta  \vert  z) \vert  \geq  x )  \nonumber \\
    &\leq& C K^m \exp\{ - \frac{1/2  \vert A \vert  x^2}{v^2 + 1/3 M x} \}, 
    \end{eqnarray}
where $g(\theta \vert z)$ is some known functions which could be taken 
as weighted log likelihood or its derivatives; $v$ and $M$ are some constants. ''$\mathbb E$" here is the conditional expectation given fixed label $z$.
\end{lemma}
\textbf{Proof of Lemma \ref{concentration}}
    Without loss of generality, we take $g(\theta  \vert  z) = l_w(\theta  \vert  z)$. 
    Define $X_{ij} = \int_0^{dT} \log \lambda_{ij}(t \vert z) dN_{ij}(t) - \int_0^{dT} \lambda_{ij}(t \vert z)dt$ 
    for any pair $(i,j) \in A$. According to Condition C5, we know that 
    there exists $M$ and $v^2$ such that 
    $ \vert X_{ij} - \mathbb E X_{ij} \vert  \leq M$ and $\textrm{var}(X_{ij}) \leq v^2$. 
    Then we apply the Bernstein
    inequality and get that 
    \begin{eqnarray}
    & & P( \vert \sum_{(i,j) \in A} X_{ij} - \mathbb E X_{ij} \vert  \geq  \vert A \vert x)  \nonumber \\
    &\leq& 2 \exp\{\frac{-\frac{1}{2} \vert A \vert ^2 x^2}{ \vert A \vert v^2 + 1/3 M  \vert A \vert x}\}
    \end{eqnarray}
    By taking union bound over all possible $z$, we then have 
    \begin{eqnarray}
    & & P( \sup_z  \vert g(\theta  \vert  z) - \mathbb E g(\theta  \vert  z) \vert  \geq  x )  \nonumber \\
    &\leq& C K^m \exp\{ - \frac{- 1/2  \vert A \vert ^2 x^2}{ \vert A \vert v^2 + 1/3 M  \vert A \vert x} \}.
    \end{eqnarray}
    Thus, we conclude the proof. 

\bigskip

One immediate result from Lemma \ref{concentration} is that
\begin{corollary}\label{cor:q}
Under the same setting stated in Lemma \ref{concentration}, it holds that
\begin{eqnarray}
& & P( \vert \mathbb E_{q(z)} g(\theta  \vert  z) - \mathbb E_{q(z)} \mathbb E g(\theta  \vert  z) \vert  \geq x) \nonumber \\
&\leq& C K^m \exp\{ - \frac{1/2  \vert A \vert  x^2}{v^2 + 1/3 M x} \},
\end{eqnarray}
for any $q$.
\end{corollary}
\textbf{Proof of Corollary \ref{cor:q}}
For any distribution function $q(z)$, we can observe the following relation,
\[ \vert \mathbb E_{q(z)} g(\theta  \vert  z) - \mathbb E_{q(z)} \mathbb E g(\theta  \vert  z) \vert  \leq \sup_z  \vert g(\theta \vert z) - \mathbb E g(\theta \vert z) \vert ,\]
and 
get the desired result by Lemma \ref{concentration}. QED.

\bigskip

The following
Lemma \ref{prob:diff} and Lemma \ref{lem:concentration:ini} ensure 
the identification of latent memberships.

\begin{lemma}\label{prob:diff}
Under Conditions C0 - C2, C5 - C6, with probability $1 -  \exp\{ - C d_n \}$, it holds that 
\begin{eqnarray}
\sum_{z \neq z^{\ast}} L(\theta  \vert  z) = L(\theta  \vert  z^{\ast}) \cdot O(\exp\{- c_1 d_n\})
\end{eqnarray}
for any $\theta \in B(\theta^{\ast}, \delta)$ for some constants $c_1$ and $\delta$. 
Here, $L(\theta  \vert  z) = \exp\{ \vert A \vert  l_w (\theta  \vert  z)\}$.
\end{lemma}
\textbf{Proof of Lemma \ref{prob:diff}}
The main step of the proof is to show that
\begin{eqnarray}\label{target:lemma2}
l_w(\theta  \vert  z) \leq l_w(\theta  \vert  z^{\ast}) - c/2 \frac{d_n \vert z - z^{\ast} \vert _0}{ \vert A \vert } 
\end{eqnarray}
holds for all $z$ with high probability. We take $g(\theta  \vert  z)$
as $l_w(\theta  \vert  z) - l_w(\theta  \vert  z^{\ast})$. Similar to the 
proof of Lemma \ref{concentration}, we have that 
\begin{eqnarray*}
& & P(  \vert  l_w(\theta  \vert  z) - l_w(\theta  \vert  z^{\ast}) - 
\mathbb E \{l_w(\theta  \vert  z) - l_w(\theta  \vert  z^{\ast})\} \vert  \nonumber \\
& & \geq x \frac{d_n \vert z - z^{\ast} \vert _0}{ \vert A \vert }) \nonumber \\
&\leq& 
\exp\{ - \frac{d_n^2  \vert z - z^{\ast} \vert _0^2 x^2}{ \vert z - z^{\ast} \vert _0 d_{max} (v^2 + 1/3 Mx)}\}
\end{eqnarray*}
by noticing that there are at most $O( \vert z - z^{\ast} \vert _0 d_{n})$ 
number of non-zero $X_{ij}$'s in $l_w(\theta  \vert  z) - l_w(\theta \vert z^{\ast})$.
By taking $x = c/2$, we have
\begin{eqnarray*}
& & P(  \vert  l_w(\theta  \vert  z) - l_w(\theta  \vert  z^{\ast}) - 
\mathbb E \{l_w(\theta  \vert  z) - l_w(\theta  \vert  z^{\ast})\} \vert   \nonumber \\
& & \geq c/2 \frac{d_n \vert z - z^{\ast} \vert _0}{ \vert A \vert }) \nonumber \\ 
&\leq& 
\exp\{ - \frac{\tilde c d_n^2  \vert z - z^{\ast} \vert _0}{ d_{n} (v^2 + 1/6 M c)}\},
\end{eqnarray*}
by using the fact that $d_{max} \asymp d_n$ and adjusting the constant $\tilde c$.
Hence, by union bound, we get 
\begin{multline*}
P( \sup_z  \vert  l_w(\theta  \vert  z) - l_w(\theta  \vert  z^{\ast}) - \\
\mathbb E \{l_w(\theta  \vert  z) - l_w(\theta  \vert  z^{\ast})\} \vert   \geq  c/2 \frac{d_n \vert z - z^{\ast} \vert _0}{ \vert A \vert })
\end{multline*}
\begin{eqnarray}\label{concentration2}
&\leq& \sum_{z} \exp\{ - \frac{\tilde c d_n^2  \vert z - z^{\ast} \vert _0}{ d_{n} (v^2 + 1/6 M c)}\} \nonumber \\ 
&=& 
\sum_{m_0 = 1}^n \sum_{ \vert z - z^{\ast} \vert _0 = m_0} \exp\{ - \frac{\tilde c d_n m_0}{ v^2 + 1/6 M c}\}
\end{eqnarray}

By Condition C2, $d_n = n^{r_d} (r_d > 0)$, \eqref{concentration2} becomes
\begin{multline*}
P\left( \sup_z  \vert  l_w(\theta  \vert  z) - l_w(\theta  \vert  z^{\ast}) - \right. \\
\left. \mathbb E \{l_w(\theta  \vert  z) - l_w(\theta  \vert  z^{\ast})\} \vert   
\geq 
c/2 \frac{d_m \vert z - z^{\ast} \vert _0}{ \vert A \vert } \right) 
\end{multline*}
\begin{eqnarray*}
&\leq& \sum_{m_0 = 1}^m K^{m_0} \exp\{ - \frac{\tilde c d_n m_0}{ v^2 + 1/6 M c}\}\\
&=& \sum_{m_0 = 1}^n  \exp\{ - \frac{\tilde c d_n m_0}{ v^2 + 1/6 M c} + n_0 \log K \}\\
&\leq& \sum_{m_0 = 1}^n  \exp\{ - \frac{\tilde c d_n m_0}{ 2(v^2 + 1/6 M c)} \}\\
&\leq& \exp\{ - C d_n\}
\end{eqnarray*}
for adjusting constant $C$. Together with Condition C2, 
\eqref{target:lemma2} holds with probability $1 - \exp\{-C d_n\}$.

By definition of $L(\theta  \vert  z)$ and \eqref{target:lemma2}, we get that 
$L(\theta  \vert  z) \leq L(\theta  \vert z^{\ast}) \cdot \exp\{ - c/2 \cdot d_m  \vert z - z^{\ast} \vert _0\}$ 
holds for any $z$ with probability $1 - \exp\{- C d_m\}$. Thus, 
\begin{eqnarray}
& &\sum_{z \neq z^{\ast}} L(\theta  \vert  z) \nonumber \\
&\leq& \sum_{z \neq z^{\ast}} L(\theta  \vert  z^{\ast}) \exp\{-c/2 d_n  \vert z - z^{\ast} \vert _0\} \nonumber  \\
&\leq& \sum_{m_0 = 1}^n \sum_{z:  \vert z - z^{\ast} \vert _0 = m_0} \exp\{-c/2 d_n m_0\} \nonumber \\
&\leq& \exp\{- c_1 d_n\} \nonumber,
\end{eqnarray}
by adjusting constant $c_1$. This completes the proof.

\begin{lemma}\label{lem:concentration:ini}
For approximate function $q^{(1)}$, it holds that
\begin{eqnarray}
\sum_{z_i \neq z_i^{\ast}} q_i^{(1)}(z_i) = q_i^{(1)}(z_i^{\ast}) O(\exp\{-\tilde c d_i\}).
\end{eqnarray}
\end{lemma}
\textbf{Proof of Lemma \ref{lem:concentration:ini}}
We first have that 
\begin{multline}
\mathbb E_{q^{(0)}(z_{-i})} l_0(\theta^{0} \vert z_i, z_{-1})\} \\
 \leq \mathbb E_{q^{(0)}(z_{-i})} l_0(\theta^{0} \vert z_i^{\ast}, z_{-1})\} - \frac{c}{2} d_i \label{approx:q1}
\end{multline}
with high probability for any $z_i \neq z_i^{\ast}$.
Under initial condition C7, \eqref{approx:q1} can be proved via the same technique used in Lemma \ref{concentration}.
Secondly, we note that 
\begin{eqnarray}
q_i^{(1)}(z_i) \propto \exp\{\mathbb E_{q^{(0)}(z_{-i})} l_0(\theta^{0} \vert z)\}.
\end{eqnarray}
We then have
\begin{eqnarray}
q_i^{(1)}(z_i) \leq
q_i^{(1)}(z_i^{\ast}) \exp\{- \frac{c}{2} d_i\}.
\end{eqnarray}
By summing over all $z_i$, it gives that 
\begin{eqnarray*}
\sum_{z_i \neq z_i^{\ast}} q_i^{(1)}(z_i) \leq m q_i^{(1)}(z_i^{\ast}) \exp\{-\frac{c}{2} d_i\}
\leq q_i^{(1)}(z_i^{\ast}) \exp\{- \tilde c d_i\},
\end{eqnarray*}
by adjusting the constant $\tilde c$. This concludes the proof.

\section{Proofs of Theorem 2 and Theorem 3} \label{app:thm_proof}
With aid of useful lemmas stated in previous sections, 
we are ready for the proof of main theorems.

\noindent \textbf{Proof of Theorem \ref{regret:thm}}
According to definition of \textrm{Regret}, we have 
\begin{eqnarray*}
\textrm{Regret}(T) &=& 
\sum_{m=1}^M \tilde l_m(\theta^{(m)}  \vert  z) - \sum_{m=1}^M \tilde l_m(\theta^{\ast}  \vert  z^{\ast}) \nonumber \\
&=& \sum_{m=1}^M \{\mathbb E_{q^{(m)}} \tilde l_m(\theta^{(m)}  \vert  z) - \mathbb E_{q^{(m)}} \tilde l_m(\theta^{\ast}  \vert  z)\} \nonumber \\
& & - \sum_{m=1}^M \{ \mathbb E_{q^{(m)}} \tilde l_m(\theta^{(m)}  \vert  z) - 
\tilde l_m(\theta^{(m)}  \vert  z^{\ast}) \} \nonumber \\
& & +  \sum_{m=1}^M \{ \mathbb E_{q^{(m)}} \tilde l_m(\theta^{\ast}  \vert  z) - 
\tilde l_m(\theta^{\ast}  \vert  z^{\ast}) \}.
\end{eqnarray*}
Next we prove the result by the following three steps.
\begin{itemize}
    \item[] \textbf{Step 1.} With high probability, it holds that $\theta^{(m)} \in B(\theta^{\ast}, \delta)$ and 
    \begin{eqnarray}\label{q:diff}
    q^{(m)} (z^{\ast}) \geq 1 - C \exp\{-c_1 m d_n\}
    \end{eqnarray}
    for $m = 1, 2, \ldots$. 
    \item[] \textbf{Step 2.} With high probability, it holds that
    \begin{eqnarray}
    & & \sum_{m=1}^M \{ \mathbb E_{q^{(m)}} \tilde l_m(\theta^{(m)}  \vert  z) - 
    \mathbb E_{q^{(m)}} \tilde l_m(\theta^{\ast}  \vert  z) \} \nonumber \\ 
    &\leq& C \sqrt{M} \log(M  \vert A \vert )^2, 
    \end{eqnarray} 
    for some constant $C$.
    \item[] \textbf{Step 3.} With high probability, it holds that
    \begin{eqnarray}
    & &  \vert \sum_{m=1}^M \{ \mathbb E_{q^{(m)}} \tilde l_m(\theta^{(m)}  \vert  z) - 
    \tilde l_m(\theta^{(m)}  \vert  z^{\ast}) \} \vert   \nonumber \\
    &\leq& M \exp\{ - c d_n\},
    \end{eqnarray}
    and 
    \begin{eqnarray}
    & &  \vert \sum_{m=1}^M \{ \mathbb E_{q^{(m)}} \tilde l_m(\theta  \vert  z) - \tilde l_m(\theta  \vert  z^{\ast}) \} \vert 
     \nonumber \\ 
    &\leq& M \exp\{ - c d_n\}
    \end{eqnarray}
    for any $\theta \in B(\theta^{\ast}, \delta)$ and some constant $c$.
\end{itemize}

\textbf{Proof of Step 1.}
We prove this by mathematical induction. 
When $m = 0$, it is obvious that $\theta^{(0)} \in B(\theta^{\ast}, \delta)$ 
according to the assumption on initialization.
By Lemma \ref{lem:concentration:ini}, we have that 
$\sum_{z_i \neq z_i^{\ast}} q_i^{(1)}(z_i) = q_i^{(1)}(z_i^{\ast}) O(\exp\{- \tilde c d_i\})$. 
Then $q^{(1)}(z^{\ast}) = \prod_i q_i^{(1)}(z_i^{\ast}) \geq 1 - n O(\exp\{-c_1 d_n\})$.
That is, \eqref{q:diff} holds for $q^{(1)}$ by adjusting constant $c_1$.
Next we assume that $\theta^{(m)} \in B(\theta^{\ast}, \delta)$ 
and \eqref{q:diff} holds for any $m \leq m_1$ and need to show 
that $\theta^{(m_1+1)} \in B(\theta^{\ast}, \delta)$ 
and \eqref{q:diff} holds for $m = m_1 + 1$.

We consider the following two scenarios,\\ (1) $0 \leq \ \vert \theta^{(m_1)} - \theta(q^{(m_1)})\ \vert  < \frac{\delta}{2}$ and \\
(2)
$\frac{\delta}{2} \leq \ \vert \theta^{(m_1)} - \theta(q^{(m_1)})\ \vert $.
We can compute that 
\begin{multline}
\ \vert \theta^{(m_1 + 1)} - \theta(q^{(m_1)})\ \vert ^2 = \\ 
\ \vert \theta^{(m_1)} + \eta_{m_1} \frac{ \partial \mathcal Q_{m_1+1}(\theta, q)}{\partial \theta} - \theta(q^{(m_1)}) \ \vert ^2 \label{one-step} \\
= \ \vert \theta^{(m_1)} - \theta^{\ast}\ \vert ^2\\ + 2 \eta_{m_1}
\frac{ \partial \mathcal Q_{m_1+1}(\theta, q)}{\partial \theta}(\theta^{(m_1)} - \theta(q^{(m_1)})) \\
+ \eta_{m_1}^2 \ 
\vert \frac{ \partial \mathcal Q_{m_1+1}(\theta, q)}{\partial \theta}\ \vert ^2. 
\end{multline}
By Lemma \ref{concentration}, we know that 
$ \frac{ \partial \mathcal Q_{m_1+1}(\theta, q)}{\partial \theta} = 
\frac{\partial \bar{\mathcal Q}(\theta, q)}{\partial \theta} + \epsilon$ where $\epsilon = o(1)$ 
for all $q$. 
Therefore, the
right hand side of \eqref{one-step} 
becomes 
\begin{multline}
\ \vert \theta^{(m_1)} - \theta(q^{(m_1)})\ \vert ^2 + 2 \eta_{m_1}
\frac{\partial \bar{\mathcal Q}(\theta, q)}{\partial \theta}(\theta^{(m_1)} - \theta(q^{(m_1)}) ) \\
+ 2 \epsilon \eta_{m_1} \ \vert \theta^{(m_1)} - \theta(q^{(m_1)})\ \vert 
+ 2(\eta_{m_1}^2 + \epsilon^2)\ \vert \frac{\partial \bar{\mathcal Q}(\theta, q)}{\partial \theta}\ \vert ^2.
\label{one-step-2}
\end{multline}
In the first scenario, by Condition C4, \eqref{one-step-2} implies that 
\begin{multline}
\ \vert \theta^{(m_1 + 1)} - \theta(q^{(m_1)})\ \vert  \leq \ \vert \theta^{(m_1)} 
- \theta(q^{(m_1)})\ \vert 
+ c \delta \epsilon \\ + O(\epsilon^2)
< \frac{3}{4} \delta
\label{one-step-3}
\end{multline}
when the step size $\eta_{m_1}$ is small (e.g.$\eta_{m_1} \leq c/2$).

In the second scenario, by Condition C4, \eqref{one-step} implies that 
\begin{multline}
\vert \theta^{(m_1 + 1)} - \theta(q^{(m_1)})\ \vert  
\leq  \vert \theta^{(m_1)} 
- \theta(q^{(m_1)})\ \vert \\
- 2(c + o(1)) \eta_{m_1} \ \vert \theta^{(m_1)} - \theta(q^{(m_1)})\ \vert ^2 \\
+ \eta^2(1 + o(1))c^2 \ \vert \theta^{(m_1)} - \theta(q^{(m_1)})\ \vert ^2 \\
\leq \ \vert \theta^{(m_1)} 
- \theta(q^{(m_1)})\ \vert  - \frac{1}{2}c \eta_{m_1} \ \vert \theta^{(m_1)} 
- \theta(q^{(m_1)})\ \vert ^2
\label{one-step-4}
\end{multline}
for $\eta_{m_1} \leq 1/2c$.

According to \eqref{q:diff} via induction, we have that 
$d(q^{(m_1)}(z^{\ast}), \delta_{z^{\ast}}) = C \exp\{-c_1 m_1 d_n\}$. 
This further implies that $\ \vert \theta(q^{(m_1)}) - \theta^{\ast}\ \vert  = O(\exp\{-c_1 m_1 d_n\})$ by Condition C3. 
By above facts, in the first scenario, we 
have 
\begin{eqnarray}
\ \vert \theta^{(m_1 + 1)} - \theta^{\ast}\ \vert  &\leq& \ \vert \theta^{(m_1 + 1)} - 
\theta(q^{(m_1)})\ \vert  + \ \vert \theta(q^{(m_1)}) -
\theta^{\ast}\ \vert  \nonumber \\
&\leq& 
\ \vert \theta^{(m_1)} - 
\theta(q^{(m_1)})\ \vert  + \ \vert \theta(q^{(m_1)}) -
\theta^{\ast}\ \vert  \nonumber \\
&\leq& 3\delta/4
+
 2 \ \vert \theta(q^{(m_1)}) - \theta^{\ast}\ \vert  < \delta.
\end{eqnarray}

 In the second scenario,
 we have
 \begin{multline*}
  \vert \theta^{(m_1 + 1)} - \theta^{\ast}\ \vert  \leq \ \vert \theta^{(m_1 + 1)} - 
\theta(q^{(m_1)})\ \vert  +  \vert \theta(q^{(m_1)}) -
\theta^{\ast}\ \vert  \\
\leq \ \vert \theta^{(m_1)} - 
\theta(q^{(m_1)})\ \vert  \\
- \frac{1}{2}c \eta_{m_1} \ \vert \theta^{(m_1)} 
- \theta(q^{(m_1)})\ \vert ^2 + \ \vert \theta(q^{(m_1)}) - 
\theta^{\ast}\ \vert  \\
 < 
\delta - \frac{c \eta_{m_1} \delta}{4} + O(\exp\{-c_1m_1d_n\})
< \delta,
 \end{multline*}
 where the last inequality holds since 
 $\eta_{m_1} = c m_1^{-1/2} >> \exp\{-c_1m_1d_n\}$.
Hence, we conclude that $\theta^{(m_1 + 1)} \in B(\theta^{\ast}, \delta)$.

Next, we turn to study approximate distribution $q^{(m+1)}$ to show that \eqref{q:diff} holds for $m = m_1 + 1$. According to Condition C2 and Lemma \ref{prob:diff}, we have that
$ \vert A \vert  l_{m_1}(\theta^{(m_1)} \vert z) \leq  \vert A \vert  l_{m_1}(\theta^{(m_1)} \vert z^{\ast}) - c d_n$
 for any $z \neq z^{\ast}$. This implies that 
\begin{eqnarray*}
\mathbb E_{q_{-z_i}^{(m_1)}} 
 \vert A \vert  l_{m_1}(\theta^{(m_1)} \vert z_i , z_{-i}) \leq 
 \vert A \vert  l_{m_1}(\theta^{(m_1)} \vert z^{\ast}) - c d_n 
\end{eqnarray*}
for any $z_i \neq z_i^{\ast}$ and 
\begin{eqnarray*}
& & \mathbb E_{q_{-z_i}^{(m_1)}}  \vert A \vert  l_{m_1}(\theta^{(m_1)} \vert z_i^{\ast}, z_{-i}) \\
&\geq&  \vert A \vert  l_{m_1}(\theta^{(m_1)} \vert z^{\ast}) - C  \vert A \vert  \exp\{-c m_1 d_n\}. 
\end{eqnarray*}
Combining these two facts, we have that 
\begin{equation*}
E_{q_{-z_i}^{(m_1)}}  \vert A \vert  l_{m_1}(\theta^{(m_1)} \vert z_i , z_{-i})
< 
\mathbb E_{q_{-z_i}^{(m_1)}}  \vert A \vert  l_{m_1}(\theta^{(m_1)} \vert z_i^{\ast}, z_{-i})
- c_2 d_n,
\end{equation*}
holds for any $z_i \neq z_i^{\ast}$ and some adjusted constant $c_2$.

By recursive formula 
$$
S^{(m_1 + 1)}(z_i) = S^{(m_1)}(z_i) \exp\{\mathbb E_{q_{-z_i}^{(m_1)}}
 \vert A \vert  l_m(\theta^{(m_1)}  \vert  z_i, z_{-i}) \},
$$ 
we then have 
$$
\sum_{z_i \neq z_i^{\ast}} S^{(m_1 + 1)}(z_i) = S^{(m_1 + 1)}(z_i^{\ast})
O(\exp\{- (c_1 m_1 + c_2) d_n\}),$$ 
which indicates that 
$$
q^{(m_1 + 1)}(z_i^{\ast}) \geq 1 - \exp\{- (c_1 m_1 + c_2) d_n\}.
$$ 
Finally, noting $q^{(m_1 + 1)}(z^{\ast}) = \prod_{i=1}^n q_i^{(m_1 + 1)}(z_i^{\ast})$ 
gives us 
\begin{align*}
     & q^{(m_1 + 1)}(z^{\ast}) \geq 1 - n \exp\{-(c_1 m_1 + c_2) d_n\} \\
     & \geq 1 - \exp\{-c_1(m_1 + 1) d_n\}.\\
\end{align*}
Hence, we complete \textbf{Step 1} by induction.
\textbf{Proof of Step 2.}
For notational simplicity, we denote 
$\mathbb E_{q^{(m)}} \tilde l_m(\theta  \vert  z)\}$ 
as $h_m(\theta)$ in the remaining part of the proof.
By local convexity, we have 
\begin{eqnarray*}
h_m(\theta^{\ast}) - h_m(\theta^{(m)}) \geq \nabla h_m(\theta^{(m)})^T 
(\theta^{\ast} - \theta^{(m)}),
\end{eqnarray*}
which is equivalent to 
\begin{equation}
\label{expand:grad}
h_m(\theta^{(m)}) - h_m(\theta^{\ast}) \leq \
\nabla h_m(\theta^{(m)})^T (\theta^{(m)} - \theta^{\ast}).
\end{equation}

We know that 
\begin{eqnarray*}
& & d(\bar \theta^{(m+1)}, \theta^{\ast}) - d(\theta^{(m)}, \theta^{\ast}) \nonumber \\
&\leq& \ \vert \theta^{(m)} -  \eta_m \nabla h_m(\theta^{(m)}) - 
\theta^{\ast}\vert ^2 - \ \vert  \theta^{(m)} - \theta^{\ast}\vert ^2 \nonumber \\
&\leq& \eta_m^2 \ \vert \nabla h_m(\theta^{(m)}) \vert ^2 - 
2 \eta_m \nabla h_m(\theta^{(m)})^T (\theta^{(m)} - \theta^{\ast}),
\end{eqnarray*}
where $\bar \theta^{(m+1)} = \theta^{(m)} - \eta_m \nabla h_m(\theta^{(m)})$.
By summing over $m$ and the fact that 
$d(\theta^{(m)}, \theta^{\ast}) \leq d(\bar \theta^{(m)}, \theta^{\ast})$, 
we have
\begin{eqnarray}\label{expand:dist}
& & \sum_m \{ d(\theta^{(m+1)}, \theta^{\ast}) - d(\theta^{(m)}, \theta^{\ast})\} \nonumber \\
&\leq& \sum_m \{ d(\bar \theta^{(m+1)}, \theta^{\ast}) - d(\theta^{(m)}, \theta^{\ast})\} \nonumber \\
&\leq& \sum_m \{ \eta_m^2 \ \vert \nabla h_m(\theta^{(m)})\ \vert ^2 \\
&    & - 2 \eta_m \nabla h_m(\theta^{(m)})^T (\theta^{(m)} - \theta^{\ast}) \}.
\end{eqnarray}
By equation \eqref{expand:grad}, we then have 
\begin{eqnarray}\label{main:regret}
\textrm{regret} &\leq& \sum_m \nabla h_m(\theta^{(m)})^T (\theta^{(m)} - \theta^{\ast}) \nonumber \\
&\leq& \frac{1}{2 \eta_m} (d(\theta^{(0)}, \theta^{\ast}) - d(\theta^{(m+1)}, \theta^{\ast})) \nonumber \\ 
& & + \sum_m \frac{\eta_m}{2} \ \vert \nabla h_m(\theta^{(m)}) \vert ^2,
\end{eqnarray}
where the second inequality uses \eqref{expand:dist}.

Next, we prove that $\nabla h_m(\theta^{(m)})$ is bounded 
with probability going to $1$ for any $m$. Note that,
\begin{eqnarray*}
\nabla l_m(\theta  \vert  z) &=& \nabla \left( \frac{1}{ \vert A \vert } \left[ \sum_{(i,j) \in A} 
\int_{(m - 1)\omega}^{m \omega} \log \lambda_{ij}(s \vert z) dN_{ij}(s) \right. \right. \nonumber \\
& & \left. \left. - \int_{(m - 1)\omega}^{m\omega} \lambda_{ij}(s \vert z) ds\} \right] \right) \\
&\leq& \frac{1}{ \vert A \vert } \left\{ \sum_{(i,j) \in A} \int_{(m - 1)\omega}^{m \omega}
 \frac{\lambda_{ij}^{'}(s \vert z)}{\lambda_{ij}(s \vert z)} dN_{ij}(s) \nonumber \right. \\
& & \left. - \int_{(m - 1)\omega}^{m\omega} \lambda_{ij}^{'}(s \vert z) ds\right\}.
\end{eqnarray*}
Let $B_1 = \sup_{t,z} \frac{\lambda_{ij}^{'}(t \vert z)}{\lambda_{ij}(t \vert z)}$ 
and $B_2 = \sup_{t, z} \lambda_{ij}^{'}(s \vert z)$. 
Both $B_1$ and $B_2$ are bounded according to Condition C5.
We know that the number of events, $M_w$, in each time window follows a 
Poisson distribution with mean $\int_0^{\omega} \lambda(s) ds$. 
Therefore, we get $P(M_w \geq m_w) \leq \exp\{-c m_w\}$ for some constant $c$.
We then have that $\nabla l_m(\theta  \vert  z) \leq C(B_1 m_w + B_2)$ 
with probability at least $1 - M  \vert A \vert  \exp\{ - c m_w\}$.

By letting $\eta_m = \frac{1}{\sqrt{M}}$, \eqref{main:regret} becomes
\begin{eqnarray*}
\textrm{regret} &\leq& C (\sqrt{M} d(\theta^{(0)}, \theta^{\ast}) + \sqrt{M} (B_1 m_W + B_2)^2) \\
&\leq& C \sqrt{M} \log(M  \vert A \vert )^2,
\end{eqnarray*}
where we set $m_w = c \log (M  \vert A \vert )$.

\textbf{Proof of Step 3.}
We only need to show that for each $m$, it holds that 
\begin{eqnarray*}
 \vert \mathbb E_{q^{(m)}} \tilde l_m(\theta \vert z) - \tilde l_m(\theta  \vert  z^{\ast}) \vert  \leq C \exp\{-c d_n\}.
\end{eqnarray*}
We know that 
\begin{eqnarray*}
& & \mathbb E_{q^{(m)}} \tilde l_m(\theta  \vert  z^{\ast}) \\
&=& q^{(m)}(z^{\ast}) \tilde l_m(\theta  \vert  z^{\ast}) + \sum_{z \neq z^{\ast}} q^{(m)}(z)
  \tilde l_m(\theta  \vert  z) \\
&\leq& q^{(m)}(z^{\ast}) \tilde l_m(\theta  \vert  z) + 
\sum_{z \neq z^{\ast}} q^{(m)} B_0 m_W.
\end{eqnarray*}
This implies that 
\begin{eqnarray*}
 \vert \mathbb E_{q^{(m)}} \tilde l_m(\theta \vert z) - 
\tilde l_m(\theta  \vert  z^{\ast}) \vert  &\leq& (1 - q^{(m)}(z^{\ast}))  
\tilde l_m(\theta  \vert  z^{\ast})  \\
& & + \sum_{z \neq z^{\ast}} q^{(m)}(z)  \tilde l_m(\theta  \vert  z) \\
&\leq& C B_0 m_w \exp\{-c d_n\} .
\end{eqnarray*}
This completes the proof.

\noindent \textbf{Proof of Theorem \ref{rate:thm}}
By the update rule, we know that $\theta^{(m+1)} = \theta^{(m)} - 
\eta_m \nabla h_m(\theta^{(m)})$, so
\begin{equation*}
\ \vert \theta^{(m+1)} - \theta^{\ast} \vert ^2 \leq \ \vert \theta^{(m)} - 
\eta_m \nabla h_m(\theta^{(m)}) - \theta^{\ast} \vert ^2 \nonumber\\
\end{equation*}
\begin{multline}
= \ \vert \theta^{(m)} - \theta^{\ast}\ \vert ^2 
- \eta_m \nabla h_m(\theta^{(m)}) (\theta^{(m)} - \theta^{\ast}) \\ 
+ \eta_m^2 \nabla h_m^2(\theta^{(m)}).     
\end{multline}
Furthermore, 
\begin{eqnarray}\label{update:ineq}
& & \ \vert \theta^{(m+1)} - \theta^{\ast}\ \vert ^2  \nonumber \\
&\leq& 
\ \vert \theta^{(m)} - \theta^{\ast}\ \vert ^2 - 
\eta_m \nabla {h}_m(\theta^{(m)}) (\theta^{(m)} - \theta^{\ast}) \nonumber \\
& & + \eta^2_n \nabla h_n^2(\theta^{(n)})  \nonumber \\
&=& \ \vert \theta^{(m)} - \theta^{\ast}\ \vert ^2 - 
\eta_m (\nabla {h}_m(\theta^{(m)}) - 
\nabla {\tilde{l}}(\theta^{(m)}) \nonumber \\
& & + \nabla {\bar \tilde{l}}(\theta^{(m)})) (\theta^{(m)} - \theta^{\ast}) + \eta^2_m
  \nabla h_m^2(\theta^{(m)})  \nonumber \\
&\leq& \ \vert \theta^{(m)} - \theta^{\ast}\ \vert ^2 - 
\eta_m \nabla \bar {\tilde{l}}(\theta^{(m)} \vert z^{\ast}) (\theta^{(m)} - \theta^{\ast})  \nonumber \\
& & + \eta^2_m  \nabla h_m^2(\theta^{(m)}) + 
c \eta_m \delta d(q^{(m)}, \delta_{z^{\ast}}) \nonumber \\
& & + c \delta \eta_m O_p\left(\frac{1}{\sqrt{ \vert A \vert }}\right),
\end{eqnarray} 
where the term $\frac{1}{\sqrt{ \vert A \vert }}$ comes from the probability bound in Lemma \ref{concentration}.
Notice that $\theta(q) = \theta^{\ast}$ when $q = \delta_{z^{\ast}}$, we have that 
$- \nabla \bar {\tilde{l}}(\theta^{(m)} \vert z^{\ast})
 (\theta^{(m)} - \theta^{\ast}) \leq 
 - c \ \vert \theta^{(m)} - \theta^{\ast}\ \vert ^2,$
according to Condition C4. 
Furthermore, we know that $d(q^{(m)}, \delta_{z^{\ast}}) \leq \exp\{-m d_n\}$ 
(see \eqref{q:diff}). Therefore, $\eta_m \delta d(q^{(m)}, \delta_{z^{\ast}})$ 
can be absorbed into $\eta^2_m  \nabla h_m^2(\theta^{(m)})$.
In summary, \eqref{update:ineq} becomes 
\begin{eqnarray*}
\ \vert \theta^{(m+1)} - \theta^{\ast} \vert ^2  &\leq& (1 - c \eta_m) \vert \theta^{(m)} 
- \theta^{\ast} \vert ^2 \\
& & + C \eta_m^2 \nabla h_m^2(\theta^{(m)})
+ C \frac{1}{\sqrt{ \vert A \vert }} \eta_m ,
\end{eqnarray*}
which further gives,
\begin{equation}\label{update:final}
\ \vert \theta^{(m+1)} - \theta^{\ast}\ \vert ^2 \leq 
(1 - c \eta_m) \ \vert \theta^{(m)} - \theta^{\ast}\vert ^2 \nonumber 
+ B \left(\eta_m^2 (\log(M \vert A \vert ))^2 + 
\eta_m \frac{1}{\sqrt{ \vert A \vert }}\right), 
\end{equation}
by adjusting constants and noticing that 
$\nabla h_m^2(\theta)$ is bounded by $(\log(M \vert A \vert ))^2$. 
After direct algebraic calculation, we have 
\begin{eqnarray}\label{update:recursive}
& & \ \vert \theta^{(m+1)} - \theta^{\ast} \vert ^2 \nonumber \\
&\leq& \ \vert \theta^{(0)} - \theta^{\ast}\ \vert ^2 \prod_{t=0}^m (1 - c \eta_t) \nonumber \\ 
& & + B \sum_{t = 0}^m \eta_t \left(\eta_t (\log(N \vert A \vert ))^2 + 
\frac{1}{\sqrt{A}}\right) \prod_{s = t+1}^m (1 - c \eta_t). \nonumber \\
\end{eqnarray}
For the first term in \eqref{update:recursive}, 
we have that 
\begin{eqnarray*}
\prod_{t=0}^m (1 - c \eta_t) &\leq& \prod_{t=0}^m \exp\{-c \eta_t\} \\
&= & \exp\{-c\sum_{t=0}^m \eta_t\} \\
&\leq& C \exp\{- m^{1-\alpha}\}.
\end{eqnarray*}
Next, we define 
$x_{1t} = \eta_t (\log(M \vert A \vert ))^2$ and 
$x_2 = 1/\sqrt{ \vert A \vert }$, to simplify the remainder of the proof.
For the second term in \eqref{update:recursive}, we have that 
\begin{eqnarray*}
& & \sum_{t=0}^m \eta_t (x_{1t} + x_2) \prod_{s = t+1}^m (1 - c\eta_s) \nonumber \\
& = & \sum_{t=0}^{m/2} \eta_t (x_{1t} + x_2)  \prod_{s = t+1}^m (1 - c\eta_s) \\
& & + \sum_{t=m/2 + 1}^{m} \eta_t (x_{1t} + x_2)  \prod_{s = t+1}^m (1 - c\eta_s) \nonumber \\
& = &  \sum_{t=0}^{m/2} \eta_t (x_{1t} + x_2)  \prod_{s = t+1}^m (1 - c\eta_s) \\
& & + \sum_{t=m/2 + 1}^{m} (x_{1t} + x_2)  (1 - (1 - c\eta_t))/c \times \\
& & \prod_{s = t+1}^m (1 - c\eta_s) \nonumber \\
& \leq &  \sum_{t=0}^{m/2} \eta_t (x_{1t} + x_2)  \prod_{s = t+1}^m (1 - c\eta_s) \\
& & + \frac{1}{c} (x_{1m/2} + x_2)  \sum_{t=m/2 + 1}^{m}  
(1 - (1 - c\eta_t)) \times \\
& & \prod_{s = t+1}^m (1 - c\eta_s) \nonumber \\
&\leq& \sum_{t=0}^{m/2} \eta_t (x_{1t} + x_2)  
\prod_{s = t+1}^m (1 - c\eta_s) + 
\frac{1}{c} (x_{1m/2} + x_2)  \nonumber \\
&\leq& \exp\left\{ -c \sum_{t = m/2 + 1}^m \eta_t\right\} (\sum_{t=0}^{m/2} \eta_t (x_{1t} + x_2) ) \\
& & + \frac{1}{c} (x_{1m/2} + x_2)  \nonumber \\
&\leq& m \exp\{ - c m^{1-\alpha}\} + 1/c (x_{1 m/2} + x_2) \\
&\leq& c_0 (m^{-\alpha} (\log(M \vert A \vert ))^2 + \frac{1}{\sqrt{ \vert A \vert }}), 
\end{eqnarray*}
by adjusting the constants. 
Combining the above inequalities, 
we have 
$$
\ \vert \theta^{(m)} - \theta^{\ast} \vert ^2 
= O_p \left(m^{-\alpha} (\log(M \vert A \vert ))^2 +
\frac{1}{\sqrt{ \vert A \vert }}\right).
$$ 
This concludes the proof.

\section{Community Recovery under Relaxed Conditions}
\label{app:extended}
In this section, we relax the conditions mentioned in previous section by considering the case of uneven degree distribution. 
Let $d_i$ be the number of nodes that $i$-th individual connects to.
Then uneven degree distribution means that $d_i$ are not in the same order. 
Degree $d_i$ goes to infinity for some node $i$'s and is bounded for other $i$'s.
Under this setting, we establish results for consistent community recovery.
We start with introducing a few more modified conditions.
\begin{itemize}
    \item[C2'] [\textbf {Latent Membership Identification}]
    Assume 
    $$ \bar l_w(\theta  \vert  z_{\mathcal N}, z_{-\mathcal N}^{\ast}) 
    \leq \bar l_w(\theta  \vert  z^{\ast}) - c \frac{\sum_{i \in \mathcal N} d_i}{ \vert A \vert },$$
    for any subset $\mathcal N \subset \{1,\ldots,n\}$ and $z_{\mathcal N}$ 
    and $z_{-\mathcal N}$ are the sub-vectors of $z$ with and without elements in $\mathcal N$ respectively.  
    \item[C3'] [\textbf{Continuity}]
    Assume 
    \begin{eqnarray}
    \bar Q(\theta, q) - \bar l_w (\theta  \vert  z^{\ast}) 
    \leq c \frac{1}{ \vert A \vert } \sum_i d_i \cdot d(q_i, \delta_{z_i^{\ast}})
    \end{eqnarray}
    holds. Also assume
    $ \vert \theta(q) - \theta^{\ast} \vert  \leq c \frac{1}{ \vert A \vert } 
    \sum_i d_i \cdot d(q_i, \delta_{z_i^{\ast}})$ 
    holds for any $q$ and some constant $c$.
    \item[C6'][\textbf{Network Degree}]
    Suppose $\{1,\ldots,n\}$ can be partitioned into 
    two sets $\mathcal N_u$ and $\mathcal N_b$. 
    $\mathcal N_u$ is the set of nodes with degree 
    larger than $d_n$ and $N_b$ is the set of nodes with bounded degree. 
    $d_n = m^{r_d} (r_d > 0)$.
    
    Let $d_{i, \mathcal N_b}$ be the number of nodes within 
    $\mathcal N_b$ that individual $i$ connects to. 
    We assume $d_{i, \mathcal N_b}$ is bounded for all $i$.  
    
    In addition, the cardinality of $N_b$ satisfies $ \vert \mathcal N_b \vert  /  \vert A \vert  = o(1)$.
\end{itemize}

\begin{lemma}\label{prob:diff:extend}
With probability $1 -  \exp\{ - C d_n \}$, it holds that 
\begin{eqnarray}
\sum_{z: z_i \neq z_i^{\ast} } L_i(\theta  \vert  z) = 
L_i(\theta  \vert  z^{\ast}) \cdot O(\exp\{- c_0 d_n\})
\end{eqnarray}
for any $i \in \mathcal N_{u}$ and any $\theta \in B(\theta^{\ast}, \delta)$ 
for some constants $c_1$ and $\delta$. 
Here $L_i(\theta  \vert  z) := \exp\{ \vert A \vert  l_i (\theta  \vert  z)\}$ and 
$$
l_i(\theta\vert  z):= \frac{1}{ \vert A \vert }(\sum_{(i,j) \in A} l_{ij}(\theta \vert z_i, z_j) +
 \sum_{(j,i) \in A} l_{ji}(\theta \vert z_j, z_i)).
 $$
\end{lemma}
\textbf{Proof of Lemma \ref{prob:diff:extend}}
Similar to the proof of Lemma \ref{prob:diff}, we can prove that
\begin{eqnarray}\label{target:lemma3}
l_i(\theta  \vert  z_{\mathcal N}, z_{-\mathcal N}^{\ast}) 
\leq l_i(\theta  \vert  z^{\ast}) - 
c/2 \frac{\sum_{j \in \mathcal N} d_j}{ \vert A \vert } 
\end{eqnarray}
holds for any fixed $z_{\mathcal N}$ with probability at least 
$1 - \exp\{- C(\sum_{j \in \mathcal N} d_j)\}$. 
Then, we can compute
\begin{eqnarray}\label{eqn:compute}
& &P(l_i(\theta  \vert  z_{\mathcal N}, 
z_{-\mathcal N}^{\ast}) \geq l_i(\theta  \vert  z^{\ast}) - 
c/2 \frac{\sum_{j \in \mathcal N} d_j}{ \vert A \vert } \nonumber \\
& & ~ \textrm{for some} ~  z_{N}) \nonumber \\
&\leq& \sum_{z_{-i} \neq z_{-i}^{\ast}} 
\exp\{ - C (d_i + \sum_{j: z_j \neq z_j^{\ast}} d_j) \} \nonumber \\
&\leq& \sum_{m_0=1}^n \sum_{ \vert z_{\mathcal N_u} - z_{\mathcal N_u}^{\ast} \vert _0 = m_0} 
\exp\{ - C d_n m_0 \} \nonumber \\
& & + K^{ \vert d_{i,\mathcal N_b} \vert } \exp\{-C d_i\} \label{decomp} \\
&\leq& \exp\{ - c_0 d_n\}  \nonumber,
\end{eqnarray}
by adjusting the constants.
\eqref{decomp} uses the fact that $l_i(\theta \vert z)$ only depends 
on a finite number of nodes in $\mathcal N_{b}$.
This completes the proof.

We define the estimator of latent class membership 
as $\hat z_i := \arg\max_{z} q_i^{(M)}(z)$.
The following result says that we can consistently 
estimate the latent class membership of those individuals with large degrees exponentially fast.

\begin{theorem}\label{thm:dense}
Under Conditions C1, C4, C5, C7 and C2', C3', C6', with 
probability $1 - M \exp\{-C d_n\}$, we have 
$q_i^{(m)}(z_i = z_i^{\ast}) \geq 1 - C \exp\{\-c_1 m d_n\}$ for all $i \in \mathcal N_u$ and $m = 1, \ldots, M$.
Especially when $N_u = \{1, \ldots, n\}$, we can recover true labels of all nodes.
\end{theorem}
\textbf{Proof of Theorem \ref{thm:dense}}
To prove this, we only need to show that 
$q_i^{(m)}(z_i^{\ast}) \geq 1 - C \exp\{-c_1 m d_n\}$
for all $i \in \mathcal N_u$ with probability $1 - \exp\{-Cd_n\}$. ($m = 1,2,\ldots$ and $c_1$ is some small constant.)
Without loss of generality, we can assume 
$\theta^{(m)}$ is always in $\mathcal B(\theta^{\ast}, \delta)$. 
(The proof of this argument is almost same as that in the proof of Theorem \ref{regret:thm}.)

Take any $i \in \mathcal N_u$. 
We first prove that 
$q_i^{(1)}(z_i^{\ast}) \geq 1 - C \exp\{-c_1 d_n\}$ for $i \in \mathcal N_u$.
This is true by applying Condition C7. 
In the following, we prove the result by induction.

According to Lemma \ref{prob:diff:extend} and Condition C2', we have that 
$ \vert A \vert  l_{i}(\theta^{(m_1)} \vert z) \leq  \vert A \vert  l_{i}(\theta^{(m_1)} \vert z^{\ast}) -
c d_n$ for any $z$ with $z_i \neq z_i^{\ast}$ with probability $1 - \exp\{-C d_n\}$. 
This implies that 
\begin{eqnarray*}
\mathbb E_{q_{-z_i}^{(m_1)}}  \vert A \vert  l_{m_1}(\theta^{(m_1)} \vert z_i , z_{-i})
\leq  \vert A \vert  l_{m_1}(\theta^{(m_1)} \vert z^{\ast}) - c d_n 
\end{eqnarray*}
for any $z_i \neq z_i^{\ast}$ and 
\begin{eqnarray*}
& & \mathbb E_{q_{-z_i}^{(m_1)}}  \vert A \vert  l_{m_1}(\theta^{(m_1)} \vert z_i^{\ast}, z_{-i}) \\
&\geq&  \vert A \vert  l_{n_1}(\theta^{(m_1)} \vert z^{\ast}) - C  \vert A \vert  \exp\{-c_1 m_1 d_n\} - 
C  \vert d_{i, N_b} \vert . 
\end{eqnarray*}
Combining these two facts, we have that 
\begin{eqnarray*}
& &E_{q_{-z_i}^{(m_1)}}  \vert A \vert  l_{m_1}(\theta^{(m_1)} \vert z_i , z_{-i}) \\
&<& \mathbb E_{q_{-z_i}^{(m_1)}}  \vert A \vert  
l_{m_1}(\theta^{(m_1)} \vert z_i^{\ast}, z_{-i}) - c_2 d_n
\end{eqnarray*}
holds for any $z_i \neq z_i^{\ast}$ and some adjusted constant $c_2 > c_1$.

By the recursive formula 
$$
S^{(m_1 + 1)}(z_i) = S^{(m_1)}(z_i) 
\exp\{\mathbb E_{q_{-z_i}^{(m_1)}}  \vert A \vert  l_m(\theta^{(m_1)}  \vert  z_i, z_{-i}) \},
$$
we then have 
$$\sum_{z_i \neq z_i^{\ast}} S^{(m_1 + 1)}(z_i) = 
S^{(m_1 + 1)}(z_i^{\ast}) O(\exp\{- (c_1 m_1 + c_2) d_n\}),$$ 
which indicates that 
\begin{eqnarray}
q^{(m_1 + 1)}(z_i^{\ast}) &\geq& 1 - \exp\{- (c_1 m_1 + c_2) d_n\} \nonumber \\
&\geq& 1 - \exp\{- (m_1 + 1) c_1 d_n \}.
\end{eqnarray}
Hence, we complete the proof by induction.
\label{app:theorem}

\noindent Lastly, it is easy to see that proof of Theorem 
\ref{thm:q:community} is the special case of Theorem \ref{thm:dense}.

\end{appendices}

\end{document}